\documentclass[10pt, a4paper, twocolumn]{article} % 10pt font size (11 and 12 also possible), A4 paper (letterpaper for US letter) and two column layout (remove for one column)
\usepackage[english]{babel} % English language hyphenation
\usepackage{microtype} % Better typography
\usepackage{amsmath,amsfonts} % Math packages for equations
\usepackage[svgnames]{xcolor} % Enabling colors by their 'svgnames'
\usepackage{graphicx} % Required for adding images
\usepackage{enumitem} % Required for customising lists
\setlist{noitemsep} % Remove spacing between bullet/numbered list elements
\usepackage{sectsty} % Enables custom section titles
\allsectionsfont{\usefont{OT1}{phv}{b}{n}} % Change the font of all section commands (Helvetica)

\usepackage{geometry} % Required for adjusting page dimensions

\usepackage[T1]{fontenc} % Output font encoding for international characters
\usepackage[utf8]{inputenc} % Required for inputting international characters

\usepackage{XCharter} % Use the XCharter font

\usepackage{fancyhdr} % Needed to define custom headers/footers
\fancypagestyle{firstpage}{ % Page style for the first page with the title
	\fancyhf{}
}

\usepackage{titling} % Allows custom title configuration

\newcommand{\HorRule}{\rule{\linewidth}{1pt}} % Defines the gold horizontal rule around the title

\pretitle{
	\vspace{-60pt} % Move the entire title section up
	\HorRule\vspace{5pt} % Horizontal rule before the title
	\fontsize{16}{26}\usefont{OT1}{phv}{b}{n}\selectfont % Helvetica
}

\posttitle{\par\vskip 3pt} % Whitespace under the title

\preauthor{} % Anything that will appear before \author is printed

\postauthor{ % Anything that will appear after \author is printed
	\vspace{5pt} % Space before the rule
	\par\HorRule % Horizontal rule after the title
	\vspace{0pt} % Space after the title section
}

\usepackage[numbers]{natbib}
\usepackage{lineno}
\modulolinenumbers[1]
\usepackage{amsmath}
\usepackage{cmll}
\usepackage{amssymb}
\usepackage{latexsym}
\usepackage{wasysym}
\usepackage{amsthm}
\usepackage{hyperref}
\usepackage{doi}
\usepackage{stmaryrd}

\newcommand{\prela}{ pre\-lambda }

\def\eg{e.g., }
\def\cf{cf.  }
\def\ie{i.e., }

\newcommand{\NF}{\reflectbox{\textnormal{F}}}
\newcommand{\NB}{\reflectbox{\textnormal{B}}}

\newcommand{\vv}{\langle}
\newcommand{\ww}{\rangle}

\newcommand{\Var}{\mathbb{V}}
\newcommand{\Lam}{\mathbb{T}}

\newcommand{\sR}{r}

\newcommand{\sT}{t}
\newcommand{\sL}{\ell}
\newcommand{\sA}{a}
\newcommand{\sS}{s}

\newcommand{\tx}{\mathsf{x}}
\newcommand{\tv}{\mathsf{v}}
\newcommand{\ty}{\mathsf{y}}
\newcommand{\tz}{\mathsf{z}}

\newcommand{\tF}{\mathfrak{F}}

\newcommand{\tw}{\mathsf{w}}

\newcommand{\pl}{\llbracket}
\newcommand{\pr}{\rrbracket}

\newcommand{\ta}{\mathsf{A}}
\newcommand{\tb}{\mathsf{B}}
\newcommand{\tc}{\mathsf{C}}
\newcommand{\td}{\mathsf{D}}
\newcommand{\te}{\mathsf{E}}
\newcommand{\tf}{\mathsf{F}}
\newcommand{\tg}{\mathsf{G}}

\newcommand{\tm}{\mathsf{M}}
\newcommand{\tn}{\mathsf{N}}

\newcommand{\eqdef}{\stackrel{\text{\tiny \textnormal{DEF}}}{=}}

\newcommand{\cv}{ \ \! \sim \ \! }

\newcommand{\nocv}{\ \! \ensuremath{\not\,\!\sim_{\cI}} \ \! }

\newcommand{\mo}{\,\! \ensuremath{\sim_{\cI}} \,\! }

\newcommand{\st}{\taleche}

\newcommand{\REL}{ \ \! \sim \ \! }
\newcommand{\nREL}{ \  \not\!\!\!\,{\sim} \  }

\newcommand{\lb}{\{ \,\!}
\newcommand{\rb}{ \,\! \}}

\newcommand{\bE}{I}

\newcommand{\ind}[1]{\mathord{\sim\!\!(#1)}}

\newcommand{\taleche}{ \ \, | \ \, }

\newcommand{\cI}{I}

\newcommand{\RVn}{R\'ev\'esz}
\newcommand{\RV}{R\'ev\'esz }

\newcommand{\qedef}{\hfill{{\ensuremath{\triangle}}}}

\newcommand{\inconsi}{\ensuremath{\sim \mathop{=}  \Lam^2}}
\newcommand{\consi}{\ensuremath{\sim \mathop{\neq}  \Lam^2}}

\newtheorem{theorem}{Theorem}[section]
\newtheorem{proposition}[theorem]{Proposition}
\newtheorem{lemma}[theorem]{Lemma}
\newtheorem{corollary}[theorem]{Corollary}
\theoremstyle{definition}
\newtheorem{definition}[theorem]{Definition}
\newtheorem{remark}[theorem]{Remark}

\newtheorem{example}[theorem]{Example}

\title{Lambda  Congruences and  Extensionality} % The article title

\author{Michele Basaldella \hfill \texttt{michele.basaldella@gmail.com}} % Authors

\date{} % Add a date here if you would like one to appear underneath the title block, use \today for the current date, leave empty for no date

\begin{document}

\maketitle % Print the title

\thispagestyle{firstpage} % Apply the page style for the first page (no headers and footers)

%----------------------------------------------------------------------------------------
%	ABSTRACT
%----------------------------------------------------------------------------------------

\begin{abstract}

In this work we provide
 alternative formulations of the concepts of lambda theory and extensional
theory without introducing the notion of substitution and the sets of all, free and 
bound variables occurring in a term. 
We  also clarify the actual
role of $\alpha$--renaming and $\eta$--extensionality in the lambda calculus: both of them can be
described as properties of
extensionality for  certain classes
of terms.  
\end{abstract}

%\begin{keyword}
%Lambda calculus, lambda theory, extensionality
%\end{keyword}

\section{Introduction} \label{intro}

 In the early thirties, Church \cite{CHURCH} introduced  a formal system  intended to provide a common foundation for  logic and mathematics.
However, after the discovery of some paradoxes,
 a  proper portion of the original system  was extracted. This part constitutes
 what is known as  the \emph{lambda calculus} today; see
Church \cite{CHU}, Curry and Feys \cite{CF},  Curry, Hindley and Seldin \cite{CHS}, Stenlund \cite{STE}, Stoy \cite{STOY}, Barendregt \cite{BAR},  \RV \cite{RV}, Krivine  \cite{KR},  Hankin \cite{HA},
S{\o}rensen  and Urzyczyn \cite{CHIU}, Hindley and Seldin \cite{HS}, Cardone and Hindley
\cite{CAHI} and Seldin \cite{SELDIN}
for good expositions.

In a nutshell,  the lambda calculus
 can be  described as
 a general theory of computable functions: it provides a formalism for dealing with \emph{functions--as--terms}  in the context of the foundations of mathematics: 
In axiomatic set theories,  functions from a given set to itself 
 are  defined by their graphs as --- typically infinite --- sets of ordered
pairs.  By contrast, in  the  lambda  calculus 
functions are simply implemented by means of some ---always finite --- formal expressions called \emph{terms}.
Furthermore, in set theory    it follows from the axioms  that  a function cannot be a member of
its  domain. In the  lambda calculus, 
every term is in  the domain of a given function;
 in particular, every  function can take  itself as input.
 
The most interesting feature of this calculus is that it constitutes %possible to uniformly represent ---  as terms --- several kinds of data  such as natural numbers and  booleans, as well as   functions on
%these   data. 
%%A fundamental result is that
%Furthermore, the computable  functions, in the sense of Herbrand--G\"odel--Kleene, are exactly
%those  functions on natural numbers which 
% can be represented by terms.  
%The  lambda calculus is 
%therefore
 a model of  computation and for this reason this powerful language is  the theoretical core of many modern functional programming languages. 
%Thus, under this perspective, 
% terms  can also be  regarded as \emph{programs}
%of an idealized programming language.

To be more precise, the lambda calculus in itself has no expressive power at all. It becomes an expressive language only  when equipped with a suitable relation which allows us to determine 
 when two terms are \emph{intended} to be the \emph{same}.  While in set theory
% two sets are the same if they contains
% exactly the same members and   
two functions from a given set  to itself  are  the same if  they have identical   input--output behaviour, 
in the  lambda calculus
the situation is rather different;
there is  a  huge collection of  relations   called
 \emph{lambda theories} and each of these relation is
intended to capture
some notion of sameness.

%under our functional perspective.

A little more technically,  lambda theories 
 are  congruences on terms  %(equivalence relations compatible with the rules of term formation) 
 which satisfy the conditions of
\emph{$\alpha$--renaming} and  \emph{$\beta$--rule}.
The  two most    important 
lambda theories are \emph{lambda conversion}  --- the least lambda theory --- and
 \emph{extensional conversion}
 --- the least lambda theory which
 also satisfies the condition of \emph{$\eta$--extensionality}. Here we adopt 
the Church's terminology \cite{CHU}; other names
for lambda conversion and extensional conversion 
are $\alpha\beta$--equality and 
$\alpha\beta\eta$--equality, respectively. 
 %Both of them are strong enough
% to faithfully represent as terms the Herbrand--G\"odel--Kleene computable functions
% between natural numbers.
In order to avoid misunderstanding,
in this paper
we call $\alpha$--renaming, $\beta$--rule
and $\eta$--extensionality
 the conditions that in the literature are usually denoted by
($\alpha$),
($\beta$) and ($\eta$). 
Since  several 
versions of ($\alpha$),
($\beta$) and ($\eta$)
have been proposed,
for the sake of clarity we use our terminology
when we do not refer to a specific formulation.
We reserve the names  ($\alpha$),
($\beta$) and ($\eta$) to the particular conditions that will be introduced in Section  \ref{STA}. 
%they  are   standard and taken from an important work on the subject, Barendregt's dissertation \cite{BARTH}.
%Thus,  by  
%  $\alpha$--renaming, $\beta$--rule
%  and  $\eta$--extensionality
%we respectively mean the condition which allows us to \emph{rename}
%bound variables,  \emph{simplify}
% terms
%and \emph{identify} terms
%which have the same  input--output behaviour.

In general, 
 each lambda theory has to be seen
as a reasonable  relation of  sameness
for terms. 
We refer to Barendregt \cite{BAR},
Berline
\cite{BER,BER06} and Manzonetto and Salibra \cite{MS}
for a survey of results on lambda theories.

%The most typical example
%of lambda theory is the relation
%of $\alpha\beta$--conversion.

In this paper we focus our attention
on the syntactical formalizations
of the notions of  lambda theory and extensional theory --- we call extensional theory
any lambda theory which also satisfy
$\eta$--extensionality.
The usual way to proceed is to introduce
some auxiliary notions first: 
\begin{itemize}
\item the sets
 of \emph{all}, \emph{free} and \emph{bound} variables occurring in a term, 
\item a meta--theoretic operation of \emph{substitution}. 
%\footnote{Many reader would say
 %``$\alpha$--equivalence'', rather than out }  
\end{itemize}
Henceforth, we refer to these
auxiliary notions as \emph{ancillary concepts}. Note that they usually occur,
 explicitly or implicitly, in the traditional formulations of the conditions of 
 $\alpha$--renaming, $\beta$--rule
  and  $\eta$--extensionality.
One of the aims of this paper is to show that
\begin{quote}
 \emph{in order to syntactically define lambda and extensional theories it is not necessary
to introduce any ancillary concept}.\end{quote}
%Furthermore, 
%we  show that
%$\alpha$--renaming is not necessary
%to  axiomatize the notion of extensional theory.

Regarding our motivations, we mention that it is well--known that
 the usual definition of substitution is
 intrinsically difficult and error--prone. The same problems also arise
 in logics with quantifiers
 and more generally  in languages with binding operators.
 As for the lambda calculus,
  we believe that
 the actual complications do not lie in the operation of substitution  itself  but
 rather  in the formulation
of
the  $\beta$--rule.
In fact, it is not so well--known that it is possible to consider
a  simple definition
of substitution at the price
of having some restrictions
in the $\beta$--rule.
This is the approach of Barendregt's
dissertation \cite{BARTH},
which is actually a more elegant variant of the original approach
developed by Church \cite{CHU}.
 By contrast, following Curry and Feys \cite{CF},   the vast majority of the presentations of the lambda calculus prefers to
consider a complicated notion
 of substitution which allows a simple formulation of the $\beta$--rule. In Section \ref{STA} we shall compare the two options in more detail.
 But in either case, we believe it is definitely more
 convenient to \emph{not consider substitution at all} and replace the $\beta$--rule with the
 more elementary conditions  we shall introduce later on.
  
 Historically, these complications were well--known to the fathers of modern formal logic. 
  In fact,  approximately in same period the lambda calculus was conceived,  a  way to circumvent the problems caused by substitution was proposed by Sh\"onfinkel \cite{SCHO2} and Curry \cite{CURRY29,CF,CHS} under the name of \emph{combinatory logic}.
% see also the monograph of Curry and his co--workers
 %\cite{CF,CHS}. 
  However, we think that in such a framework  the elimination of concepts is too drastic:
 the ancillary concepts are eliminated  by means of  the rejection  of
 the notion of (bound) variable.
As a negative side effect,  in that setting the   ``functional intuitions''
coming from the
 lambda notation are completely lost.
This is not an accident, since this elimination of variables was one of the objectives of 
Sh\"onfinkel (see  Cardone and Hindley
\cite{CAHI} for a survey of Sh\"onfinkel's work).
By contrast, in this paper
we want to maintain the syntax of the lambda calculus unaltered and perform a more precise surgery 
which allows us to eliminate only the ancillary concepts.

Regarding the sets
 of all, free and bound variables occurring in a term, we see these concepts as mere ``syntactic bureaucracy'' without  any real  mathematical  substance. 
As we shall see, these sets can be eliminated from the picture with almost  no
 effort and hence in this way  more elegant  formulations of the notions of lambda theory and extensional theory can be obtained.

In addition to the theoretical and pedagogical interest of our results, we also see a more concrete advantage as we now explain. In the model theory
of the lambda calculus, in order to   prove
 that a given mathematical structure is a model of, say,  lambda  conversion, all we need to do is to check that the binary relation 
obtained by relating terms which have the same interpretation
  in the given structure --- the  relation usually  called   \emph{theory of a model} --- is a
  lambda  theory. In equivalent words, we
  have to prove a \emph{soundness theorem}.
  To show this theorem, a typical approach is
  to prove the \emph{substitution lemma}, see Wadsworth \cite{WAD}
and Meyer \cite{MEYER}, as well as Stoy
 \cite[pp. 161--166]{STOY} for a detailed proof of
this lemma. 
   We believe that
a direct verification of the substitution--free conditions
which define (our formulation of)
lambda  theories is not only easier 
than the  lemma, 
but also  more convenient inasmuch as  
 the proof --- and the exportability  --- of the lemma
 largely depends
on the \emph{specific definition of substitution}. (In particular,  the substitution lemma  cannot
be simply invoked by  authors who use a different notion of substitution. Unfortunately, we also believe that this delicate point has been overlooked by many authors.)
To convince the reader, a direct verification of some of our substitution--free conditions
will be given in Section \ref{modelsection}.

In order to avoid terminological 
misunderstandings, in the rest of the paper 
the relations of sameness on terms defined 
by using \emph{our} formalizations
 are 
 called \emph{lambda} and
 \emph{extensional congruences},
 while we exclusively reserve  the terminology \emph{lambda} and
 \emph{extensional theory}  for
the \emph{traditional}  axiomatizations of these relations.
Then,  one of our main results can be stated as follows:
\begin{quote}
 \emph{lambda congruences
and lambda theories coincide, 
as well as extensional
congruences and extensional theories}. \end{quote}

Another aspect we want to analyze
in this paper is the precise role
of $\alpha$--renaming and $\eta$--extensionality in the lambda calculus.
As for the latter, it is well--known 
that this condition   allows us to consider
 \emph{terms which have
the same input--output behaviour
as the same}. 
Thus, expressed in this form $\eta$--extensionality  has a very intuitive and satisfactory mathematical significance: it reproduces  in the lambda calculus
a form of extensionality
familiar from ordinary set theory.  
As for the former, the situation is
not so clear. To the best of our knowledge, the most intuitive
informal description of this concept
is this: \emph{terms which only differs in the names of bound variables should be considered as the same}. We think that this explanation is  too syntactical --- it even mentions an ancillary concept! --- and we do not see
any clear mathematical content 
in it. One of the aim of this paper
is to clarify the meaning of $\alpha$--renaming in intuitive and mathematically reasonable terms.
It turns out that  
$\alpha$--renaming admits
the following  simple and satisfactory equivalent description:
\begin{quote}
 \emph{abstractions which have
the same input--output behaviour
should be considered 
as the same}, \end{quote}
as we shall show in Section
\ref{lambdacong}.
Here ``abstractions'' are, roughly speaking, terms which begin with
a $\lambda$. 
Thus,
$\alpha$--renaming
can be regarded as a form of extensionality for abstractions.
We believe that  when expressed in this form the real mathematical significance
of this condition  emerges.

To obtain the aforementioned results we 
have to consider a new \emph{factorization} of the conditions
which define lambda and extensional theories, in the  sense we now explain. Consider, for instance, 
lambda theories.
Following Church \cite{CHU},  a way of introducing this concept is to put $\alpha$--renaming
and the $\beta$--rule
together in the definition of these theories, like in the original formulation of lambda conversion.
Following Barendregt \cite{BAR},
another
approach is to define an auxiliary equivalence relation on terms
called
 \emph{alpha conversion}.
 Terms which belongs to
the same equivalence class are
then identified. After this identification, 
substitution and the $\beta$--rule are finally introduced.
We refer to Crole \cite{CROLE}
and the references therein for more on alpha conversion.

Our factorization is completely different and, to the best of our knowledge,  it  seems to be new. Roughly speaking,  we obtain it by proceeding as follows.
First,  we  only consider
our substitution--free version of the $\beta$--rule.
We obtain a family of relations that
we call \emph{prelambda congruences}.
These relations  need
to satisfy neither $\alpha$--renaming
nor $\eta$--extensionality; in fact,
we shall show
is 
that there exists a prelambda congruence
in which both conditions  do not hold.
Only at this stage
we add $\alpha$--renaming to prelambda congruences.
We then obtain our version of
lambda theories:
lambda congruences.
We shall prove that
\begin{quote}
\emph{$\alpha$--renaming
and the property of extensionality
for abstractions mentioned above
are equivalent in every prelambda congruence}.\end{quote}
Also, we  add 
$\eta$--extensionality to prelambda congruences and get our version of
extensional theories:
extensional  congruences.
Similarly, we shall prove that
\begin{quote}\emph{$\eta$--extensionality
and the property of extensionality
(for all terms) mentioned above
are equivalent in every prelambda congruence}. \end{quote}
Notice that these equivalences are 
significant  precisely because
neither $\alpha$--renaming nor
$\eta$--extensionality
is valid in 
 every prelambda congruence.
%Note, too,   that we do not consider  
%$\alpha$--renaming as an ingredient 
%to build extensional  congruences.
%The reason is that $\alpha$--renaming 
%is 
%\emph{implied} by the other conditions. 
%This fact is not so  well--known but   it  has been  observed by 
%Do\v{s}en and Petri\'c
%\cite{DP}  in the simply typed lambda calculus.

Regarding our methodology, 
our analysis of $\alpha$--renaming and $\eta$--extensionality
in the lambda calculus
is somehow inspired by the usual treatment of   classical logic from a constructive point of view.
In that context,  intuitionistic logic
 is first introduced and developed.  Only at a second stage,  the law of excluded middle is motivated,
 discussed and added, and some equivalent formulations
 % (which are theorems of intuitionistic logic)  
 are established --- such as double negation elimination.
 Equivalents of excluded middle
are of interest precisely because this law 
 is not a theorem of intuitionistic logic.
 Of course, 
 $\alpha$--renaming and $\eta$--extensionality are by no means ``controversial'' in the lambda calculus
 as   excluded middle is in constructive mathematics, but we do believe they deserve 
 a better understanding
 as they are too often taken  ``for granted'' in the literature, especially $\alpha$--renaming.
 One of the aims of this paper is to 
  provide  a more intuitive and clear comprehension  of these conditions.
  %Regarding our methodology, 
%our analysis of $\alpha$--renaming and $\eta$--extensionality
%in the lambda calculus
%is somehow inspired by the usual treatment of   the axiom of choice  
%in axiomatic set theories. For the case of
%$ZFC$ (Zermelo--Fraenkel set theory with Choice), usually  the choice--free set theory
%$ZF$ is first introduced and developed.  Only in a second stage,  the axiom of choice is motivated,
% discussed and added, and some equivalent formulations (which are theorems of the choice--free part $ZF$)  are established --- such as the well--ordering theorem.
% Clearly, equivalent
% formulations of the axioms of choice
%are of interest because
%the axiom of choice is not a theorem of $ZF$ (assuming the consistency of $ZF$)  
% Of course, 
% $\alpha$--renaming and $\eta$--extensionality are by no means ``controversial'' in the lambda calculus
% as   the axiom of choice in set theory, but we do believe they deserve 
% a better understanding
% as they are too often taken  ``for granted'' in the literature, especially $\alpha$--renaming.
% One of the aims of this paper is to
%  provide  a more intuitive and clear comprehension  of these conditions. 
To sum up, the result of our analysis  is that \begin{quote} \emph{in   every prelambda congruence both $\alpha$--renaming and $\eta$--extensionality
can be  both described in a natural way as
properties of extensionality}.\end{quote}
This completes the description of this paper. 

We now discuss some related work.
As far we know, there are  three lines of apparently 
independent research which  implicitly or explicitly try to eliminate the
ancillary concepts without modifying the syntax of the lambda calculus.

In the sixties, following some ideas of Church \cite{church1940},  Henkin  introduced
a theory of propositional types  based on the simply typed lambda calculus \cite{HENKIN}.
In order to develop a deductive
system for his theory,  Henkin 
decomposed the typed $\beta$--rule
 in several substitution--free
clauses which are strikingly similar to the conditions that we use 
in this paper. The reason  for this decomposition  was the simplification of the proof of the soundness theorem
of his deductive system with respect to the set--theoretic semantics of his type theory.
Again, here we see  strong similarities
between his and our  motivations. 
Following Henkin,   Andrews
employed similar clauses in his work
on type theory \cite{ANDREWS1965, AndrewsBook}.
However,  some ancillary concepts
are  present in all these works.

In the  early eighties,  \RV \cite{RV1,RV85,RV} proposed a
substitution--free formalization of lambda conversion. 
His and our motivations  are rather similar, 
the main difference is that
he was mainly interested in developing concrete computer implementation
while 
we also think that  substitution--free formalizations  are also useful to simplify  some practical aspects of the model theory of the lambda calculus --- namely,  the elimination of the substitution lemma.
   \RVn's conditions are very similar
   to (untyped versions) of Henkin's
   ones and in particular 
 some ancillary concepts
are still present in his formalization.
From an historical perspective, the   main   novelty of his work 
   was the first substitution--free formulation of $\alpha$--renaming.

Some years later, with the aim of
 providing  a general algebraic setting for  the
lambda calculus,
    Pigozzi and Salibra \cite{PSLM,PS,PS98}
introduced the theory of \emph{lambda abstraction algebras}. Despite
the aims of this line of work
are completely different from ours,
from a purely syntactical standpoint the improvement emerging from the introduction of these algebras  is the complete elimination the ancillary concepts. In fact, the conditions
defining 
 lambda abstraction algebras  are very similar to those we introduce in this paper, as we shall see.

We finally point out  that
the conditions we use in this paper
to define lambda and extensional
congruences do not form just a  selection of clauses taken from these works.
On the contrary,  the crucial conditions ($\beta_5$), 
($\alpha_e$) and ($\eta_e$) that we use in this paper
are  simplified   --- for our purposes --- reworked versions of some similar clauses present in the aforementioned literature.

The paper is organized as follows.
In Section \ref{lambda sec}, we shall recall the syntax of the terms of the  lambda calculus and introduce some
notation and terminology.
In Section \ref{PBC}
 we shall introduce
the concept of prelambda congruence
and in  Section \ref{indsec}
we shall study
the basic properties of
this notion. 
In Section \ref{lambdacong}
we shall discuss the real --- for us --- reason
why $\alpha$--renaming is 
so important in the lambda calculus
and introduce the concept
of lambda congruence.
We shall also prove that in every prelambda congruence
$\alpha$--renaming is equivalent
to a suitable principle of extensionality. 
Similarly, 
 in Section \ref{ACF} we
 shall  define the notion
 of extensional congruence
 and prove that 
 in every prelambda congruence
 $\eta$--extensionality  is 
 equivalent
to a  principle of extensionality. 
In Section 
 \ref{STA} we shall recall the  concepts of  lambda and extensional theory and in Section \ref{EQU}
we shall prove the equivalence between lambda  congruences and theories as well as the one between 
extensional congruences and theories.
%As a consequence of this result, lambda and extensional theories can be both formalized 
%without using any ancillary concept.
In
Section
\ref{modelsection} 
we shall provide a construction of a model whose theory forms
a prelambda congruence which
do not satisfy a simple instance of $\alpha$--renaming.
% In particular, we shall prove that
%prelambda congruences which do not
%satisfy $\alpha$--renaming do exist.
Finally, in Section \ref{concl}
we shall conclude the paper.

%In order to save space, in the rest of the paper we
%simply say conversion and lambda calculus  for  lambda conversion and untyped lambda calculus,
%respectively. 

\section{The  Lambda   Calculus} \label{lambda sec}

In this section we recall the syntax of the terms of the  lambda calculus
and introduce some terminology.

\begin{definition}[Variables]
Let $\Var$ be an infinite set. We call its member \textbf{variables} and we denote them
by
$\tx$, $\ty$, $\tz$,\ldots. 
\qedef
\end{definition}

Some comments concerning the set $\Var$ are given at the end of this section.

%\begin{definition}[Constant]
%Let $\const$ be an arbitrary but fixed (possibly empty) set such that
%$\Var \cap \const = \emptyset$.
%We refer to its elements as  \textbf{constants} and we denote them 
%by $\ccc$, $\ddd$, $\eee$,\ldots. \qedef
%\end{definition}

The basic elements of the lambda calculus are called \emph{terms}.  
They are special words built from the infinite alphabet consisting of all
  variables together with  the  following auxiliary symbols: 
$
\lambda$ (\emph{lambda--abstractor}), 
 $[$   (\emph{left lambda--bracket})
 and 
 $]$ (\emph{right lambda--bracket}). 
 In this paper, we write words  by juxtaposition of symbols;
 in particular, we do not notationally distinguish
 between members of the alphabet and unary words.

\begin{definition}[Term] \label{term}
The \textbf{terms} of the  lambda calculus
are the elements of the set $\Lam$ which is inductively defined as follows:
\begin{itemize}
\item[(T$_1$)]   $\tx \in \Var$ implies  $\tx \in \Lam$;
%\item[(T$_2$)] \  $\ccc \in \const$ implies  $\ccc \in \Lam$;
\item[(T$_2$)]   $\ta\in \Lam$ and $\tb \in\Lam$ imply $[\ta\tb] \in \Lam$;
\item[(T$_3$)]   $ \tx \in \Var$ and   $\ta\in \Lam$ imply $[\lambda \tx\ta] \in \Lam$.
\end{itemize}
\noindent In the sequel, we use $\ta$, $\tb$, $\tc$, \ldots to denote arbitrary terms.
\qedef
\end{definition}

We now introduce some useful
terminology and notation.

As is standard, two terms $\ta$ and $\tb$ are said to be \emph{equal}, in symbols $\ta = \tb$, if and only if  $\ta$ and $\tb$ are exactly  the same  word. In particular, we have 
$\tx = \ty$  as elements of $\Lam$ if and only if $\tx = \ty$  
as elements of $\Var$. Our notion of equality corresponds to  the relation which is  often called \emph{syntactical equality} in the literature. %Analogously,
 %one has $\ccc = \ddd$ in $\Lam$
 %if and only if 
 %$\ccc = \ddd$ holds in 
 %$\const$. 

In order 
to save space,
% we henceforth   remove the outermost lambda--brackets surrounding a term (if any); in other words, 
we always write   $\ta\tb$ and $\lambda \tx\ta$ for $[\ta\tb]$
and  $[\lambda \tx\ta]$, respectively. Henceforth, we also refer to a term of the form  $\ta\tb$ as an \emph{application}  and to a term of the form $\lambda \tx \ta$ as  an 
 \emph{abstraction}.

As usual, we call  a subset of $\Lam^2$  (the Cartesian product  of $\Lam$ with itself) a \emph{binary
relation on terms} and
we use the symbol $\cv$ to denote  binary relations on terms. We also 
 write $\ta \cv\tb$ for $(\ta,\tb) \in \cv$ and $\ta \, \! \!\! \nREL \, \tb$
for $(\ta,\tb) \notin  \cv$.
Furthermore,  if we have $\ta \cv \tb$ and  $\tb \cv \tc$, then sometimes we simply write $\ta \cv \tb \cv \tc$ to express this fact.
Analogously, we may write
$\ta \cv \tb \cv \tc \cv \td$ in case we  have
 $\ta \cv \tb$,  $\tb \cv \tc$
 and $\tc \cv \td$,  
and so on.
%in the case of four or more  terms.

An important class of binary relations 
on terms is that of \emph{congruences},
that we now formally define and discuss.

\begin{definition}[Congruence, structural conditions] \label{congr}
Let $\REL$ be a binary relation
on terms. We say that $\REL$
is a \textbf{congruence} if
it satisfies the following conditions:
 
\begin{itemize}
\item[($\sR$)]  $\ta \cv \ta$; 

    \item[($\sS$)]  $\ta \cv \tb $ implies $\tb \cv \ta$; 
   
\item[($\sT$)]  $\ta \cv \tb $ and $\tb \cv \tc$ imply $\ta \cv \tc$; 
  
   \item[($\sL$)]   $\ta \cv \tb $ implies $\lambda \tx\ta \cv \lambda \tx\tb$; 
   
   \item[($\sA$)]    $\ta \cv \tb $ and $\tc \cv \td$ imply $\ta\tc \cv \tb\td$;   

 \end{itemize}

\noindent where $\ta$, $\tb$, $\tc$ and $\td$ are arbitrary terms, and 
   $\tx$ is an arbitrary variable. 
   We collectively refer to conditions ($\sR$) to ($\sA$) above as \textbf{structural conditions}.
   %   The set of all  congruences is denoted
%   by $\cC$.
   \qedef
 \end{definition}

 Conditions  ($\sR$), ($\sS$)  and ($\sT$)
  ensure that any congruence is  an equivalence relation between terms, while conditions ($\sL$) and ($\sA$) express the fact that
  any congruence is a relation compatible with the operations of term formation (T$_2$) and (T$_3$) of Definition \ref{term}.
   In particular, a congruence defined as above is a congruence
 in the usual algebraic sense.
 In other words, structural conditions are nothing more than  a formalization of
 the  algebraic principle usually called  \emph{replacement of equals
 by equals}. This principle  --- which is  ubiquitous in algebra --- is often used   in an  \emph{implicit}
 and \emph{tacit} fashion. 
In this paper,
for the sake of clarity ---  but in contrast  to the algebraic
 tradition --- we shall \emph{explicitly} mention 
 when and where structural conditions
 are used in  formal proofs of our results.
 
 While important,
congruences alone cannot  be considered as  candidates for a good formalization of the notion  of sameness for terms: in fact, in order to develop
a  
theory of functions using the lambda calculus
more conditions are needed.
This situation  is similar to what happens for the traditional presentations of sequent calculi for classical logic in proof theory: there, in addition to
structural rules,  we also need
 \emph{logical rules} to be able to derive classical tautologies 
 in a syntactical way. In our context,   ``structural rules''
 are, unsurprisingly,  what we are  calling structural conditions and
 the  ``logical rules'' specific to this work that we shall introduce and discuss
 in detail in later sections are
($\beta_1$), ($\beta_2$), ($\beta_3$),  ($\beta_4$),  ($\beta_5$), 
($\alpha_e$) and
($\eta_e$).
%, as well as the  well--known conditions ($\alpha$), ($\beta$) and ($\eta$). 

We now give an informal presentation of our logical rules.
In the lambda calculus,  
the functional  intuition on terms is
this:   
an abstraction of the form
$\lambda \tx \ta$ should be thought as a function depending
on the variable $\tx$.
If, for a moment, we think of $\ta$ 
as the polynomial
 $2x +y$,  then $\lambda \tx \ta$
represents the  function given by
 $f(x) = 2x +y$.
Also,  an application of 
 the form $[\lambda \tx \ta]\td$
 should be intuitively regarded
 as the function $\lambda \tx \ta$ applied to the specific   input $\td$. 
 Thus, if $\td$ represent the number 7, then 
 $[\lambda \tx \ta]\td$
 represents $f(7)$. It follows from elementary rules of algebra that  this expression can be simplified to 
$14+y$.
Now, in our formalizations of lambda
and extensional theories,
the role of conditions ($\beta_1$), ($\beta_2$), ($\beta_3$),  ($\beta_4$) and ($\beta_5$)  is in some sense algebraic: 
they serve to simplify 
expressions like  $[\lambda \tx \ta]\td$  by performing some symbolic manipulations.  
The remaining conditions,
namely ($\alpha_e$) and  
($\eta_e$),  have a different task;
 in this work, 
both of them
are  seen as conditions
of \emph{extensionality} in the sense that
they allows us to infer that two terms
(abstractions, in case of ($\alpha_e$)) 
are the same (\ie related by $\cv$) if
they have the same input--output behaviour,  
as we shall see in detail later on.

As promised, some comments on the set of variables $\Var$
are now in order.

Firstly, we make a rather trivial --- but 
essential --- remark:  we observe that the property
of being infinite is crucial.
To see this, let us call  a set of variables $S$
  \emph{cofinite} if its complement with respect to $\Var$ is finite, in other words
if  $S= \Var \setminus F $ for some
finite set of variables $F$.
Now, the reason for the infinity of $\Var$ is that 
in order to develop even the most elementary part of the lambda calculus we need
the following facts to hold:

\begin{itemize}
\item If $S$ and $T$ are cofinite, then so is
$S \cap T$;
 \item
 if $S$ is  cofinite  and $S \subseteq T$, then $T$ is  cofinite;

  \item every cofinite set 
is non--empty.
 \end{itemize}
 
These properties  immediately follow from the definition of cofinite set 
 and from the infinity of $\Var$. In the sequel, in order to make the exposition less pedantic,
 we shall  use these simple properties
 without explicitly mentioning
 them. 
Our choice is perfectly in line with what it is  done the literature, where we often find  statements like
``let $\tz$ be a variable
which occurs neither free in $\ta$ nor bound in $\tb$'' without the
 explicit proof of the existence of this variable (which simply goes as follows: First,  observe that the sets of variables which  occur  free or bound  in a  term are finite, so that their complements
 are cofinite.
 Then, the result follows,  as  the intersection  of cofinite sets  is  cofinite  and cofinite sets cannot be empty).

%\begin{proof} (i) \
%For each $i \in \lb1,\ldots,n \rb$,
%we have $A_i = \Var \setminus F_i$
%for some finite set of variables $F_i$.
%Now, it follows that $A_1 \cap \cdots \cap A_n
%= \big(\Var \setminus F_1\big) \cap \ldots \cap
%\big(\Var \setminus F_n\big) = \Var \setminus \big( F_1 \cup \cdots \cup F_n\big)$.  Now, since
% $F_1 \cup \cdots \cup F_n$ is finite, we conclude that $A_1 \cap \cdots \cap A_n$
% is a cofinite set. 
% 
% (ii) Suppose that $A = \Var \setminus G$
% for some finite set of variables $G$.
% Then,  $ B = \Var \setminus H$, for some
% (necessarily finite) set $H$ such that  $ H \subseteq G$.
% 
% (iii) We have $A \setminus F
% = A \cap (\Var \setminus F)$
% and, as $F$ is finite, we conclude by (ii) above.
%
% (iv) This immediately follows from the fact that $\Var$ is infinite.
%\end{proof}

Secondly, we note that it is fairly common in the literature of the lambda calculus 
to assume that the set of variables
is \emph{countably} infinite, see
for instance Hindley and Seldin \cite[Def. 1.1]{HS}.
Of course, in order to make
the above properties of cofinite sets  true,
countable infinity  suffices.
In this paper we do not need assume that the set of variables is  countable and the reason is plain:
we never invoke  that fact  in our definitions and proofs.

Finally --- and most importantly ---
we observe that
it is also common  to assume that 
the set of variables
comes equipped with an \emph{ordering}; again, see
 \cite[Def. 1.1]{HS}.
  Here,  the motivation
is ultimately related  to the definition of \emph{substitution}: if  this notion
 is formalized as in  \cite{HS},
then at some point of the definition  we need to invoke the existence of a  \emph{choice function} $f$ 
from cofinite
sets of variables   to variables
such that $f(S) \in S$ for every  $S$.  
(Note that cofinite sets must be
non--empty for this discussion to make sense and this is indeed the case!)
We refer to Vestergaard \cite{VES} for more on this aspect,
where the \emph{axiom of cofinite choice} is explicitly discussed.
In passing, 
we also mention that choice functions
are  specifically  introduced
by Stoughton \cite{STOUGHTON} in his approach
to \emph{simultaneous substitution},
see also
 Copello, Szasz and Tasistro 
\cite{COPELLO}.
The crucial point of the present discussion is this:
if the set of variables is appropriately ordered, then a  choice function
$f$ 
 naturally arises. Indeed, 
for each cofinite set  $S$ we can simply define $f(S)$ as
 \emph{the first, in the given order, variable  which is  a member of $S$}  (see also Remark \ref{remcurry}).
 In this paper we do not
define substitution as in \cite{HS}:
in Section  \ref{STA} we follow an elegant
approach due to  Barendregt
\cite{BARTH}  which allows us to define substitution, lambda theories and extensional theories without
assuming  any ordering on variables.
To sum up, in the present work we do not need to impose any ordering
on the set of variables. Even more generally,  we do not need to appeal 
to the existence of some choice function either.

\section{Prelambda Congruences} \label{PBC}

In this section, we introduce 
the notion of prelambda congruence and discuss the conditions which define
this concept.
Basic properties of prelambda congruences
will be studied in the next section. 

Prelambda congruences are our first step
towards our formalizations of the concepts of
lambda theory and extensional theory. Even though it is  technically incorrect,  it may be helpful to think
of \emph{pre}lambda congruences 
 as  lambda
theories  which  do not necessarily satisfy the condition of
$\alpha$--renaming 
--- but, of course,  they have  to satisfy
the $\beta$--rule.
(In a similar manner, \emph{pre}orders
can be seen as partial orders which do not necessarily satisfy anti--symmetry 
but they do satisfy reflexivity and transitivity.)
As a matter of fact, we show that  there exists a prelambda congruence $\cv$ with the following property:
we have $\lambda \tx \tx \!\! \nREL
\lambda \ty \ty$ for all variables
$\tx$ and $\ty$ such that $\tx \neq \ty$.
Since $\lambda \tx \tx \cv \lambda \ty \ty$ holds
in every lambda theory $\cv$, 
 in order to develop our formalization of lambda  theory it is actually \emph{necessary} to add
$\alpha$--renaming at some stage. 

%We shall consider this situation in Section \ref{lambdacong}.

Prelambda congruences which
satisfy a suitable condition of $\alpha$--renaming will be considered in Section
\ref{lambdacong} under the name of 
\emph{lambda congruences}.
Similarly,  in Section \ref{ACF} we shall introduce \emph{extensional congruences}, which are prelambda congruences which
satisfy an appropriate condition of $\eta$--extensionality.
%In Section \ref{EQU}, 
%we shall prove 
%that
%lambda congruences
%and extensional
%congruences
% coincide with  lambda theories and 
%   extensional  theories, respectively.

%A similar observation  applies to the case of  extensional
%theories and $\eta$--extensionality as well.

There are essentially 
two motivations for introducing the concept of prelambda congruence.

Firstly, the results that
we prove in the next section
on the concept of \emph{independence of variables} do not depend
on  $\alpha$--renaming
and $\eta$--extensionality.

Secondly, 
the factorization
of lambda and extensional congruences respectively
as  ``prelambda congruences plus $\alpha$--renaming'' 
and ``prelambda congruences plus $\eta$--extensionality''
will provides us with a  mathematically reasonable understanding
of these two syntactical  conditions in terms of
extensionality.

It is time for us to formally introduce
prelambda congruences.

\begin{definition}[Prelambda congruence] \label{deri}
Let $\REL$ be a congruence. We say that $\REL$
is a \textbf{prelambda congruence} if
it also satisfies the following conditions:
 
\begin{itemize}

%\item[($\beta_\const$)] \  $[\lambda \tx\ccc]\td \cv \ccc$;
\item[($\beta_1$)]   $[\lambda \tx\tx]\td \cv \td$;
    \item[($\beta_2$)]   $[\lambda \tx\ty]\td \cv \ty$, provided $\tx \neq \ty$; 
     \item[($\beta_3$)]  $[\lambda \tx  [\ta\tb]] \td \cv  [[\lambda \tx \ta]\td][[\lambda \tx \tb]\td]$;   
     \item[($\beta_4$)]  $[\lambda \tx [\lambda \tx  \ta]]\td \cv \lambda \tx  \ta$;  
      \item[($\beta_5$)] $[\lambda \ty\td]\tx \cv \td$ implies
$ [\lambda \tx   [\lambda \ty  \ta]] \td \cv  \lambda \ty   [[\lambda \tx  \ta]\td ]$, \\  provided $\tx \neq \ty$; 
 \end{itemize}

\noindent where   $\ta$, $\tb$ and $\td$  are arbitrary terms, and  
   $\tx$ and $\ty$ are arbitrary variables. 
    We collectively refer to conditions ($\beta_1$) to ($\beta_5$) above as \textbf{beta conditions}.
   %   The set of all prelambda congruences is denoted
%   by $\cP\cC$.
   \qedef
 \end{definition}
 
 Thus, in plain words a prelambda congruence is any binary
 relation on terms which simultaneously
 satisfies all structural and beta conditions.
 
A simple example of prelambda congruence is the whole set
 $\Lam^2$.  A more interesting example is  given in the following definition.
 
 \begin{definition}[Prelambda conversion] \label{preconv}
We call
 \textbf{prelambda conversion}
 the  prelambda congruence 
inductively defined
 by    structural and 
beta conditions.
 \qedef
 \end{definition}

  In Corollary \ref{coroconsi}
 below it is shown that 
  $\Lam^2$ and prelambda conversion are in fact different relations.
Hence, if we regard
prelambda conversion 
   as  an equational  theory  where
 equations are  expressions of the form $\ta \cv \tb$,  then 
    prelambda conversion  is \emph{consistent}
   in the  sense of Hilbert and Post.  
 
 Let us now discuss the conditions of Definition \ref{deri} in some detail.
% Let $\cv$ be a prelambda congruence.

% Conditions  ($\sR$), ($\sS$), ($\sT$), ($\sL$) and ($\sA$)
% appear --- sometimes in equivalent form --- in any formalization of extensional conversion. 

% In the first place, we note that no ancillary concept appears in the formalization of these congruences.

 With possibly some minor differences, 
 conditions ($\beta_1$), ($\beta_2$),
($\beta_3$) and ($\beta_4$) 
also appear in some work by 
 Henkin \cite{HENKIN},  Andrews \cite{AndrewsBook}, \RV   \cite{RV1,RV85,RV},   and Pigozzi and Salibra
\cite{PSLM,PS,PS98}.
For the sake of completeness,
we also mention that
conditions ($\beta_1$) and
($\beta_3$) 
are also considered by 
Andrews 
 \cite{ANDREWS1965},
but  there conditions
($\beta_2$) of ($\beta_4$) are condensed into the following   clause: 
\begin{itemize}
\item[ ($\gamma_{1}$)]
 $[\lambda \tx \ta] \td \cv \ta$, provided $\tx$ does not occur free in $\ta$.
\end{itemize}

%  As  noticed by Meyer \cite{MEYER}
%  it possible to drop condition ($\sR$)
% from the formalization and replace
% conditions ($\sS$) and ($\sT$) with a single condition.
%  However, the present conditions are so natural and well--established
%  that we prefer to have them explicitly in our setting.
  
Among beta conditions, the more interesting one is certainly
  ($\beta_5$). 
This condition  constitutes a point of  divergence in the relevant literature. 
Actually,  ($\beta_5$) as formulated in this paper  seems to be 
new.
 Consider the following conditions:
    \begin{itemize}
 \item[($\gamma_2$)]  
$ [\lambda \tx   [\lambda \ty  \ta]] \td \cv  \lambda \ty   [[\lambda \tx  \ta]\td ]$, \\ provided $\tx \neq \ty$ and $\ty$
 does not occur free in $\td$;
 \item[($\gamma_3$)] 
$ [\lambda \tx   [\lambda \ty  \ta]] \td \cv  \lambda \ty   [[\lambda \tx  \ta]\td ]$, \\ provided 
$\ty$ is distinct from $\tx$ and all variables 
occurring in  $\td$;
\item[($\gamma_4$)] 
$ [\lambda \tx   [\lambda \ty  \ta]] \td \cv  \lambda \ty   [[\lambda \tx  \ta]\td ]$, \\  provided $\tx \neq \ty$ and at least one of these two 
conditions hold:  $\tx $ does not occur free in $\ta$, $\ty$ does not occur free in $\td$;

\item[($\gamma_5$)]  $[\lambda \ty\td]\tz \cv \td$ implies
$ [\lambda \tx   [\lambda \ty  \ta]] \td \cv  \lambda \ty   [[\lambda \tx  \ta]\td ]$, \\  provided $\tx \neq \ty$ and
 $\ty \neq \tz$;
 \item[($\gamma_6$)] 
$ [\lambda \tx   [\lambda \ty  \ta]] [[\lambda \ty\td]\tz] \cv  \lambda \ty   [[\lambda \tx  \ta][[\lambda \ty\td]\tz] ]$, \\  provided $\tx \neq \ty$ and
 $\ty \neq \tz$.

 \end{itemize}

  Again, up to minor details, 
  ($\gamma_2$) is the version of ($\beta_5$) considered by Henkin   \cite[Axiom 7.5 p. 331]{HENKIN} and Andrews
  \cite[$4_4$ p. 3]{ANDREWS1965}, 
  ($\gamma_3$)
is another version  considered
by Andrews \cite[$4_4$ p. 164]{AndrewsBook} and finally
($\gamma_4$) is the one taken into account by \RV  \cite[$\beta 4$ p. 29]{RV}. As we can see, 
all these conditions explicitly mention
some ancillary concept.

 The condition which is most
 similar to our ($\beta_5$) is the one given by Pigozzi
 and Salibra 
in the context of lambda abstraction algebras ($\gamma_5$)  \cite[$\beta_6$ p. 12]{PS};  
in fact,  our condition ($\beta_5$) is actually
a simplification of ($\gamma_5$) ---
in the sense that
 our  condition requires
 two variables and one inequality to be expressed,
 while ($\gamma_5$)
 needs three variables and two inequalities.
Our ($\gamma_6$) is a  shortened
 version of ($\gamma_5$)
due to Pigozzi
 and Salibra   \cite[Prop. 1.5]{PS}.
 Conditions
 ($\gamma_5$) and  ($\gamma_6$)
 are equivalent in every lambda abstraction algebra.
%  While compact,
% we believe  that ($\gamma_6$)  is too complicated for our aims, but 
% it is worth noticing that
% for the purposes of 
% the theory of lambda abstraction algebras  ($\gamma_6$)
% plays a great role, as it serves to prove that the class of these algebras  forms a \emph{variety} (in the sense of universal algebra).
 Finally, it is worth mentioning that Pigozzi and Salibra also consider another ``beta condition'' in their setting, namely
$[\lambda \tx   \ta]\tx \cv \ta$  (see \cite[$\beta_3$ p. 12]{PS}).
 But for our purposes,  we do not  need
to consider  this condition at all.

As we have just seen,  
(suitable versions of) 
beta conditions have
already been  considered in the literature.
They perhaps  seem a bit complicated at first,
but a
 closer look  reveals that they are in fact
 quite natural and even easy to memorize.
To see this, notice that
they are all of the form
$[\lambda \tx \tm]\td \cv \tg$.
Now, if we replace
each term of the form
$[\lambda \tx \tm]\td$ by an expression of the form $\pl\td/\tx\pr(\tm)$ 
and write $=$ instead of $\cv$, then by a slight
rearrangement
we obtain
\begin{itemize}
\item[($\beta_1^s$)]  $\pl \td/\tx \pr(\tx) \, = \, \td$;
    \item[($\beta_2^s$)]   $\pl \td/\tx \pr(\ty) \, = \, \ty$,  if $\tx \neq \ty$; 
     \item[($\beta_3^s$)]   $\pl \td/\tx \pr(\ta\tb) \, = \,  \pl \td/\tx \pr(\ta)\pl \td/\tx \pr(\tb)$;   
     \item[($\beta_4^s$)]   $\pl \td/\tx \pr(\lambda \tx\ta) \, = \, \lambda\tx\ta$;  
     \item[($\beta_5^s$)] 
$\pl \td/\tx \pr(\lambda \ty\ta) \, = \, \lambda\ty\pl \td/\tx \pr(\ta)$,   \\ if  $\tx \neq \ty$ \ (provided 
  $\pl \tx/\ty \pr(\td) \, = \, \td$).
\end{itemize}

Conditions ($\beta_1^s$), ($\beta_2^s$),
($\beta_3^s$), ($\beta_4^s$) and ($\beta_5^s$) above
really look like a kind of inductive \emph{definition
of substitution} which is not too dissimilar to the ones we can find in the literature (\cf Definition \ref{subdef}).
In passing, it is worth observing that
these clauses are very
similar to some of the conditions which axiomatize  \emph{lambda substitution algebras}, 
 algebraic structures
 introduced by Diskin and Beylin \cite{DB} with the purposes of
algebraizing  the lambda calculus  --- in  fact, these algebras share several similarities with lambda abstraction algebras.

At this point, we hope that we have almost convinced the reader that
 beta conditions are quite natural
 and intuitive. However, if there is
 some reader which
still feels uncomfortable with
 the precondition $[\lambda \ty\td]\tx \cv \td$ in condition ($\beta_5$) (that is to say, the proviso
$\pl \tx/\ty \pr(\td)\, = \, \td$  in   ($\beta_5^s$) above), then we would like to reassure the reader by
saying that the whole next section
is dedicated to the study
of the  ``consequences'' of  $[\lambda \ty\td]\tx \cv \td$. 

We now show that $\alpha$--renaming is not  available 
in every prelambda congruence.
For, let $\tx$ and $\ty$ be two distinct variables, and let $\cv$ be an arbitrary prelambda congruence.
If $\cv$ were a lambda theory,
we would certainly have
$\lambda \tx \tx \cv \lambda \ty \ty$
by using $\alpha$--renaming.
However, here nothing ensures
that we have
$\lambda \tx \tx \cv \lambda \ty \ty$. In fact, 
  the following theorem
shows that 
any attempt of 
 deriving $\lambda \tx \tx \cv \lambda \ty \ty$ by means of  our  structural and beta conditions is doomed to failure.

\begin{theorem}
\label{teononalpha}
 Let $\tx$ and $\ty$ be two variables
 such that $\tx \neq \ty$. Then,
 there exists a prelambda congruence
 $\mo$ 
such that \linebreak $\lambda \tx \tx \nocv\lambda \ty\ty$. \qed
\end{theorem} 
The proof of Theorem \ref{teononalpha}
is quite long and technical.
Furthermore, in order to produce it
we are forced to introduce
further  notions and methods that we do not
need in the rest of this paper.
For these reasons,  it is  postponed  to Section \ref{modelsection}.
Nevertheless, the interested
reader can directly go there without any
hesitation, as no further material
is needed to read that part of this work.

As a consequence
of the previous result,  a very simple
proof of the fact 
 that $\Lam^2$ and prelambda conversion are distinct relations  is at our disposal.

\begin{corollary} \label{coroconsi}
Let $\cv$ be prelambda conversion,
and let $\tx$ and $\ty$ be two distinct variables.
Then, we have $\lambda \tx \tx \!\! \nREL \lambda \ty \ty$.
In particular,  
$\cv$   is consistent.

\end{corollary}
\begin{proof}
Since $\cv$
is the inductively defined prelambda congruence, it is  the intersection of all   prelambda congruences. In particular, 
we have  $\cv \, \subseteq  \, \mo$.
From this and Theorem \ref{teononalpha} it follows that $\lambda \tx \tx \! \nREL \lambda \ty \ty$.
In particular, we have $\consi$.
\end{proof}

\section{Basic Properties of Prelambda Congruences}
\label{indsec}

%In this section, we study
%some general properties shared by every
%prelambda congruence. 

In the previous section,  we did not 
discussed 
the significance of the precondition
  $[\lambda \ty\td]\tx \cv \td$
in 
 condition ($\beta_5$).  In this section,  we aim to fill this gap.
 Since this requirement turns
 out to be very important in our setting, it is convenient
 to introduce some terminology 
to describe
 this situation. 

\begin{definition}[Independence of variables]
Let $\cv$ be a prelambda congruence and let $\ta$ be a  term. We define the set
of variables $\ind{\ta}$ as follows:

\smallskip

{\centering
$\ind{\ta} \eqdef \lb \tx \taleche [\lambda \tx \ta]\tz \cv \ta$ for some variable $\tz$ \\ such that $\tx \neq \tz\rb$\enspace.
\par}

\smallskip

\noindent Let $\tx$ be a variable. We say that  $\ta$ \textbf{is independent of $\tx$ in $\cv$} if $\tx \in \ind{\ta}$.
\qedef
\end{definition}

Thus, since in ($\beta_5$) we have $\tx \neq \ty$,  in the new terminology the
precondition  $[\lambda \ty\td]\tx \cv \td$ means  that $\td$ is independent
of $\ty$ in $\cv$. 
There is a clear  analogy with
the proviso  ``$\ty$
 does not occur free in $\td$''
of Henkin's version of 
 ($\beta_5$)  --- the condition   we called
 ($\gamma_2$) in the previous section. 
%However, observe
% that the concept of independence
%is formulated by means of the relation $\cv$ while 
% the  notion of free  variable
% only depend on 
% terms.
A precise relationship between
independent variables
and free variables is given 
in  Theorem  \ref{EQ}.

The notion of independence of variables in the lambda calculus has been  introduced
and studied by  Pigozzi and Salibra  in the context of lambda abstraction algebras \cite{PSLM,PS,PS98}.
% It is reminiscent
%of a similar notion
%which plays and important role in the theory of  \emph{cylindric algebras},  an ``algebrification'' of
%first--order classical logic with identity developed by  Henkin, Monk and Tarski \cite{HMT}.
Even though  prelambda congruences are not lambda abstraction algebras, 
facts which look quite similar to  ``sharpened versions'' 
 of Proposition \ref{IND}
 and
 Lemma \ref{lemmalpha} 
 below have been  already established, see 
 \cite[Lem 1.6 and Prop. 1.7]{PS}.
 With a little effort, 
we could generalize our facts in the same
vein,  but  we find  more instructive
 to present  them in the actual   form because 
we do not need anything more elaborated to
prove 
  our main results.
% In particular, 
% we prefer to delegate  a detailed study of the concept of independence of variables in the lambda calculus to  subsequent work.
% 

Technically speaking,  the peculiar use of
the concept of independence of variables we make in this paper
allows us to avoid the introduction
of the sets of all, free and bound variables occurring in a term
in the definitions of our notions
of prelambda, lambda and extensional congruences.

To begin the study of the concept of independence of variables in prelambda congruences, 
we show that a property similar to the condition called ($\gamma_1$) that we discussed in the previous section also holds in our setting.

\begin{proposition} \label{IND} \label{prop notfree} Let $\cv$ be a prelambda congruence. Let $\tx$ be a variable and let $\ta$ be a term. Suppose that $\tx \in \ind{\ta}$.
 Then, we have $[\lambda \tx \ta]\td \cv \ta$ for every term $\td$.
\end{proposition}
\begin{proof}
Let $\td$ be a term.
Since $\tx\in \ind{\ta}$,  there exists a variable $\tz$ such that $\tx \neq \tz$ and
$[\lambda \tx   \ta ] \tz
\cv \ta$.
Let 
$\tc \eqdef [\lambda \tx   [[\lambda \tx   \ta ] \tz]] \td$, $\te \eqdef [\lambda \tx   [\lambda \tx   \ta ]] \td$ and $\tf \eqdef
[\lambda \tx   \tz] \td $.
Since $[\lambda \tx   \ta ] \tz \cv \ta$, we obtain  $\lambda \tx   [[\lambda \tx   \ta ] \tz]
\cv \lambda \tx  \ta$ by $(\sL)$.
We have $\td \cv \td$ by ($\sR$), and we get $ \tc \cv [\lambda \tx \ta]\td$ by ($\sA$). Now,
by ($\beta_3$)  we deduce $\tc \cv \te\tf$
and
 we have 
$\te \cv \lambda \tx \ta$ by ($\beta_4$). 
As $\tx \neq \tz$, we obtain $\tf  \cv \tz$ by ($\beta_2$) and it follows from ($\sA$) that 
$\te\tf \cv [\lambda \tx \ta]\tz$.
By using ($\sS$) 
we have $[\lambda \tx\ta]\td \cv \tc \cv \te\tf \cv
 [\lambda \tx \ta]\tz \cv \ta$
 and  we finally get $[\lambda \tx\ta]\td \cv  \ta$ from  ($\sT$).
 \qedhere
\end{proof}

%As a consequence of the previous proposition,  the following lemma  can be used
%in place of condition ($\beta_5$).
%
%\begin{lemma} \label{beta5rev} \label{prop notfree} Let $\cv$ be a prelambda congruence. Let $\tx$ and $\ty$ be  variables such that $\tx \neq \ty$,  and let $\ta$ and $\td$ be terms. Then,  $\ty \in \ind{\td}$ implies $  [\lambda \tx   [\lambda \ty  \ta]] \td \cv  \lambda \ty   [[\lambda \tx  \ta]\td ]$.
%\end{lemma}
%\begin{proof}
%Suppose that $\ty \in \ind{\td}$.
%Then,   we have
%$[\lambda \ty \td]\tx \cv \td$ by Proposition \ref{IND}. Hence, we obtain
% $  [\lambda \tx   [\lambda \ty  \ta]] \td \cv  \lambda \ty   [[\lambda \tx  \ta]\td ]$ by means of ($\beta_5$).
%\end{proof}

In regard to the syntactical category a term belongs to, 
the following  lemma
gives a partial description of  sets
of independent variables.

\begin{lemma} \label{lemmalpha}
Let $\cv$ be a prelambda congruence.
Let $\tx$ be a  variable, 
and let $\tb$ and $\tc$ be terms.
  Then, the following statements
  hold:
  \begin{itemize}
  \item[\emph{(i)}]
  $\Var \setminus \lb \tx \rb
\subseteq \ind{\tx}$;
    \item[\emph{(ii)}] $\ind{\tb} \cap \ind{\tc}  \subseteq  \ind{\tb\tc}$;
      \item[\emph{(iii)}]    $\ind{\tb} \cup \lb \tx \rb \subseteq \ind{\lambda \tx \tb}$ (in particular,
      $\ind{\tb} \subseteq \ind{\lambda \tx \tb}$).
  \end{itemize}

\end{lemma}
\begin{proof}
(i)
Let $ \ty \in \Var \setminus \lb \tx \rb$.
Since $\tx \neq \ty$,    we have $[\lambda \ty \tx]\tx \cv \tx$  by ($\beta_2$). Hence, we obtain $\ty \in \ind{\tx}$.

(ii)
Let $\ty \in  \ind{\tb} \cap \ind{\tc}$ and let $\tz$ be a variable such that $\ty \neq \tz$. 
By Proposition \ref{IND}
it follows that    $ [\lambda \ty \tb]\tz \cv \tb$ 
and  $ [\lambda \ty \tc]\tz \cv \tc$, 
%as $\ty \in  \ind{\tb}$ and $\ty \in \ind{\tc}$.
and  we obtain $[[\lambda \ty \tb]\tz][[\lambda \ty \tc]\tz]\cv \tb\tc$ by ($\sA$).
 From ($\beta_3$), we have 
$[\lambda \ty [\tb\tc]]\tz \cv  [[\lambda \ty \tb]\tz][ [\lambda \ty \tc]\tz]$ and  we get
$[\lambda \ty [\tb\tc]]\tz \cv \tb\tc$ from ($\sT$).
Therefore, we conclude that $\ty \in \ind{\tb\tc}$.

(iii) Let $\tz$ be a variable such that $\tx \neq \tz$. Then, we have $[\lambda \tx[\lambda \tx\tb]]\tz \cv \lambda \tx \tb$ by ($\beta_4$) and hence  it follows that $\tx \in \ind{\lambda \tx\ta}$. 
Now, let   $\ty \in \, \ind{\tb} \setminus \lb \tx \rb$ and let $\tz$ be a variable such that $\tx \neq \tz$ and $\ty \neq \tz$. 
By Proposition \ref{IND}, we get  
$[\lambda \ty  \tb] \tz \cv \tb$.
Since $\tx \neq \ty$ and $\tx \neq \tz$, we have
$[\lambda \tx \tz]\ty \cv \tz$ by ($\beta_2$) and hence
 $[\lambda \ty  [\lambda \tx \tb]] \tz \cv  \lambda \tx  [[\lambda \ty \tb]\tz ]$ by ($\beta_5$).
Since $[\lambda \ty  \tb] \tz \cv \tb$, we also have 
$\lambda \tx[[\lambda \ty  \tb] \tz] \cv \lambda \tx\tb$ from ($\sL$). Finally,
we  obtain
$[\lambda \ty  [\lambda \tx \tb]] \tz  \cv \lambda \tx \tb$
from ($\sT$). As $\ty \neq \tz$,
we conclude that $\ty \in \ind{\lambda \tx\tb}$. 
\end{proof}

%\begin{lemma} \label{lemmalphalast}
%Let $\cv$ be a pure congruence
%and
%let $\ccc$ be a constant.
%  Then, we have 
%  $\ind{\ccc} = \Var$.
%  \end{lemma}
%\begin{proof}
%Let $\tx$ be an arbitrary variable.
%Let $\ty$ be a variable such that $\tx \neq \ty$.
%By using ($\beta_\const$)
%we obtain $[\lambda \tx \ccc]\ty
%\cv \ccc$. Therefore,
%we conclude $\tx \in \ind{\ccc}$
%and $\ind{\ccc} = \Var$.
%\end{proof}

Our next step is 
to relate the notion of independence 
of variables in prelambda  congruences to the usual concept of \emph{free variable}.
Up to now, we dealt with them informally
and  appealed to the reader's previous
knowledge of this concept.
Since at this point we need to prove some formal facts
about free variables, we now recall the formal
definition. Nevertheless, for the sake of better readability, 
in informal discussions  below
we shall often express the concept
of free variable in words, rather
than
employing the symbolism
of Definition \ref{nonfree}.

As
already noticed by Welch \cite[Rem. 0.3.4]{WELCH}, for the development of the lambda calculus
it is actually more clear and convenient 
to formalize the concept of 
variable which does   \emph{not}
occur free in a given term.
Here, we follow the same idea
 because this approach makes the connection between non--free and independent variables tighter.
Non--free
variables are sometimes called \emph{fresh} variables in the literature,
for instance in Copello, Szasz and Tasistro 
\cite{COPELLO}.

\begin{definition}[Non--free  variable] \label{nonfree} Let $\ta$ be a term.
We define the  cofinite set of variables  $\NF(\ta)$
  by induction on the structure of $\ta$ as follows:
\begin{itemize} 
\item $\NF(\tx)  \eqdef \Var \setminus\{\tx\}$;
\item  $\NF(\tb\tc)  \eqdef \NF(\tb) \cap \NF(\tc)$;
\item
$\NF(\lambda \tx   \tb) \eqdef \NF(\tb) \cup \lb \tx \rb$.
\end{itemize}
%
%\begin{itemize}
%\item[($\NB_1$)] \ $\NB(\tx)  \eqdef \Var$;
%\item[($\NB_2$)] \ $\NB(\tb\tc)  \eqdef \NB(\tb) \cap \NB(\tc)$;
%\item[($\NB_3$)]  \  $\NB(\lambda \tx   \tb) \eqdef \NB(\tb) \setminus \lb \tx \rb$.
%\end{itemize}
\noindent 
%Let $\tx$ be a variable.
 Finally, we  say that a variable \textbf{$\tx$ does not occur free in $\ta$} if $\tx \in \NF(\ta)$.
  \qedef
\end{definition}

%
%We now show a simple but important fact.
%
%\begin{proposition} \label{coffree}
% Let $\ta$ be a term.
%Then, 
%$\NF(\ta)$ is a cofinite set.
%\end{proposition}
%\begin{proof}
%We reason  by induction on the structure of $\ta$.
%
% If $\ta = \tx$, then  $\NF(\tx) = \Var \setminus \lb \tx \rb$ is a cofinite set. 
% 
%  If $\ta = \kkk$, then  $\NF(\kkk) = \Var$ is a cofinite set. 
% 
%If $\ta = \tb\tc$, then 
%by inductive hypothesis it follows that
%$\NF(\tb)$ and $\NF(\tc)$
%are cofinite.
%By Lemma  \ref{lemmacof}(i),
%$\NF(\tb)  \cap   \NF(\tc) = \NF(\tb\tc)$
%is a cofinite set.
%
%If $\ta = \lambda \tx \tb$, then
%by inductive hypothesis it follows that
%$\NF(\tb)$ is  a cofinite set.
%Since $\NF(\tb) \subseteq \NF(\tb) \cup \lb \tx \rb = \NF(\lambda \tx\tb)$ we conclude that $\NF(\lambda \tx \tb)$ is a cofinite set  by Lemma \ref{lemmacof}(ii).
% \end{proof}
% 
 The next theorem  is fundamental for the main results of this paper.

\begin{theorem} \label{propfree}
Let $\cv$ be a prelambda congruence and let
$\ta$ be a term.
Then,  we have $\NF(\ta) \subseteq \ind{\ta}$.

\end{theorem}
\begin{proof} 
  We proceed by induction on the structure
of $\ta$.

 Suppose that $\ta = \tx$.
By Lemma \ref{lemmalpha}(i), we have $\NF(\tx) =  \Var \setminus \lb \tx \rb \subseteq  \ind{\tx}$.

 Suppose that
$\ta = \tb\tc$. 
 By inductive hypothesis, we have $ \NF(\tb)  \subseteq  \ind{\tb}$ and  $\NF(\tc)  \subseteq \ \ind{\tc}$.
 Thus, we obtain $\NF(\tb\tc)
 = \NF(\tb) \cap \NF(\tc) \subseteq
 \ind{\tb}  \cap \ind{\tc}
 \subseteq \ind{\tb\tc}$ by using Lemma
 \ref{lemmalpha}(ii).

  Finally, suppose that
$\ta = \lambda \tx   \tb$.
By inductive hypothesis, we have $ \NF(\tb)  \subseteq \ind{\tb}$.
By  Lemma \ref{lemmalpha}(iii),
it follows that $\NF(\lambda \tx \tb)
= \NF(\tb)
 \cup \lb \tx \rb \subseteq
 \ind{\tb} \cup \lb \tx \rb
 \subseteq \ind{\lambda \tx\tb}$.
 \qedhere
\end{proof}

An immediate consequence of Theorem
\ref{propfree} is  that
in every  prelambda congruence $\cv$  all   sets of variables of the form $\ind{\ta}$ are  cofinite.
A more remarkable consequence is that
condition ($\gamma_2$), discussed in the previous section,  always holds in our setting as we show in the second part of the following proposition.

\begin{proposition} \label{beta5rev} \label{prop notfree} Let $\cv$ be a prelambda congruence. 
Let $\ta$ and $\td$ be terms. 
Let  $\tx$ and $\ty$ be   variables
such that $\tx \neq \ty$.
Then, the following statements hold:
\begin{itemize}
\item[\emph{(i)}] Suppose that
$\ty \in \ind{\td}$. Then, we
have  $  [\lambda \tx   [\lambda \ty  \ta]] \td \cv  \lambda \ty   [[\lambda \tx  \ta]\td ]$.

\item[\emph{(ii)}] Suppose that
$\ty \in \NF(\td)$. Then, we
have  $  [\lambda \tx   [\lambda \ty  \ta]] \td \cv  \lambda \ty   [[\lambda \tx  \ta]\td ]$.
\end{itemize}
\end{proposition}
\begin{proof}

 (i)
We have
$[\lambda \ty \td]\tx \cv \td$ by Proposition \ref{IND}. As $\tx \neq \ty$, we obtain
 $  [\lambda \tx   [\lambda \ty  \ta]] \td \cv  \lambda \ty   [[\lambda \tx  \ta]\td ]$ by  ($\beta_5$).
 
 (ii) It follows from (i) above, as  we have
 $\ty \in \ind{\td}$ by Theorem \ref{propfree}. 
\end{proof}

%
%\begin{example} \label{exvar}
%Let $\cv$ be a prelambda congruence.
%%Let $\tx$ be a variable and consider
%%the term $\lambda \tx\tx$.
%%Since $\NF(\lambda \tx \tx) = (\Var \setminus \lb \tx \rb) \cup \lb \tx \rb = \Var$,  we have $\ind{\lambda \tx\tx} = \Var$ by the previous theorem.
%Let us call 
% \emph{closed} any term $\ta$ such that
%$\NF(\ta) = \Var$. Then,  by  Theorem \ref{propfree} we have
%$\ind{\ta} = \Var$ for every closed term $\ta$.  
% \qedef
%
%\end{example}

%On a first reading of
% our condition
%($\beta_5$) in the previous section,  a criticism
% about the value  of our elimination of the concept of free variable
%  could come to  mind: when compared to \eg condition ($\gamma_2$), one can think that
%proving $[\lambda \ty \td]\tx \cv \td$ is definitely
%more complicated than 
%calculating the set $\NF(\td)$
%by a straightforward induction. 
%Now, thanks to the  the previous theorem we can argue that the criticism 
%is actually inappropriate.
%For,  
%by Proposition \ref{IND} in order
%to find a proof of $[\lambda \ty \td]\tx \cv \td$ it suffices to prove
%that $\ty \in \ind{\td}$
%and by the previous theorem
%it is sufficient to show that
%$\ty \in \NF(\td)$. 
%Hence, as in the case of ($\gamma_2$),  the simple
%computation of the set
%$\NF(\td)$ is enough
%to apply ($\beta_5$). 

We now  further analyze  the relationship
between independent and non--free variables, though
the facts we are now going to establish are not strictly needed to prove our main results.
 To begin with, we show that
independence and non--freedom are not equivalent notions.
 In the following  example
 we show  that  the converse of 
 Theorem \ref{propfree}
 does not hold.

 \begin{example} \label{extriv}
 Let  
us consider the case $\inconsi$.
Since  all terms are related,
 we have $\ind{\ta}= \Var$
for every term $\ta$. However, 
it is not true that
$\NF(\ta) = \Var$ for every term $\ta$; for instance, we have $\NF(\tx) = \Var \setminus \lb \tx \rb$ for every variable $\tx$.
 \qedef
\end{example}

The previous example
does not exclude the existence of a prelambda congruence $\cv$ in which it is possible to have
 $\ind{\ta} = \NF(\ta)$ for every term $\ta$. However,  
 we can show 
that such a $\cv$ cannot exist.
 Actually,  we   prove an even stronger version of this fact in Corollary \ref{coroind} below, which  is a consequence of the  theorem
that we now show.

\begin{theorem}[Independence and non--free variables] \label{EQ}
Let $\cv$ be a prelambda congruence and let
$\ta$ be a term. We have

\smallskip

{\centering
$\ind{\ta} = \lb \tv \taleche \tv \in \NF(\tb) $ \textnormal{ for some term $\tb$\\ $\phantom{aaass}$ such that } $\tb \cv \ta \rb$\enspace.
\par}
\end{theorem}
\begin{proof}
 Let $S \eqdef \lb \tv \taleche \tv \in \NF(\tb) $ for some term $\tb$ such that $\tb \cv \ta \rb$ and let $\tx$ be a variable.

Suppose that $\tx \in \ind{\ta}$.
Then, 
 there exists a variable $\tz$ such that
$\tx \neq \tz$ and $[\lambda \tx   \ta]\tz \cv \ta$. Let $\tb \eqdef [\lambda \tx   \ta]\tz$.
Since  $\tx \neq \tz$, we have
$\tx \in (\NF(\ta) \cup \lb \tx \rb) \cap(\Var \setminus \lb \tz \rb) = \NF(\tb)$. As  $\tb \cv \ta$,
we conclude that
 $\tx \in  S$.

Suppose now that
$\tx \in S$.
Let $\tb$ be a term such that 
$\tx \in \NF(\tb)$ and $\tb \cv \ta$. Since $\tx \in \NF(\tb)$, we have 
 $\tx \in \ind{\tb}$ by Theorem \ref{propfree}.
In particular, 
 there exists a variable $\tz$ such that
$\tx \neq \tz$ and $[\lambda \tx   \tb]\tz \cv \tb$.
Now, since $\tb \cv \ta$, we  have $\lambda \tx \tb \cv \lambda \tx \ta$ by ($\sL$). From ($\sR$) it follows
that $\tz \cv \tz$. So, we obtain
$[\lambda \tx \tb]\tz \cv [\lambda \tx \ta]\tz$
by ($\sA$).
By using ($\sS$), 
we have $[\lambda \tx \ta]\tz \cv [\lambda \tx \tb]\tz \cv \tb \cv \ta$
and we obtain
 $[\lambda \tx \ta]\tz \cv\ta$ from ($\sT$).
As $\tx \neq \tz$, we conclude that $\tx \in \ind{\ta}$.
\end{proof}

\begin{corollary} \label{coroind}
There exists  a term $\ta$ such that in every prelambda congruence $\cv$ we have
  $\ind{\ta} \nsubseteq \NF(\ta)$.

\end{corollary}
\begin{proof}
Let $\tx$ and $\ty$  be variables 
such that $\tx \neq \ty$, and
let $\ta \eqdef [\lambda \tx   \ty]\tx$.
Let $\cv$ be a prelambda congruence.
Since $\NF(\ta) = \NF(\lambda \tx \ty) \cap \NF(\tx) = \NF(\lambda \tx \ty) \cap \big(\Var \setminus \{\tx\}\big)$, we have $\tx \notin \NF(\ta)$.
Since $\tx \neq \ty$,  we have
$\ta \cv \ty$ from ($\beta_2$) and  $\ty \cv \ta$ from ($\sS$). Also, we have $\tx \in  \Var \setminus \lb \ty \rb = \NF(\ty)$ and so
it follows that  $\tx \in \NF(\tb) $ for some term $\tb$ such that $\tb \cv \ta$. Thus, we obtain $\tx \in \ind{\ta}$ by means of  Theorem \ref{EQ}.

Now,
if $\ind{\ta} \subseteq \NF(\ta)$
were true, then
we would get $\tx \in \NF(\ta)$ but this
would contradict
$\tx \notin \NF(\ta)$. Therefore, we have
$\ind{\ta} \nsubseteq \NF(\ta)$.
 \end{proof}

\section{Lambda Congruences}
\label{lambdacong}

We now turn our attention to the intuitive
interpretation of terms as
functions in prelambda congruences.
Let $\cv$ be an arbitrary prelambda congruence.

In Section
\ref{lambda sec}
we  observed 
that a term
of the form 
$\lambda \tx \ta$ 
should be thought as a function
whose domain is the whole set of terms $\Lam$ (including $\lambda \tx \ta$ itself).  
Under this  perspective,  
when the function $\lambda \tx \ta$ is
applied to an  input $\td$,  
\emph{any} term $\tb$ such that
$\tb \cv [\lambda \tx \ta]\td$ should be thought as  \emph{the} output
of $\lambda \tx \ta$ on input $\td$.
This  view is justified by the fact that
 the terms $\tb$ and $[\lambda \tx \ta]\td$ are meant to be same
 in $\cv$ 
when $\tb \cv [\lambda \tx \ta]\td$
holds, so that
syntactically different terms related to 
$ [\lambda \tx \ta]\td$ are just different \emph{representations}
of the same output. 
Depending on the practical aims,  better and more informative representations of the output can be computed using structural and  beta conditions.
Clearly, different prelambda congruences can produce different
outcomes; for instance, if $\inconsi$
then any term is the output
of any function on any input.

%A classical result in lambda calculus
%is the following:
%if the notion of sameness is implemented by lambda or extensional  conversion,
%then \emph{every function 
%has a fixed point}, see any textbook mentioned in the introduction.
%The situation here is similar
%as we now readily show.
%
%
%
%
%	
%	\begin{theorem}[Fixed point theorem]
%	\label{fixth}
% Let $\tx$ be a variable
%	and let $\ta$ be a term.
%	Then, there exists a term $\td$
%	such that in every 	 prelambda congruence $\cv$ we have $[\lambda \tx\ta] \td \cv \td$.
%		\end{theorem}
%\begin{proof}
%Let $\tb \eqdef \lambda \tx [[\lambda\tx\ta][\tx\tx]]$,
%$\tc \eqdef [\lambda \tx [\lambda \tx\ta]]\tb$, $\te \eqdef [\lambda \tx \tx]\tb$
%and $\td \eqdef \tb \tb$.
%
%Let $\cv$ be a prelambda congruence.
%We have 
%$\td \cv \tc [\lambda\tx [\tx\tx]\tb] \cv \tc [\te\te]$
%by using ($\beta_3$).
%By ($\beta_4$),
% it follows that 
%$\tc\cv \lambda \tx\ta$
%and by ($\beta_1$)
%we have $\te \cv \tb$.
%Hence, by using ($\sA$)
%we obtain
%$\te\te \cv \td$ and
% $\tc[\te\te] \cv [\lambda \tx\ta]\td$.
%Finally, by using ($\sS$)
%we have $[\lambda \tx\ta]\td \cv
%\tc[\te\te] \cv  \tc [\lambda\tx [\tx\tx]\tb] \cv \td$ 
%and by ($\sT$) we get $[\lambda \tx\ta] \td \cv \td$.
%\end{proof}	

%This fact
%has been explicitly employed
%in the proof of the theorem
%above, where the fixed point $\td$
%is in fact the function $\lambda \tx [[\lambda\tx\ta][\tx\tx]]$
%applied to itself. 	
	
Notice  that  we do  \emph{not} require that \emph{every} term has to be thought
as a function. So, variables
and applications can be simply  thought
as  \emph{constituents} for forming functions ---
of course,  they can be inputs and outputs of functions.
(In a similar fashion, in set theory
with \emph{urelements} it is not required  urelements  to be sets,
but they can be used to build sets.)
In Section \ref{ACF} the possibility
of regarding every term as a function will be investigated.

So far so good. However, if we  seriously want to develop a
 reasonable theory of functions inside the lambda calculus we should expect the following property of \emph{extensionality for functions} to be true:
two functions should be considered as  the same if they have identical input--output behaviour.
 This property clearly  holds 
 in set theory:   given a set $Z$
and two functions
$f$ and $g$ from $Z$ to itself  
we always have

\smallskip

{\centering
$f(d) = g(d)$ for every $d \in Z$ \qquad implies \qquad  $f= g$\enspace. \par}

\smallskip

\noindent Thus,  if we want to think \emph{abstractions
as functions} then
we should expect the following  property of \emph{extensionality for abstractions}  to hold:

\smallskip

{\centering
$[\lambda \tx \ta] \td
\cv [\lambda \ty \tb]\td$ for every term $\td$ \ \ implies \ \ $\lambda \tx \ta \cv \lambda \ty \tb$\enspace. \par}

\smallskip

\noindent
Quite suggestively, we  observe that  our formulation of extensionality for abstractions  strikingly  resembles  the \emph{axiom of extensionality} of von Neumann's set theory based on functions
\cite[p. 397 and Axiom I4 p. 399]{VN}. 

A serious  problem of prelambda congruences is that
it is not always the
case that
extensionality for abstractions holds. 
For instance, if $\cv$ is prelambda conversion (see Definition \ref{preconv}),  then  for distinct variables $\tx$ and $\ty$
and any term $\td$
 we have
 $[\lambda \tx \tx] \td \cv \td$ and  $[\lambda \ty \ty] \td \cv \td$
 by ($\beta_1$).
 So, by using ($\sS$) we obtain $ [\lambda \tx \tx] \td  \cv \td \cv [\lambda \ty \ty] \td$ and by ($\sT$)
 we get
 $ [\lambda \tx \tx] \td  \cv [\lambda \ty \ty] \td$ which shows that
 the two abstractions 
 $\lambda \tx \tx$ and $\lambda \ty \ty$ have the same input--output behaviour. But 
     $\lambda \tx \tx 
\cv \lambda \ty \ty$ does not hold,  by
Corollary \ref{coroconsi}.
Hence, in the section
we focus our attention on  prelambda congruences which satisfy extensionality of abstraction.

Nevertheless,  we do not regard prelambda congruences that do not satisfy such a form of
extensionality --- such as prelambda conversion   --- as  uninteresting.
We strongly believe that these congruences may serve  as  a foundation for the development 
of a  truly \emph{intensional} (as opposed to extensional)
theory of     \emph{algorithms} \ldots but this is another story.

Another  intriguing motivation for considering extensionality for abstractions 
is the fact that this condition
has  been present in the lambda calculus from the very beginning --- though under a completely different guise.
In fact,  in every prelambda congruence extensionality
for abstractions and $\alpha$--renaming are equivalent conditions,
as we show in Theorem \ref{thmalpha}.
%Furthermore, in Section
%\ref{EQU} we will show that
%that prelambda congruences
%which satisfy $\alpha$--renaming
%precisely corresponds to lambda theories. 
%In other words, extensionality
%for abstractions has always been present in the lambda calculus as equivalent to the 
%condition 
%which allow us to relate
% terms which \emph{differ
%only by the names of bound variables}. 

We believe that the usual 
informal explanations of $\alpha$--renaming,  as
``the names of bound variables
are immaterial'' and the like, 
provide neither a good justification 
of what is really going on nor
a good motivation for accepting this condition as natural.
Moreover,  in the literature we are not aware of  any line of work which attempts to clarify the real significance
of this condition.
In this paper, we can give a
mathematically
significant answer:
$\alpha$--renaming is just a compact reformulation of the
 property of extensionality for abstractions.

 We now formally introduce the concept of  \emph{lambda congruence}: 
 prelambda congruences which satisfy our  condition
of  $\alpha$--renaming ($\alpha_e$). 
%For the reasons  discussed above, we decide to call it \emph{$\alpha$--extensionality}.
 
\begin{definition}[Lambda congruence
%,  $\alpha$--extensionality
] \label{alphacond}
Let $\REL$ be a \prela  congruence. We say that $\REL$ is a  \textbf{lambda congruence}
if  it  also satisfies the following condition:
% that we call \textbf{$\alpha$--extensionality}:
 
\begin{itemize}

 \item[($\alpha_e$)] 
$[\lambda \ty \ta] \tx \REL \ta$ implies
  $\lambda \tx\ta\REL \lambda \ty[[\lambda \tx\ta]\ty]$;
  \end{itemize}
   where $\tx$ and $\ty$ are arbitrary variables and $\ta$ is an arbitrary term.
   %      The set of all lambda congruences is denoted
%   by $\cL\cC$.
    \qedef
 \end{definition}
 
%Note that no ancillary concept  is used for  formalizing our condition of $\alpha$--renaming  ($\alpha_e$).
% 
In other words, a lambda congruence is any binary
 relation on terms which simultaneously
 satisfies  all structural and beta conditions together with condition  ($\alpha_e$).
 In particular, every
 lambda congruence is also a prelambda congruence and so the general results
 on independence of variables we proved in the previous section
 can be applied to lambda congruences as well.
%As expected, 
% the converse does not hold
% as we formally show in Proposition \ref{idprop}.

Since our formulation of $\alpha$--renaming as 
  ($\alpha_e$)   seems to be 
new,
we now discuss our condition in relation to other clauses
of $\alpha$--renaming
proposed in the relevant literature, 
as we did in Section \ref{PBC}
for ($\beta_5$). Other conditions of $\alpha$--renaming which are perhaps more standard are considered by Crole \cite[Sec. 3]{CROLE}; however,  they
do not fit  the aims of  the present  work because  their formulations  require either substitution or new   operations such as    atom swapping.
Consider  the following conditions:
    \begin{itemize}
 \item[($\delta_1$)]  
$\lambda \tx\ta\REL \lambda \ty[[\lambda \tx\ta]\ty]$,  provided  $\ty$
 does not occur free in $\ta$;

\item[($\delta_2$)]  $[\lambda \ty\ta]\tz \cv \ta$ implies
$ \lambda \tx \ta  \cv  \lambda \ty   [[\lambda \tx  \ta]\ty ]$,  provided 
 $\ty \neq \tz$;
\item[($\delta_3$)] 
$ \lambda \tx [[\lambda \ty\ta]\tz] \cv  \lambda \ty   [[\lambda \tx  [[\lambda \ty\ta]\tz]]\ty ]$,  provided 
 $\ty \neq \tz$.

 \end{itemize}

Up to minor differences,  ($\delta_1$)
is the condition of $\alpha$--renaming considered by \RV  \cite[$\alpha$ p. 29]{RV},
while conditions ($\delta_2$)
and ($\delta_3$)
have been introduced
by Pigozzi and Salibra
in  the theory of lambda abstraction
algebras  \cite[$\alpha$ p. 12
and Prop. 1.5]{PS}.
In their setting,  ($\delta_2$)
and ($\delta_3$) are  equivalent.
%and this fact
%is used  to prove that lambda abstraction algebras form an algebraic variety.
Our condition ($\alpha_{e}$)
is actually
 a simplified version  of ($\delta_2$)
 inasmuch as ($\alpha_e$) requires
 two variables to be expressed,
 while ($\delta_2$)  three variables
 and one inequality.
 
In the type theories developed by
 Henkin \cite{HENKIN} and Andrews
 \cite{ANDREWS1965, AndrewsBook}, the situation is rather different;  in their setting
  it is present
 a condition which is similar to our extensionality for abstractions.
Roughly speaking, they consider
a condition of  extensionality for
 terms  which have the same  ``arrow type''; this condition 
 in turn implies
  $\alpha$--renaming for some classes of terms, see  for instance \cite[Axiom Schema 6 p. 330 and Thm. 7.21]{HENKIN}. A crucial difference
  is that  their condition of extensionality can be simply expressed as
  a term of their language.
 By contrast,  in our setting extensionality for abstractions is
 a first--order sentence of
 the meta--language.

  Regarding our formulation of $\alpha$--renaming we also observe
  the following fact. 
  
\begin{remark} \label{remalpha}
Let $\cv$ be a prelambda congruence.
Let $\ta$ be a term and let $\tx$ be a variable.
Suppose that $[\lambda \tx \ta]\tx \cv \ta$.
Then, 
from   ($\sL$) we have
$\lambda \tx [[\lambda \tx \ta]\tx] \cv \lambda \tx\ta$ and from ($\sS$)
we obtain 
 $\lambda \tx\ta\REL \lambda \tx[[\lambda \tx\ta]\tx]$.
 This show that condition  ($\alpha_e$)
is already present in any prelambda congruence in the special case
$\tx = \ty$.
Therefore, the real content of ($\alpha_e$) is actually

\begin{itemize}

 \item[($\delta_4$)] 
$[\lambda \ty \ta] \tx \REL \ta$ implies
  $\lambda \tx\ta\REL \lambda \ty[[\lambda \tx\ta]\ty]$, provided $\tx \neq \ty$.
  \end{itemize}
  However, we prefer to express
   our condition of $\alpha$--renaming
   by means of ($\alpha_e$) rather than ($\delta_4$) because the former admits
   a 
shorter formulation.
  \qedef
\end{remark}
 
In order to improve  its intuitive understanding, we may  rewrite,  in the notation employed in Section \ref{PBC}, condition ($\alpha_e$)  in the following form:

\begin{itemize}
 \item[($\alpha_{e}^s$)] 
$\lambda \tx\ta \, = \, \lambda\ty\pl \ty/\tx \pr(\ta)$   \   (provided 
  $\pl \tx/\ty \pr(\ta) \, = \, \ta$).
\end{itemize} 
In this shape,  it appears to be very similar  to the usual conditions of $\alpha$--renaming
which one  finds in the literature.

Contrarily to our treatment of prelambda congruences,  it is not worth giving examples of specific lambda congruences.
The reason is that
in Section \ref{EQU} we shall show
that lambda congruences and
 lambda theories are actually the same concept, and in the literature of the lambda calculus examples
of lambda theories  abound.
But  it is worth giving 
an example of a prelambda
congruence which is not a lambda congruence.

\begin{proposition} \label{idprop}
Let $\cv$ be a lambda congruence. Let $\tx$ and $\ty$ be two variables such that $\tx \neq \ty$. Then, we have $\lambda \tx \tx \cv \lambda \ty\ty$. In particular, prelambda conversion is not a lambda congruence. 
\end{proposition}
\begin{proof}
We have $[\lambda \ty \tx]\tx \cv \tx$
by ($\beta_2$).
Thus, we obtain $\lambda \tx \tx \cv \lambda \ty [[\lambda \tx \tx]\ty]$
by ($\alpha_{e}$).
By ($\beta_1$) it follows that 
$[\lambda \tx \tx]\ty \cv \ty$ and thus
we get $\lambda \ty [[\lambda \tx \tx]\ty] \cv \lambda \ty\ty$ from ($\sL$).
Hence, we obtain $\lambda \tx \tx \cv \lambda \ty \ty$ by ($\sT$).
In particular,  by Corollary \ref{coroconsi} prelambda conversion cannot be a lambda congruence.
\end{proof}

%In order to prove some later results
%we need to prove
%the following lemma, which is the analogous
%of Lemma \ref{beta5rev} for ($\beta_5$).
%% That is to say,  the precondition
%%$[\lambda \ty \ta] \tx \REL \ta$ in
%%($\alpha_e$) can be replaced
%%by a simpler precondition expressed by using the concept of  independence of variables.
%
%
%
% \begin{lemma} \label{alpharev} \label{prop notfree} Let $\cv$ be a lambda congruence. Let $\tx$ and $\ty$ be  variables, and let $\ta$  be  a term such that   $\ty \in \ind{\ta}$. Then, we have $\lambda \tx\ta\cv \lambda \ty[[\lambda \tx\ta]\ty]$.
%\end{lemma}
%\begin{proof}
%Since $\ty \in \ind{\ta}$,  we have
%$[\lambda \ty \ta]\tx \cv \ta$ by Proposition \ref{IND}. Hence, we obtain
% $\lambda \tx\ta\cv \lambda \ty[[\lambda \tx\ta]\ty]$ by means of ($\alpha_{e}$).
%\end{proof}

We now show, in the second part of the following proposition, that condition ($\delta_1$)  above holds  in every lambda congruence. 

\begin{proposition} \label{alpharev} \label{prop notfree} Let $\cv$ be a lambda congruence.
Let $\ta$  be a term,  and
let  $\tx$ and $\ty$ be   variables.
Then, the following statements hold:
\begin{itemize}
\item[\emph{(i)}] Suppose that
$\ty \in \ind{\ta}$. Then, we
have  $\lambda \tx\ta\REL \lambda \ty[[\lambda \tx\ta]\ty]$.

\item[\emph{(ii)}] Suppose that
$\ty \in \NF(\ta)$. Then, we
have  $\lambda \tx\ta\REL \lambda \ty[[\lambda \tx\ta]\ty]$.
\end{itemize}
\end{proposition}
\begin{proof}

 (i)
We have
$[\lambda \ty \ta]\tx \cv \ta$ by Proposition \ref{IND} and we obtain
$\lambda \tx\ta\REL \lambda \ty[[\lambda \tx\ta]\ty]$ from  ($\alpha_e$).
 
 (ii) It follows from (i) above, as  we have
 $\ty \in \ind{\ta}$ by Theorem \ref{propfree}.
\end{proof}

Our next step is to show 
that in every prelambda congruence
conditions ($\alpha_e$)
and extensionality for abstractions
are equivalent.

%
%
%\begin{proposition} \label{funct}
%Let $\cv$ be a pure congruence, and let  $\ta$ and $\tb$ be terms. 
%Let   $\tz$ be a variable such that $\tz \in \, \sim\!\!(\ta)$ and $\tz \in  \, \sim\!\!(\tb)$.
%Furthermore, suppose that $\ta\tz \cv \tb \tz$.
%Then, we have 
%
%\smallskip
%
%{\centering $\ta\tc \cv \tb \tc$ \ for every term $\tc$\enspace.
%\par }
%\end{proposition}
%\begin{proof}
%Let $\tc$ be an arbitrary term.
%By ($\sL$) we obtain
%$\lambda \tz[\ta\tz] \cv \lambda \tz[\tb \tz]$. From ($\sR$) we have $\tc \cv \tc$
%and by ($\sA$) we obtain 
%$[\lambda \tz[\ta\tz]]\tc \cv [\lambda \tz[\tb \tz]]\tc$. 
%
%Let $\tm \in \lb \ta,\tb \rb$.
%We have $[\lambda \tz[\tm\tz]]\tc \cv [[\lambda \tz \tm] \tc]][[\lambda \tz \tz] \tc]$ by ($\beta_3$).
%Since $\tz \in \, \sim\!\!(\tm)$,
%we have $[\lambda \tz \tm] \tc \cv \tm$
%by Proposition \ref{IND}. By ($\beta_1$)
%we obtain $[\lambda \tz \tz] \tc \cv \tc$.
%Thus, by  ($\sA$)
%we have $[[\lambda \tz \tm] \tc]][[\lambda \tz \tz] \tc] \cv \tm \tc$.
%
%Thus, by using ($\sS$)
%we obtain $\ta \tc \cv 
%[[\lambda \tz \ta] \tc]][[\lambda \tz \tz] \tc]
%\cv [\lambda \tz[\ta\tz]]\tc \cv
%[\lambda \tz[\tb\tz]]\tc \cv
%[[\lambda \tz \tb] \tc]][[\lambda \tz \tz] \tc]
%\cv \tb \tc$.
%Finally, from ($\sT$) we conclude $\ta\tc \cv \tb \tc$.
%\end{proof}

\begin{theorem} \label{thmalpha}
Let $\cv$ be a prelambda congruence.
Then,  $\cv$ is a lambda congruence if and only if it satisfies the  condition
of extensionality for abstractions:

\begin{itemize}
\item[\emph{($e_a$)}] 
$[\lambda \tx\ta] \td \cv [\lambda \ty\tb]\td$
   for every term $\td$ \ \  implies \ \ $\lambda \tx\ta \cv \lambda \ty\tb$,
\end{itemize}

\noindent for all terms $\ta$ and $\tb$,    
and all variables $\tx$ and $\ty$.    

\end{theorem}
\begin{proof}
Suppose that $\cv$ is a lambda congruence. Let $\ta$ and $\tb$
be  terms, 
and let $\tx$ and $\ty$ be variables.
Let $\te \eqdef \lambda \tx \ta$
and $\tf \eqdef \lambda \ty \tb$.
Assume $\te\td \cv \tf \td$ for every term $\td$.
We now show that $\te \cv \tf$. 
Let $\tz \in  \ind{\ta}  \cap  \ind{\tb}$. Then,  we have $\te\tz \cv \tf\tz$ and 
 we get
$\lambda \tz[\te\tz] \cv
\lambda \tz[\tf\tz]$
from ($\sL$).
As $\tz \in  \ind{\ta}$
and $\tz \in    \ind{\tb}$
we obtain
$\te \cv \lambda \tz [\te \tz]$ and $\tf \cv \lambda \tz [\tf \tz]$
by Proposition \ref{alpharev}(i). 
By using ($\sS$), we have
$\te \cv \lambda \tz [\te \tz] \cv \lambda \tz[\tf \tz] \cv  \tf$.
By ($\sT$), we  obtain
$\te \cv  \tf$.

Suppose
now that
$\cv$ is a prelambda congruence which satisfies  ($e_a$).
Let $\ta$ be a term, and let $\tx$ and $\ty$ be variables.
 Let $\te \eqdef \lambda \tx \ta$
 and  
 $\tf \eqdef \lambda \ty[\te\ty]$.
  Assume $[\lambda \ty \ta] \tx \REL \ta$. We now prove that $\te \cv \tf$. 
If $\tx = \ty$, then we can derive 
 $\te \cv \tf$ as in Remark \ref{remalpha}.
If $\tx \neq \ty$, then
let
$\tz \in  \ind{\te} \cap  \ind{\tf} \cap (\Var \setminus \lb \tx \rb)$.
Let $\tm \eqdef [\lambda \ty \te]\tz$ and $\tn \eqdef [\lambda \ty \ty]\tz$.
We have $\tf\tz \cv \tm\tn$
by ($\beta_3$). 
As
$\tx \neq \tz$, it follows that $[\lambda \tx \tz]\ty \cv \tz$ by ($\beta_2$).
Hence, we get
$ \tm \cv \lambda \tx[ [\lambda \ty \ta]\tz]$
by ($\beta_5$). 
Since $\tx \neq \ty$, we have $\ty \in \, \sim\!\!(\ta)$ from the assumption  $[\lambda \ty \ta] \tx \REL \ta$. Hence, by Proposition \ref{IND} we obtain $[\lambda \ty \ta]\tz \cv \ta$
and
we get $\lambda \tx[ [\lambda \ty \ta]\tz]
\cv  \te$ from ($\sL$).
Thus, we have
$\tm \cv \te$ from ($\sT$). 
Now, by ($\beta_1$)
we obtain $\tn \cv \tz$
and 
we get 
$\tm\tn \cv
\te\tz$
from ($\sA$). By ($\sT$), we  obtain
$\tf\tz \cv \te\tz$ and so
$\lambda \tz[\te\tz] \cv \lambda \tz[\tf \tz]$ from ($\sS$) and ($\sL$). 
Let $\td$ be an arbitrary term.
 From ($\sR$) we have $\td \cv \td$
and by ($\sA$) we obtain 
$[\lambda \tz[\te\tz]]\td \cv [\lambda \tz[\tf \tz]]\td$.  
Let $\tg \in \lb \te,\tf \rb$.
We have $[\lambda \tz[\tg\tz]]\td \cv [[\lambda \tz \tg] \td]][[\lambda \tz \tz] \td]$ by ($\beta_3$).
Since $\tz \in  \ind{\tg}$,
we have $[\lambda \tz \tg] \td \cv \tg$
by Proposition \ref{IND}. By ($\beta_1$)
we obtain $[\lambda \tz \tz] \td \cv \td$.
Thus, by  ($\sA$)
it follows that $[[\lambda \tz \tg] \td]][[\lambda \tz \tz] \td] \cv \tg \td$
and by ($\sT$) we obtain
$[\lambda \tz[\tg\tz]]\td \cv  \tg\td$.
This shows that we have
$[\lambda \tz[\te\tz]]\td \cv \te \td$
and $[\lambda \tz[\tf\tz]]\td \cv  \tf \td$.
By using ($\sS$)
we obtain $\te \td \cv 
 [\lambda \tz[\te\tz]]\td \cv
[\lambda \tz[\tf\tz]]\td \cv
 \tf \td$
and  $\te\td \cv \tf \td$ from ($\sT$).
Since the term   $\td$ is arbitrary, we have
$\te\td \cv \tf \td$ for every term $\td$.
As $\te$ and $\tf$ are both abstractions, we obtain $\te \cv \tf$ by using ($e_a$).
\end{proof}

Thus, by the previous theorem
our formulation of $\alpha$--renaming ($\alpha_e$) and extensionality for abstractions
($e_a$) are equivalent.
This equivalence is 
significant  precisely because, by Proposition \ref{idprop}, 
($\alpha_e$) does not hold in every prelambda congruence.
We decided to axiomatize
lambda congruences using ($\alpha_e$) because this condition is more compact  
and also because this approach is more akin to traditional
presentations of lambda theories.
But for us, the actual  meaning
of $\alpha$--renaming is precisely
the property of
extensionality for abstractions.

%Finally, we observe that from
%  a proof--theoretic point of view,
% extensionality for abstractions can be consider as
%  a rule which an infinite number of premises --- one for each term $\td$.
%  Thus, the previous theorem also shows that
% this rule does admit a finitary description ---
% in terms of the one--premise rule ($\alpha_e$).

\section{Extensional Congruences} \label{ACF}

In the previous section, we followed the idea  
that only abstractions should
be thought
as functions.
In section we expand our vision and explore the
possibility of thinking  every term
as a function.   

In order to  implement this idea, the
 most natural way is to consider
lambda congruences $\cv$ where
 \emph{each term 
is related to an abstraction}. The typical condition of
$\eta$--extensionality --- that we shall  discuss in  the next section ---  

\smallskip

{\centering $\ta \cv \lambda \ty [\ta \ty]$\, , \quad provided $\ty$ does not occur free in $\ta$\par}

\smallskip

%This condition works because
% the set $\NF(\ta)$ is cofinite
%by Proposition \ref{coffree}. Hence, 
%by Lemma \ref{lemmacof}(iii) 
%for every term $\ta$ there exists
%a variable $\ty$ such that $\ta \cv \lambda \ty [\ta \ty]$. Therefore, 
%every term is related an abstraction,
%as desired.
\noindent
clearly does the required job. In our setting, to achieve
the same  result without introducing any ancillary concept, it suffices
%to 
%find a condition 
%similar to the one above which does not mention free variables.
%At this point we expect that
%the reader has learned the  ``trick'' 
%and shortly realizes
%that
%a possible candidate would be
%\begin{itemize}
%\item[($\epsilon_1$)] \ $[\lambda \ty\ta]\tx \cv \ta$ implies $\ta \cv \lambda \ty [\ta \ty]$, provided $\tx \neq\ty$.
%\end{itemize}
%%\noindent or equivalently
%%\begin{itemize}
%%\item[($\epsilon_2$)] \ $\ty \in \ind{\ta}$ implies $\ta \cv \lambda \ty [\ta \ty]$.
%%\end{itemize}
%
%\noindent 
%While this method is sound, there is  a  simpler
%way to obtain our desiderata;   it suffices
 to consider
\emph{prelambda} congruences
which satisfy  a simpler  condition
of $\eta$--extensionality, our condition ($\eta_e$).

This  lead us to the
formalization of the concept
of \emph{extensional congruence}.

%In Section, \ref{EQU} we will show that they precisely correspond to
%extensional theories.
%For this reason, we refer
%to these  congruence, 
%in which each term can be regarded
%as a function, as extensional congruences.
%In particular, since in every extensional theory $\alpha$--renaming holds,  
%our condition
% has to be quite strong in order to
%to imply $\alpha$--renaming.
%But surprisingly enough,  it is  quite compact and even simpler
%than our ($\alpha_e$).

\begin{definition}[Extensional congruence] \label{derifirst}
Let $\cv$ be a prelambda congruence. We say that $\cv$
is an \textbf{extensional congruence} if
$\cv$ also satisfies the following condition:

\begin{itemize}
   \item[($\eta_e$)] $\ty\cv \lambda \tx[\ty\tx]$, provided $\tx \neq \ty$;
 \end{itemize}

\noindent where  
   $\tx$ and $\ty$ are  arbitrary variables. 
   %   The set of all extensional congruences
%   is denoted by $\cE\cC$.
   \qedef
 \end{definition}

Thus, an extensional congruence is any binary
 relation on terms which simultaneously
 satisfies  all structural and beta conditions together with our condition of $\eta$--extensionality  ($\eta_e$).
Of course, since every
extensional congruence is also a prelambda congruence, the general facts
 on independence of variables we observed in Section \ref{indsec}
 can be applied to extensional congruences as well.

Regarding our formalization, 
our exact formulation of condition ($\eta_e$)
seems to be new, even though
we recognize a strong similarity
with a condition introduced by Hindley and Longo for studying models of the lambda calculus  \cite[Eq. 7  Lem. 4.2]{HL}.
Their condition is, essentially,  \emph{one instance}
of our condition ($\eta_e$).
The main difference is that in their setting
\emph{$\alpha$--renaming is available};
so it is possible  to use $\alpha$--renaming to generate more instances.
On the contrary,
here we are working with arbitrary
prelambda congruences
and we are somehow forced to consider
every combination of variables in our ($\eta_e$). But the actual power of
 ($\eta_e$) is quite remarkable:   it allows us to derive our condition of $\alpha$--renaming  ($\alpha_e$) in every extensional congruence,  as we show in Theorem \ref{alpha}.

We now survey some conditions of $\eta$--extensionality
considered in the literature
which   do not mention any ancillary concept:

\begin{itemize}
\item[($\epsilon_1$)] 
 $\lambda \tx[ [[\lambda \tx \ta] \ty]\tx]\cv [\lambda \tx \ta]\ty$, provided $\tx \neq \ty$;
  \item[($\epsilon_2$)]  
  $\ty\cv \lambda \tx \ta$ for some term $\ta$, provided
     $\tx \neq \ty$;
 \item[($\epsilon_3$)] 
 $\mathbf{i} \cv \mathbf{1}$.

% \item[($\epsilon_6$)] \
% $\ta \td \cv \tb\td$
%   for every term $\td$ \quad implies \quad $\ta \cv \tb$.
\end{itemize}

Up to inessential differences,  conditions ($\epsilon_1$) and ($\epsilon_2$) are due to Salibra;
the former  is specifically employed in 
the axiomatization of the concept of extensional lambda abstraction algebra, while the latter
is provably  equivalent to ($\epsilon_1$) \cite[Def. 57 and Prop. 58]{SAL}. 
Our condition ($\eta_e$)
clearly resembles ($\epsilon_2$). 

Condition ($\epsilon_3$) is often employed in
model theoretical studies of the lambda calculus, again see
\cite[Lem. 4.2]{HL} and
\cite[Prop. 58]{SAL}. Here 
 $\mathbf{i}$ and  $\mathbf{1}$ respectively denote the terms $\lambda \ty \ty$ and $ \lambda \ty [\lambda \tx [\ty \tx]]$ for $\tx \neq \ty$.   
Again,
  ($\epsilon_3$) 
    is specifically
  designed to work 
  in settings where $\alpha$--renaming
  is available.

As already said,  in Theorem \ref{alpha} we show  that every extensional congruence is a lambda congruence.
Thus, there is no reason to explicitly include
 condition ($\alpha_{e}$)
in the axiomatization.
The  property that $\alpha$--renaming can
 be derived from $\eta$--extensionality is not so well--known
 but it
 has been already noticed 
  for the simply typed lambda calculus
 by Do\v{s}en and Petri\'c \cite[Sec. 3]{DP}.
Indeed, 
we have taken   account of their idea 
in our  formalization of  extensional congruence.

Since
in Section \ref{EQU} we shall show
that extensional congruences precisely correspond
to extensional theories, 
in the present section  
we do not give concrete examples
of extensional congruences; rather,    we concentrate on establishing
 the technical facts we need to show
the aforementioned equivalence
and on explaining the real significance of $\eta$--extensionality
in our setting.

%
% We are now in a position to show
% some results on extensional congruences.
The first result we present is
given in the second part of following proposition, where
we show that the traditional formulation of $\eta$--extensionality holds
 every  extensional congruence.

\begin{proposition} \label{propeta} \label{prop notfree} Let $\cv$ be an extensional congruence.
Let $\ta$  be a term  and
let  $\ty$ be  a  variable.
Then, the following statements hold:
\begin{itemize}
\item[\emph{(i)}] Suppose that
$\ty \in \ind{\ta}$. Then, we
have  $\ta\cv \lambda \ty[\ta\ty]$.

\item[\emph{(ii)}] Suppose that
$\ty \in \NF(\ta)$. Then, we
have  $\ta\cv \lambda \ty[\ta\ty]$.
\end{itemize}
\end{proposition}
\begin{proof} 
(i)
Let  $\tx$  be a    variable such that $\tx \neq \ty$.
 Let $\tb \eqdef [[\lambda\tx  \tx] \ta] [[\lambda\tx  \ty] \ta]$, $\tc \eqdef [\lambda \tx [\tx\ty]]\ta$
and $\td \eqdef \lambda \tx  [\lambda \ty[\tx\ty]]$.
From ($\beta_3$), we have
$\tc\cv \tb$.
From ($\beta_1$)
and ($\beta_2$) we
 obtain $[\lambda\tx  \tx] \ta \cv \ta$ and $[\lambda\tx  \ty] \ta \cv \ty$, and
 it follows from  $(\sA)$ that
$\tb \cv \ta\ty$.
By ($\sT$) we have  $\tc \cv \ta\ty$,
and so we obtain  $\lambda \ty  \tc\cv \lambda \ty  [\ta\ty]$   from  ($\sL$). 
Now, as $\tx \neq \ty$, we  have $\tx \cv \lambda \ty[\tx\ty]$ from  ($\eta_e$) and
 $\lambda \tx\tx \cv \td$  from ($\sL$).
  By
 ($\sR$), we get $\ta \cv \ta$. So, by ($\sA$) we have
$[\lambda \tx\tx]\ta \cv\td\ta$.
Since $\tx \neq \ty$ and   $\ty \in \ind{\ta}$, we   obtain 
$\td\ta \cv  \lambda \ty  \tc$ by Proposition \ref{beta5rev}(i).
By using  ($\sS$), we  have $\ta \cv [\lambda \tx\tx]\ta \cv
\td\ta \cv \lambda \ty \tc  \cv  \lambda \ty  [\ta\ty]$,  and we get
$\ta \cv  \lambda \ty  [\ta\ty]$
from ($\sT$).

 (ii) It follows from (i) above,
 as we have
 $\ty \in \ind{\ta}$ by Theorem \ref{propfree}. 
\end{proof}

 In order to achieve all the goals
we set in the beginning of this section, it only remains to prove 
 that ($\alpha_e$)
holds in every extensional congruence.

%
%\begin{proposition} \label{propalph}
% Let $\cv$ be an extensional conversion.
%Let $\ta$ be a  term and let
% $\tx$ and $\ty$ be variables
% such that
%  $\tx \in \, \sim\!\!(\ta)$. Then, we have
%  $\lambda \ty\ta\cv \lambda \tx[[\lambda \ty\ta]\tx]$.
%\end{proposition}
%\begin{proof}
%
%\end{proof}

\begin{theorem} \label{alpha}
Every extensional congruence is a lambda congruence.
\end{theorem}

\begin{proof}
 Let $\cv$ be an extensional congruence.
Let $\ta$ be a  term, and let
 $\tx$ and  $\ty$ be variables. 
 Let $\tb \eqdef \lambda \tx \ta$.
 We have to show that
$[\lambda \ty \ta] \tx \REL \ta$ implies
  $\tb\REL \lambda \ty[ \tb\ty]$.
  Assume $[\lambda \ty \ta] \tx \REL \ta$.
If $\tx = \ty$, then we proceed as in Remark \ref{remalpha}.
If $\tx \neq \ty$, then we have 
 $\ty \in \ind{\ta}$ and we obtain
 $\ty \in \ind{\tb}$ by Lemma \ref{lemmalpha}(iii). By Proposition  \ref{propeta}(i) we 
obtain   $\tb\REL \lambda \ty[\tb\ty]$.
\end{proof}

In the following corollary
we show that the set of all  extensional congruences is properly included
in the set of all prelambda congruences.

\begin{corollary} \label{cornew}
There exists a prelambda congruence which is not an extensional congruence.
\end{corollary}
\begin{proof}
 By Proposition \ref{idprop}, prelambda conversion  is not a lambda congruence
 and hence it cannot be  an extensional congruence  by Proposition \ref{alpha}.
\end{proof}

Regarding the real significance of $\eta$--extensionality
in our setting
we now consider the following condition that  we call  \emph{extensionality for terms}:

\smallskip

{\centering
$\ta\td
\cv \tb\td$ for every term $\td$ \quad implies \quad  $\ta \cv  \tb$\enspace. \par}

\smallskip

\noindent   
This form of extensionality  has been considered
by many authors: for instance, 
 in combinatory logic  by Rosenbloom \cite[ p. 112]{RSB} 
and in the lambda calculus  
by  Wadsworth \cite[p. 493]{WAD},
 Hindley and Longo \cite[p. 297]{HL}
and  Hindley and Seldin \cite[p. 77]{HS}. 
While each condition of $\eta$--extensionality
  so far considered  has a strong syntactical and artificial flavour and no apparent
connection with any property of extensionality, 
 the property of extensionality for terms above
has a clear
  mathematical significance:
two terms should be regarded as the same if they have the same 
 input--output behaviour.  
We now prove
that in every prelambda congruence  extensionality for terms
 is  equivalent
to ($\eta_e$).
%In particular, note that
 %$\alpha$--renaming is not necessary
 %to establish
  %the equivalence.

  \begin{theorem} \label{propeqext2}
  Let $\cv$ be a prelambda congruence.
Then, $\cv$ is an extensional congruence if and only if
it  satisfies the  condition
of extensionality for terms:
\begin{itemize}
\item[\emph{($e_t$)}]
 $\ta \td \cv \tb\td$
   for every term $\td$ \quad implies \quad $\ta \cv \tb$,
   \end{itemize}
   \noindent for all terms $\ta$ and $\tb$.

 \end{theorem}
 \begin{proof}
  Suppose that $\cv$ is an extensional congruence. We now show that $\cv$ satisfies ($e_t$).
  For, we assume $\ta \td \cv \tb\td$
   for every term $\td$ and prove that $\ta \cv \tb$.\\
Let
$\tz \in   \ind{\ta} \cap \ind{\tb}$. We have $\ta \tz \cv \tb \tz$ and we get
$\lambda \tz[\ta \tz] \cv \lambda \tz[\tb \tz]$ by ($\sL$).
By Proposition \ref{propeta}(i),  we get
$\ta \cv \lambda \tz [\ta\tz]$  and $\tb \cv \lambda \tz [\tb\tz]$. 
By using ($\sS$) we have $\ta  \cv \lambda \tz [\ta\tz]
\cv \lambda \tz[\tb \tz] \cv \tb$
and we  get $\ta  \cv \tb$ from ($\sT$).

      Suppose now that $\cv$ is a prelambda congruence which satisfies ($e_t$).
       Let $\tx$ and $\ty$ be two variables such that 
$\tx \neq \ty$.
      We now prove that
$\ty \cv \lambda \tx[\ty\tx]$. 
Let
 $\td$ be an arbitrary
term. 
We have $[\lambda \tx[\ty\tx]]\td\cv
 [[\lambda \tx\ty]\td][[\lambda \tx\tx]\td]$ by ($\beta_3$).
By ($\beta_2$) and ($\beta_1$)
we get $[\lambda \tx\ty]\td \cv \ty$
and $[\lambda \tx\tx]\td \cv \td$. So, we obtain
$[[\lambda \tx\ty]\td][[\lambda \tx\tx]\td] \cv
\ty \td$ by ($\sA$). 
By using ($\sS$)
we have
$\ty\td \cv [[\lambda \tx\ty]\td][[\lambda \tx\tx]\td]
\cv [\lambda \tx[\ty\tx]]\td$
and by using ($\sT$) we obtain
$\ty\td 
\cv [\lambda \tx[\ty\tx]]\td$. 
This shows that we have
$\ty\td \cv [\lambda \tx[\ty\tx]]\td$
for every term $\td$. 
Hence,
we can apply ($e_t$)
to  obtain $\ty \cv \lambda \tx[\ty\tx]$.
  \end{proof}
 
Again, observe that this equivalence is 
significant  because, by Corollary \ref{cornew}, 
$(\eta_e)$ does not hold in every prelambda congruence.

We can now give a proper  conclusion
to
our  informal discussion
of terms as functions. 
Recall that
 for us a class of
  terms intended to play the role of functions 
   must satisfy the
condition of extensionality  discussed in the Section \ref{lambdacong}, where we examined  the case of abstractions
as functions. 
By Theorem \ref{propeqext2},  
in every extensional congruence we can think
every term
  as  a function and consider the condition  ($\eta_e$) as  a compact and equivalent way to express  the  property of
  extensionality for terms --- 
  this also gives a proper justification for  the terminology   ``$\eta$--extensionality'' employed for our condition ($\eta_e$).

%Recall that for us  the
%condition of extensionality  discussed in the previous section %is a selective criterion for determine if the members of a class of 
%terms can be regarded as functions. By Theorem \ref{propeqext2}, in every extensional congruence not only abstractions but the whole 
%the class of all terms satisfies extensionality.

%The reason why
%$\alpha$--renaming
%is a consequence of $\eta$--extesionality  should  also be very clear now; if we assume
%the latter we obtain
%extensionality for terms
%which clearly implies extensionality for abstractions which in turn implies the former. 
%Finally, we note that a remarkably powerful property like
%extensionality for terms
%admits a 
%simple and finitary equivalent formulation,  our rule ($\eta_e$).

\section{Lambda and Extensional Theories} \label{STA}

Our next step is to prove that
lambda and extensional congruences precisely 
correspond to lambda and extensional theories, respectively.
In this section, we recall the 
traditional definitions of these theories
and prove some preliminary properties  that we need to
prove our results.

Our definitions of lambda and extensional theories
are perfectly in line with the ones
which can be found in the literature;
see, \eg Meyer \cite[p. 92]{MEYER}.
%  but we also note
%that the definition  given by Barendregt \cite[Def. 4.1.1]{BAR} is slightly different  --- there, for instance, $\Lam^2$ is not a  lambda theory.

In order to provide the  axiomatizations,  we need to formalize the notion of \emph{substitution} first.
To this aim, we 
 point out  that  there exist several different   definitions of substitution in the literature. 
A typical approach is to proceed as in the following remark. 

\begin{remark} \label{remcurry}
The definition of substitution which is  often found in the literature is the one given by  Hindley and Seldin \cite[Def. 1.12]{HS}. It is actually a variation of 
the one proposed by Curry and Feys \cite[p. 94]{CF} and
it is commonly called  \emph{capture--free substitution}.

We now recall its definition in order to discuss and emphasize
some of its aspects. Expressed in our notation, it can be formulated as follows:

\begin{itemize}
\item[(1)]  $\{\td/\tx\}(\tx) \eqdef \td$;
\item[(2)]  $\{\td/\tx\}(\ty) \eqdef \ty$, if $\tx \neq \ty$;
\item[(3)]  $\{\td/\tx\}(\ta\tb) \eqdef \{\td/\tx\}(\ta)\{\td/\tx\}(\tb)$;
\item[(4)]  $\{\td/\tx\}(\lambda \tx\ta) \eqdef \lambda \tx\ta$;
\item[(5)]  $\{\td/\tx\}(\lambda \ty\ta) \eqdef \lambda \ty\ta$, \\ if $\tx \neq \ty$ and $\tx$ does not occur free in  $\ta$;
\item[(6)]  $\{\td/\tx\}(\lambda \ty\ta) \eqdef \lambda \ty\{\td/\tx\}(\ta)$, \\ if $\tx \neq \ty$, $\tx$
occurs free in $\ta$ and $\ty$ does not 
 occur free in $\td$;
\item[(7)]  $\{\td/\tx\}(\lambda \ty\ta) \eqdef \lambda \tz\{\td/\tx\}(\{\tz/\ty\}(\ta))$, \\ if $\tx \neq \ty$, $\tx$ occurs free in $\ta$ and 
 $\ty$ occurs free in $\td$,  where $\tz$ is the first variable in which does not occur
  free in both   $\ta$ and $\td $.
\end{itemize}

Despite its complicated axiomatization, substitution as defined above
 has the following pleasant property:
 the usual $\beta$--rule of the lambda calculus can be simply expressed
as $[\lambda \tx \ta]\td \cv\{\td/\tx\}(\ta)$, without any restriction.
However, there are also some
 disadvantages in considering the above
 formalization.

Firstly,  the definition above is \emph{not} given
by induction on the structure
of terms. 
The reason is that in condition (7)
above there is $\{\tz/\ty\}(\ta)$
and not $\ta$  in the right hand side $\lambda \tz\{\td/\tx\}(\{\tz/\ty\}(\ta))$ of $\eqdef$.
In fact, the previous definition
is by induction on the \emph{size} of terms. So, in order to define substitution
in this way, we need first a \emph{definition} of
the concept of size of a term.
Furthermore, we need to \emph{prove} that the previous construction is well--defined and to do so, it is also necessary to \emph{prove}
that $\{\tz/\ty\}(\ta)$ and $\ta$
have the same size.

Secondly,  in condition (7) above,
it is assumed that the set of 
 variables is equipped with an appropriate ordering.
 As already said  in Section \ref{lambda sec}, we do not want to  assume this inasmuch as there are
 other ways to define substitution.
  \qedef

\end{remark}

In this article,
   we  follow  the simpler and more elegant approach to substitution
developed by Barendregt in his dissertation
 \cite{BARTH}.
(It should not be confused
with the more familiar one  developed by  Barendregt in his  classic book \cite{BAR}.)

\begin{definition}[Substitution] \label{subdef}
Let $\ta$ and $\td$ be  terms, and let $\tx$  be a variable. By induction on the structure of $\ta$, 
we define the term $\vv\td/\tx\ww(\ta)$ as follows:

\begin{itemize}
\item[(S$_1$)] 
 $\vv\td/\tx\ww(\tx)  \eqdef  \td$;
\item[(S$_2$)] \ $\vv\td/\tx\ww(\ty) \eqdef  \ty $,   if $ \tx \neq \ty$;
\item[(S$_3$)]  $
\vv\td/\tx\ww(\tb\tc)  \eqdef  \vv\td/\tx\ww(\tb)\vv\td/\tx\ww(\tc)$;
\item[(S$_4$)]  $
\vv\td/\tx\ww(\lambda \tx   \tb)  \eqdef  \lambda \tx   \tb$;
\item[(S$_5$)]  $
\vv\td/\tx\ww(\lambda \ty   \tb)  \eqdef  \lambda \ty   \vv\td/\tx\ww(\tb)$,  if $ \tx \neq \ty$.
\end{itemize}
We call the term $\vv\td/\tx\ww(\ta)$  \textbf{the result of the  substitution of $\tx$ for $\td$ in $\ta$}.
\qedef
\end{definition}

%We now define the concepts of \emph{safety} and \emph{capture--avoiding substitution}.

In order to define
the concepts of lambda and extensional theory
we now formally define the notion of
 variable occurring bound 
in a term. However, 
for the sake of conformity with our formal treatment of free variables, 
we prefer to 
 define the set of variables which do \emph{not}  occur bound in a given term instead.

\begin{definition}[Non--bound variable] \label{nonbound} Let $\ta$ be a term.
We define the  cofinite set of variables  $\NB(\ta)$  by induction on the structure of $\ta$  as follows:
\begin{itemize} 
\item $\NB(\tx)  \eqdef \Var$;
\item  $\NB(\tb\tc)  \eqdef \NB(\tb) \cap \NB(\tc)$;
\item $\NB(\lambda \tx   \tb) \eqdef \NB(\tb) \setminus \lb \tx \rb$.
\end{itemize}
%
%\begin{itemize}
%\item[($\NB_1$)] \ $\NB(\tx)  \eqdef \Var$;
%\item[($\NB_2$)] \ $\NB(\tb\tc)  \eqdef \NB(\tb) \cap \NB(\tc)$;
%\item[($\NB_3$)]  \  $\NB(\lambda \tx   \tb) \eqdef \NB(\tb) \setminus \lb \tx \rb$.
%\end{itemize}
\noindent 
%Let $\tx$ be a variable.
 We  also say that a variable \textbf{$\tx$ does not occur bound in $\ta$} if $\tx \in \NB(\ta)$.
  \qedef
\end{definition}

%In analogy with the case of
%non--free variables,
%we now show that the set of variables which do not  occur bound
%in a gives term
%is always a cofinite set.
%
%
%
%\begin{proposition} \label{nonbo}
% Let $\ta$ be a term.
%Then, 
%$\NB(\ta)$ is a cofinite set.
%\end{proposition}
%\begin{proof}
%The proof is by induction on the structure of $\ta$.
%
% If $\ta = \aaa$, then  $\NB(\ta) = \Var$ is a cofinite set. 
% 
%
% 
%If $\ta = \tb\tc$, then 
%by inductive hypothesis it follows that
%$\NB(\tb)$ and $\NB(\tc)$
%are cofinite.
%By Lemma  \ref{lemmacof}(i),
%$\NB(\tb)  \cap   \NB(\tc) = \NB(\tb\tc)$
%is cofinite. 
%
%Finally, if $\ta = \lambda \tx \tb$, then
%by inductive hypothesis it follows that
%$\NB(\tb)$ is cofinite, say $\Var \setminus F$ for some finite set of variables $F$. Then, we have
%$ \NB(\lambda \tx\tb) = \NB(\tb)
%\setminus \lb \tx \rb = (\Var \setminus F) \setminus \lb \tx \rb = \Var \setminus (F \cup \lb \tx \rb)$. Since $F \cup \lb \tx \rb$ is finite  we conclude that $\NB(\lambda \tx \tb)$ is cofinite.
% \end{proof}

Having now properly defined
the necessary ancillary concepts,
we are finally ready
to define the notions
of lambda and extensional theory.

\begin{definition}[Lambda and extensional theory] \label{deriR}
Let $\REL$ be a congruence. We say that $\REL$
is a \textbf{lambda theory} if
it also satisfies the following conditions:
 
\begin{itemize}
      \item[($\beta$)]   $[\lambda \tx \ta]\td \cv\vv\td/\tx\ww(\ta)$, provided $\NB(\ta) \cup \NF(\td) = \Var$;   
      \item[($\alpha$)] 
      $\lambda \tx \ta\cv \lambda \ty \vv\ty/\tx\ww(\ta)$, provided   $\ty \in \NF(\ta) \cap \NB(\ta)$;
\end{itemize}
 where $\ta$ and $\td$ are arbitrary terms, and 
   $\tx$ and $\ty$ are arbitrary variables. 
%   
%   The set of all  lambda  theories is denoted by $\cL\cT$.
   
   Let $\REL$ be a lambda theory. We say that $\REL$
is an \textbf{extensional theory} if
it also satisfies the following condition:
 
\begin{itemize}

      \item[($\eta$)] \ $ \ta\cv \lambda \ty[\ta\ty]$, provided $\ty \in \NF(\ta)$;
\end{itemize}
 where $\ta$ is an arbitrary term and $\ty$ is an arbitrary variable. 
%   
%   The set of all  extensional theories is denoted by $\cE\cT$.
   \qedef
 \end{definition}

   Thus, a lambda theory is any binary
 relation on terms which simultaneously
 satisfies  all structural conditions
and
conditions ($\beta$) and ($\alpha$) given above.
If, in addition, the relation is also closed under
($\eta$), then it is an  extensional theory.

The most important examples of lambda and extensional theories are \emph{lambda conversion}
and \emph{extensional conversion}.

\begin{definition}[Lambda
and extensional conversion] \label{laco}
We call \textbf{lambda conversion}
 the  lambda theory 
inductively defined
 by    structural  
 conditions, ($\beta$)
 and ($\alpha$). We also define \textbf{extensional conversion} as the extensional theory inductively defined  by  structural conditions, 
 ($\beta$),
 ($\alpha$) and ($\eta$).
 \qedef
 \end{definition}

 Let us now discuss the conditions of Definition \ref{deriR} in some detail.
 %Let $\cv$ be an extensional theory.

% Conditions  ($\sR$), ($\sS$), ($\sT$), ($\sL$) and ($\sA$)
% appear --- sometimes in equivalent form --- in any formalization of extensional conversion. 

Conditions  ($\beta$) and ($\alpha$)
are  as in Barendregt \cite[p. 4]{BARTH}.  
Note that  ($\beta$)
presents a restriction on its
applicability,  namely
 $\NB(\ta) \cup \NF(\td) = \Var$.
 In the original formulation,
 this restriction is equivalently expressed follows:
 bound variables
 of $\ta$ and free variables of $\td$
 are disjoint sets. 
  In this paper, instead of putting  complications
directly inside the definition of substitution (\cf   the capture--free substitution of Remark \ref{remcurry}), 
we prefer to have a simple notion of substitution at the price of this restriction. 
In the next example
we can see how the proviso  $\NB(\ta) \cup \NF(\td) = \Var$ works.

\begin{example}
Let $\cv$ be a lambda theory.
Let $\tx$ and $\ty$ be variables such that $\tx \neq \ty$. We have $\vv\ty/\tx\ww(\lambda \ty   \tx)
= \lambda \ty   \vv\ty/\tx\ww(\tx)
= \lambda \ty   \ty$. However,
we \emph{cannot} infer $[\lambda \tx   [\lambda \ty   \tx]] \ty \cv \lambda \ty   \ty$ directly from  ($\beta$), as we have
 $\NB(\lambda \ty   \tx) \cup \NF(\ty) = \big(\NB(\tx) \setminus \{\ty\}\big) \cup \big(\Var \setminus \{\ty\}\big) =  \big(\Var \setminus \{\ty\}\big) \cup\big( \Var \setminus \{\ty\}\big) = \Var \setminus \lb \ty\rb \neq \Var$. \qedef
\end{example}

As for
 ($\alpha$),
when compared with other formalizations
of $\alpha$--renaming that one usually finds in the literature,
the main  difference lies in
 the restriction
 $\ty \in \NF(\ta) \cap \NB(\ta)$ which is not just the ``usual'' proviso
  $\ty \in \NF(\ta)$ (that is, $\ty$ does not occur free in $\ta$).
% The reason is that our notion
% of substitution is the one
% given in Definition \ref{subdef}
% and not the one  informally described  in Remark
% \ref{remcurry}.
 The next example
 explains the situation.
 
 \begin{example}
Let $\cv$ be a lambda  theory.
Let $\tx$ and $\ty$ be variables such that $\tx \neq \ty$. As
before, we have $\vv\ty/\tx\ww(\lambda \ty   \tx)
= \lambda \ty   \ty$. 
From this, it follows that 
$\lambda \ty\vv\ty/\tx\ww(\lambda \ty   \tx)
= \lambda \ty[\lambda \ty   \ty]$.
Since $\tx \neq \ty$, we have $\ty \in \Var \setminus \lb \tx \rb = (\Var \setminus \lb \tx \rb) \cup \lb \ty \rb =  \NF(\lambda \ty \tx)$.
However,
we \emph{cannot} infer $ \lambda \tx[\lambda \ty \tx]\cv\lambda \ty   [\lambda \ty \ty]$ directly from  ($\alpha$), as we have
 $\ty \notin \Var \setminus \lb \ty\rb
 = \NB(\tx) \setminus \lb \ty\rb 
 =  \NB(\lambda \ty \tx)$ and in particular $\ty \notin \NF(\lambda \ty \tx) \cap \NB(\lambda \ty \tx)$. 
 
Notice that if we use the capture--free
  substitution
 of Remark
 \ref{remcurry}, then, 
 in order to calculate
 $\{\ty/\tx\}(\lambda \ty   \tx)$,
 we are forced to use condition (7) as
 we have
 $\tx \neq \ty$, $\tx \notin \NF(\lambda \ty \tx)$ and $\ty \notin \NF(\ty)$. In this case, we obtain
 $\{\ty/\tx\}(\lambda \ty   \tx)
= \lambda \tz   \{\ty/\tx\}(\{\tz/\ty\}\tx)) = \lambda \tz   \{\ty/\tx\}(\tx) = \lambda \tz \ty$
where $\tz$ is the first variable
in  $\Var \setminus \lb \tx,\ty \rb =\NF(\tx)
\cap \NF(\ty)$.
From this, it follows that
$\lambda \ty\{\ty/\tx\}(\lambda \ty   \tx)
= \lambda \ty[\lambda \tz   \tx]$
and
we can infer 
$\lambda \tx[\lambda \ty \tx] \cv \lambda \ty   [\lambda \tz \ty]$
by using ($\alpha$)
with the  ``usual'' restriction $\ty \in \NF(\ta)$.
 \qedef
\end{example}

Finally, we note that condition ($\eta$)
is exactly as in the literature.

 As already pointed out,  Do\v{s}en and Petri\'c \cite{DP}  observed that
 it is not necessary to include $\alpha$--renaming in the axiomatization extensional theories. Indeed,  
we  followed their idea
in our  formalization of  extensional \emph{congruences}.
However, we do not take their observation into account  in the present  context  because  traditional
formulations 
do include an explicit condition of $\alpha$--renaming.
Also, even if  condition ($\alpha$) were removed from the definition, 
other ancillary concepts would  still be  present.

%Our next task is to show that the notions of lambda congruence and
%extensional congruence are equivalent
%to the concepts of 
%lambda theory and extensional theory,
%respectively.
 In order to reach our goals,
  we now  prove a series of  preliminary results.  
  
%The reader familiar with syntactical proofs
%of cut--elimination in logic   will recognize a similar pattern here: to show that some concepts are not necessary --- in  the  case of this and next section the notions of substitution, free variable and bound  variable --- we have to carefully analyze them.

First, we show that every lambda theory
satisfies the condition of extensionality for abstractions
discussed in Section \ref{lambdacong}.

\begin{proposition} \label{propalfa}
Let $\cv$ be a lambda theory.
Let $\lambda \tx\ta$ and $\lambda \ty \tb$ be terms and suppose
that $[\lambda \tx \ta]\td \cv [\lambda \ty \tb]\td$ for every term $\td$.
Then, we have $\lambda \tx \ta\cv \lambda \ty \tb$.
\end{proposition}
\begin{proof}
Let
$\tz \in  \NF(\lambda \tx \ta) \cap
\NF(\lambda \ty \tb) \cap
\NB(\lambda \tx \ta) \cap
\NB(\lambda \ty \tb)$.
 Then, we have
 $[\lambda \tx \ta]\tz \cv [\lambda \ty \tb]\tz$. Since $\NF(\tz) = \Var \setminus \lb \tz \rb$, $\tz \in \NB(\lambda \tx \ta)$
 and $\tz \in \NB(\lambda \ty \tb)$,
 we have  $\NB(\lambda \tx \ta) \cup \NF(\tz)
 = \Var$
 and $\NB(\lambda \ty \tb) \cup \NF(\tz)
 = \Var$.
So, we can apply ($\beta$) and we obtain
 $[\lambda \tx \ta]\tz \cv \vv \tz /\tx \ww(\ta)$ and 
  $[\lambda \ty \tb]\tz \cv \vv \tz /\ty \ww(\tb)$.
  By using ($\sS$)
we obtain $ \vv \tz /\tx \ww(\ta) \cv
[\lambda \tx \ta]\tz \cv [\lambda \ty \tb]\tz \cv \vv \tz /\ty \ww(\tb)$
and so $ \lambda \tz\vv \tz /\tx \ww(\ta) \cv
\lambda \tz \vv \tz /\ty \ww(\tb)$ from ($\sT$) and ($\sL$).
Now, since $\tz \in \NB(\lambda \tx \ta)
= \NB(\ta) \setminus \lb \tx \rb$
 we obtain 
 $\tz \in \NB(\ta)$ and $\tx \neq \tz$.
 Similarly,
as $\tz \in \NB(\lambda \ty \tb)
= \NB(\tb) \setminus \lb \ty \rb$
 we have 
 $\tz \in \NB(\tb)$ and $\ty \neq \tz$.
  As  $\tz \in \NF(\lambda \tx \ta) = \NF(\ta)\cup \lb \tx\rb$
and $\tx \neq \tz$, it follows that
$\tz \in \NF(\ta)$. Analogously, as
$\tz \in \NF(\lambda \ty \tb) = \NF(\tb)\cup \lb \ty\rb$
and $\ty \neq \tz$, we obtain
$\tz \in \NF(\tb)$.
Hence, we can now apply ($\alpha$)
to obtain $\lambda \tx \ta \cv \lambda \tz \vv \tz/\tx\ww(\ta)$ and 
 $\lambda \ty \tb \cv \lambda \tz \vv \tz/\ty\ww(\tb)$. 
By using ($\sS$) we  get
$\lambda \tx \ta \cv \lambda \tz\vv \tz /\tx \ww(\ta) \cv
\lambda \tz \vv \tz /\ty \ww(\tb) \cv \lambda \ty\tb$
and finally $\lambda \tx \ta \cv \lambda \ty\tb$ from ($\sT$).
\end{proof}

%\begin{remarks}
%Let $\ta$ be a term. Since both $FV(\ta)$ and %$BV(\ta)$ are finite sets of variables
%and $\Var$ is infinite, it readily follows that there
%always exists a variable which is new for $\ta$.
%In the sequel we use this observation implicitly.
%\end{remarks}

We now establish other technical properties. 

\begin{lemma} \label{lemmatec1}
Let $\ta$   be a   term, and let $\ty$ and  $\tw$ be    variables. Then, we have  
 $\NB(\vv \tw  / \ty \ww(\ta))  = \NB(\ta)$.
\end{lemma}
\begin{proof}

We reason by induction on the structure of $\ta$.

Suppose that $\ta = \tz$. If $\ty = \tz$,
then we have $\NB(\vv \tw  / \ty \ww(\ty))
= \NB(\tw) = \Var = \NB(\ty)$.
If $\ty \neq \tz$,
then we have $\NB(\vv \tw  / \ty \ww(\tz))
= \NB(\tz)$.

Suppose that $\ta = \tb\tc$.
By inductive hypothesis, we have
 $\NB(\vv \tw  / \ty \ww(\tb))  = \NB(\tb)$
 and  $\NB(\vv \tw  / \ty \ww(\tc))  = \NB(\tc)$.
 Then, it follows that
  $\NB(\vv \tw  / \ty \ww(\tb\tc))  =  \NB(\vv \tw  / \ty \ww(\tb) \vv \tw  / \ty \ww(\tc)) =
  \NB(\vv \tw  / \ty \ww(\tb))  \cap \NB(\vv \tw  / \ty \ww(\tc)) =
   \NB(\tb) \cap \NB(\tc) = \NB(\tb\tc)$.
   
   Finally, suppose that
 $\ta = \lambda \tz\tb$.
 If $\ty = \tz$, then we have $\NB(\vv \tw / \ty \ww(\lambda \ty \tb)) = \NB(\lambda \ty \tb)$.
 If $\ty \neq \tz$, then, 
by inductive hypothesis, we have
 $\NB(\vv \tw  / \ty \ww(\tb))  = \NB(\tb)$.
 Hence, it follows that
$\NB(\vv \tw / \ty \ww(\lambda \tz \tb)) = \NB(\lambda \tz\vv \tw / \ty \ww( \tb)) =  \NB(\vv \tw / \ty \ww (\tb)) \setminus \lb \tz \rb
= \NB(\tb) \setminus \lb \tz \rb
= \NB(\lambda \tz \tb)$.
\end{proof}

\begin{proposition} \label{renam}
Let $\cv$ be a lambda theory, and
let $\ta$ and $\td$ be terms. Then, there exists a term $\tg$ such that
 $\tg \cv \ta$, 
and  $\NB(\tg) \cup \NF(\td) =\Var$.

\end{proposition}
\begin{proof}
The proof is by induction on the structure of $\ta$.

 Suppose that $\ta = \tx$.  Let $\tg \eqdef \tx$.
Then, we have $\tg \cv \tx$ by ($\sR$).  Since $\NB(\tg) = \Var$, we  have
  $\NB(\tg) \cup \NF(\td) =\Var$.

 Suppose that $\ta = \tb\tc$. By inductive hypothesis there exist $\tm$ and $\tn$ such  that
$\tm  \cv  \tb$ and  $\NB(\tm) \cup \NF(\td) =\Var$,  as well as $\tn \cv \tc$ and $\NB(\tn) \cup \NF(\td) =\Var$.
Let $\tg  \eqdef \tm\tn$.  
Then, we have $\tg \cv \tb\tc$ by ($\sA$) and
$
 \NB(\tg) \cup \NF(\td)  = \big( \NB(\tm) \cap \NB(\tn) \big) \cup \NF(\td) 
= \big(\NB(\tm) \cup \NF(\td)\big) \cap  \big(\NB(\tn) \cup \NF(\td)\big) 
 =  \Var \cap\Var  =  \Var$.

 Finally, suppose that $\ta = \lambda \ty   \tb$.
By inductive hypothesis, there exists
a term $\tm$ such that $\tm \cv \tb$ and $\NB(\tm) \cup \NF(\td) =\Var$.
Let $\tw \in  \NB(\tm) \cap
\NF(\tm) \cap \NF(\td)$
and 
let $\tg \eqdef \lambda \tw   \vv\tw/\ty\ww(\tm)$.
We have $\lambda \ty \tm \cv \lambda \ty \tb$ from ($\sL$) and from
 ($\alpha$) it follows that $ \lambda \ty \tm \cv
 \tg$,  as $\tw \in  \NF(\tm) \cap \NB(\tm)$.
 By using ($\sS$),
we obtain 
 $\tg \cv \lambda \ty \tm \cv \lambda \ty \tb$. Hence, we have
  $\tg \cv \lambda \ty \tb$ from ($\sT$). 
  In order to finish the proof, we now show that 
 $\NB(\tg)
 \cup \NF(\td) = \Var$.
 For, we have  
 $\NB(\tg) =
 \NB(\vv\tw/\ty\ww(\tm)) \setminus\lb \tw \rb$. 
 By Lemma \ref{lemmatec1},
it follows that $\NB(\vv\tw/\ty\ww(\tm)) = \NB(\tm)$.
 Hence, we get  
$\NB(\tg)
 \cup \NF(\td)  =  \big(\NB(\tm) \setminus \lb \tw \rb \big) \cup \NF(\td)$. Now,  observe that 
$\big(\NB(\tm) \setminus \lb \tw \rb \big) \cup \NF(\td) = \NB(\tm)  \cup \NF(\td)$,
as $ \tw \in \NF(\td)$. From this, we obtain 
$\NB(\tm)  \cup \NF(\td) = \Var$
 by the inductive hypothesis.
\end{proof}

\begin{lemma} \label{lemfree}
Let $\ta$ and $\td$  be  terms, and let $\tx$ and $\ty$ be   variables such that $\tx \neq \ty$.
Then, we have $\vv \td / \tx \ww (\vv \ty / \tx \ww(\ta))= \vv \ty / \tx \ww(\ta)$.
\end{lemma}
\begin{proof}
We proceed by induction on the structure of $\ta$.

 Suppose that $\ta = \tz$. If $\tx = \tz$, then we have  
 $\vv \td / \tx \ww (\vv \ty / \tx \ww(\tx)) =
 \vv \td / \tx \ww (\ty) = \ty 
 = \vv \ty / \tx \ww(\tx)$, as $\tx \neq \ty$.
If $\tx \neq \tz$, then it follows that 
 $\vv \td / \tx \ww (\vv \ty / \tx \ww(\tz)) =
 \vv \td / \tx \ww (\tz) = \tz 
 = \vv \ty / \tx \ww(\tz)$.

 Suppose that $\ta = \tb\tc$. 
By inductive hypothesis,  we have
$\vv \td / \tx \ww (\vv \ty / \tx \ww(\tb)) $ $ = \vv \ty / \tx \ww(\tb)$ and
$\vv \td / \tx \ww (\vv \ty / \tx \ww(\tc))$ $ = \vv \ty / \tx \ww(\tc)$. Then, we get
$\vv \td / \tx \ww (\vv \ty / \tx \ww(\tb\tc)) $ $ = $  $
\vv \td / \tx \ww (\vv \ty / \tx \ww(\tb) \vv \ty / \tx \ww(\tc))
 = \vv \td / \tx \ww (\vv \ty / \tx \ww(\tb)) 
\vv \td / \tx \ww ( \vv \ty / \tx \ww(\tc))  $  $ = 
\vv \ty / \tx \ww(\tb)\vv \ty / \tx \ww(\tc)  =
\vv \ty / \tx \ww(\tb\tc)$.

Finally, suppose that $\ta = \lambda \tz  \tb$. If  $\tx = \tz$, then  we
 have \linebreak $\vv \td / \tx \ww (\vv \ty / \tx \ww(\lambda \tx  \tb)) =
 \vv \td / \tx \ww (\lambda \tx  \tb)  =
 \lambda \tx  \tb = \vv \ty / \tx \ww(\lambda \tx  \tb)$.
If $\tx \neq \tz$, then, by inductive hypothesis, we have 
$\vv \td / \tx \ww (\vv \ty / \tx \ww(\tb))$ $= \vv \ty / \tx \ww(\tb)$.
Then, we obtain
$\vv \td / \tx \ww (\vv \ty / \tx \ww(\lambda \tz  \tb))  =
\vv \td / \tx \ww (\lambda \tz \vv \ty / \tx \ww( \tb)) $ $ = \lambda \tz\vv \td / \tx \ww ( \vv \ty / \tx \ww( \tb)) = \lambda \tz 
 \vv \ty / \tx \ww(\tb) = 
  \vv \ty / \tx  \ww( \lambda \tz \tb)$.
 \qedhere
\end{proof}

\section{Lambda and Extensional  Theories Revisited} \label{EQU}

In this section, we finally prove  that the notions of lambda and
 extensional congruence are respectively equivalent to the concepts of lambda and extensional  theory.
% In particular, this results shows that
% it is possible to give
% an alternative and
% simplified formulation of the concepts
% of lambda and extensional theory 
% in which no ancillary concept  appears.
To begin with, we show in the next proposition
that prelambda  congruences behave well
with  respect to condition ($\beta$).
%(even though 
%we need this fact to hold only  for lambda and extensional congruences, we show it for prelambda congruences because 
%neither $\alpha$--renaming nor $\eta$--extensionality   used in the proof).

\begin{proposition} \label{betaconv}
Let $\cv$ be a prelambda congruence. Let $\tx$ be a variable, and 
let $\ta$ and $\td$ be terms  such that $\NB(\ta) \cup \NF(\td) = \Var$. Then, we have  $[\lambda \tx  \ta] \td \cv \vv\td/\tx\ww(\ta)$.

\end{proposition}
\begin{proof}  We reason
 by induction on the structure of $\ta$.

 Suppose that $\ta = \ty$. If $\tx=\ty$, then we have
$ \vv\td/\tx\ww(\tx) = \td$ and it follows that
$ [\lambda \tx \tx]\td \cv \td  $ from ($\beta_1$).
If $\tx \neq \ty$, then
we have $\vv\td/\tx\ww(\ty) = \ty$
and we obtain 
$ [\lambda \tx \ty]\td \cv \ty  $
by ($\beta_2$).

 Suppose that $\ta = \tb\tc$.
 Since $\NB(\ta) \subseteq \NB(\tb)$, 
 $\NB(\ta) \subseteq \NB(\tc)$
and  $\NB(\ta) \cup \NF(\td) = \Var$, we have   $\NB(\tb) \cup \NF(\td) = \Var$
 and $\NB(\tc) \cup \NF(\td) = \Var$. 
 Hence, by inductive hypothesis, it follows that  $ [\lambda \tx \tb]\td \cv \vv\td/\tx\ww(\tb)  $
 and  $ [\lambda \tx \tc]\td \cv \vv\td/\tx\ww
 (\tc)  $. From this,  we get
 $[[\lambda \tx \tb]\td][[\lambda \tx \tc]\td] \cv \vv\td/\tx\ww(\tb)\vv\td/\tx\ww(\tc)$ from ($\sA$).
  Now,
 by ($\beta_3$) we have
 $ [\lambda \tx [\tb\tc]]\td \cv [[\lambda \tx \tb]\td][[\lambda \tx \tc]\td] $, and we obtain
  $[\lambda \tx [\tb\tc]]\td\cv \vv\td/\tx\ww(\tb)\vv\td/\tx\ww(\tc) $  from ($\sT$). 
As  $ \vv\td/\tx\ww(\tb)\vv\td/\tx\ww(\tc) = \vv\td/\tx\ww(\tb\tc) $,  we conclude.

Finally,  suppose that $\ta = \lambda \ty  \tb$. If $\tx = \ty$, 
 then we have $\vv\td/\tx\ww(\lambda \tx  \tb) = \lambda \tx  \tb$ and we obtain
$[\lambda \tx[\lambda \tx \tb]]\td \cv \lambda \tx  \tb$ by ($\beta_4$). 
  If  $\tx \neq \ty$, then we have $\ty \notin \NB(\tb) \setminus \lb \ty \rb = \NB(\ta)$.
Since
$\NB(\ta) \cup \NF(\td) = \Var$, it follows that $\ty \in \NF(\td)$.
 As $\tx \neq \ty$, we obtain
$ [\lambda \tx [\lambda \ty   \tb]]\td \cv  \lambda \ty   [[\lambda \tx  \tb]\td ]$ by Proposition \ref{beta5rev}(ii).
Since $\NB(\ta) \subseteq \NB(\tb)$ and $\NB(\ta) \cup \NF(\td) = \Var$, we have   $\NB(\tb) \cup \NF(\td) = \Var$.
By inductive hypothesis,
we get
$ [\lambda \tx \tb]\td \cv \vv\td/\tx\ww(\tb)$.
From this, we obtain
$\lambda \ty   [[\lambda \tx \tb]\td ]\cv \lambda \ty   \vv\td/\tx\ww(\tb)$ from ($\sL$) and
$[\lambda \tx [\lambda \ty   \tb]]\td \cv \lambda \ty   \vv\td/\tx\ww(\tb)$ from ($\sT$).
Since $ \lambda \ty  \vv\td/\tx\ww(\tb) = \vv\td/\tx\ww(\lambda \ty  \tb)$, we conclude.
\end{proof}

We are now in a position to prove the main results of this section.

\begin{proposition} \label{T1}
Every lambda congruence is a lambda theory.

\end{proposition}
\begin{proof}
Let $\cv$ be an lambda congruence. In order to prove that $\cv$ is a lambda theory
we only have to show that
$\cv$ satisfies conditions 
($\beta$) and ($\alpha$). 
%of Definition \ref{deriR}.

 ($\beta$) Suppose that $\NB(\ta) \cup \NF(\td) = \Var$. 
Since every lambda congruence is also a prelambda congruence, we obtain $[\lambda \tx \ta]\td \cv \vv\td/\tx\ww(\ta)$ by 
Proposition \ref{betaconv}.

 ($\alpha$)  
 Suppose that $\ty \in \NF(\ta) \cap \NB(\ta)$.
As $\ty \in \NF(\ta)$, we have  
 $\lambda \tx\ta\cv \lambda \ty[[\lambda \tx\ta]\ty]$ by Proposition \ref{alpharev}(ii).
As $\ty \in \NB(\ta)$ and $\NF(\ty) = \Var \setminus \lb \ty \rb$, we have  $\NB(\ta) \cup \NF(\ty) = \Var$. By Proposition \ref{betaconv}, it follows that  $[\lambda \tx   \ta ]\ty \cv \vv\ty/\tx\ww(\ta)$, and we get 
$\lambda \ty [[\lambda \tx   \ta ]\ty] \cv \lambda \ty \vv\ty/\tx\ww(\ta)$ from ($\sL$).
Then, we conclude that
$\lambda \tx\ta \cv \lambda \ty\vv\ty/\tx\ww(\ta)$ from
($\sT$). 
 \end{proof}

\begin{proposition} \label{T2} Every lambda theory is a lambda congruence.

\end{proposition}
\begin{proof}
Let $\cv$ be a lambda theory. 
By Theorem \ref{thmalpha},
in order to prove that $\cv$ is a lambda congruence all
we have to show is that
$\cv$ is a prelambda congruence which satisfy the condition of
  extensionality for
abstractions ($e_a$).

($\beta_1$) 
Since  $\NB(\tx) = \Var$, we have $\NB(\tx)
\cup \NF(\td) = \Var$.
As $\vv\td/\tx\ww(\tx)  = \td$,
we obtain $[\lambda \tx \tx]\td \cv \td$
by ($\beta$).

($\beta_2$) Suppose that $\tx \neq \ty$.
Since  $\NB(\ty) = \Var$, we have $\NB(\ty)
\cup \NF(\td) = \Var$. 
As $\vv\td/\tx\ww(\ty)  = \ty$,
we have $[\lambda \tx \ty]\td \cv \ty$
by ($\beta$).

($\beta_3$)
By Proposition \ref{renam}, there exist some terms
 $\tm$ and $\tn$   such that $\tm \cv \ta$
 and $\NB(\tm) \cup \NF(\td) = \Var$, as well as
  $\tn \cv \tb$
 and
$\NB(\tn) \cup \NF(\td) = \Var$.
We have $\tm\tn \cv \ta\tb$
 by ($\sA$) and
$  \NB(\tm\tn) \cup \NF(\td)  = \big( \NB(\tm) \cap \NB(\tn) \big) \cup \NF(\td) = \big(\NB(\tm) \cup \NF(\td)\big) \cap  \big(\NB(\tn) \cup \NF(\td)\big) 
 =  \Var \cap\Var 
 = \Var$. 
For $\te \in \lb \ta, \tb, \ta\tb\rb$, 
let

 \smallskip
 
 {\centering $\tf \, \eqdef \, $   $\left\{
  \begin{array}{cl}
\tm & \hbox{if $\te = \ta$}  \\
\tn & \hbox{if $\te = \tb$}  \\
\tm\tn  & \hbox{if $\te = \ta\tb$\enspace.} \\
  \end{array}
\right.$ \par}

\smallskip
 
\noindent
Suppose that $\te \in \lb \ta, \tb, \ta\tb\rb$. 
Since $\tf \cv \te$, we have  $\lambda \tx \tf \cv \lambda \tx\te$ by $(\sL$).
 From ($\sR$), we have $\td \cv \td$
 and so we get 
 $[\lambda \tx \tf]\td \cv [\lambda \tx\te]\td$ by ($\sA$).
 Since $\NB(\tf) \cup \NF(\td) = \Var$ 
 we have $[\lambda \tx \tf]\td \cv
 \vv \td/\tx \ww(\tf)$ from ($\beta$).
 By using ($\sS$),
 we obtain
 $ [\lambda \tx\te]\td \cv [\lambda \tx \tf]\td \cv \vv \td/\tx \ww(\tf)$.
 Hence, we have $ [\lambda \tx\te]\td \cv \vv \td/\tx \ww(\tf)$  from ($\sT$). 
This shows that
  $ [\lambda \tx\ta]\td  \cv \vv \td/\tx \ww(\tm)$,
    $ [\lambda \tx\tb]\td  \cv \vv \td/\tx \ww(\tn)$ and   $ [\lambda \tx[\ta\tb]]\td  \cv \vv \td/\tx \ww(\tm\tn)$. 	
    Now, as  $ \vv\td/\tx\ww(\tm)\vv\td/\tx\ww(\tn) = \vv\td/\tx\ww(\tm\tn)$,
     we obtain 
    $[[\lambda \tx   \ta]\td][[\lambda \tx   \tb]\td] \cv$  $\vv\td/\tx\ww(\tm\tn)$
    by ($\sA$).
 By using ($\sS$),
 we have
 $ [\lambda \tx   [\ta\tb]]\td \cv \vv\td/\tx\ww(\tm\tn) \cv $  $  [[\lambda \tx   \ta]\td][[\lambda \tx   \tb]\td]$.
 From ($\sT$) we finally obtain
   $ [\lambda \tx   [\ta\tb]]\td  \cv  [[\lambda \tx   \ta]\td][[\lambda \tx   \tb]\td]$.

($\beta_4$)
By Proposition \ref{renam},
 there exists a  term
 $\tg$    such that $\tg \cv \ta$
 and
 $\NB(\tg) \cup \NF(\td) = \Var$.
 Let $\tw \in \NF(\tg) \cap
 \NB(\tg) \cap \NF(\td) \cap (\Var \setminus \lb\tx \rb)$
  % From ($\sL$)
% we obtain $\lambda \tx \ta \cv \lambda \tx \tg$.
and let $\tm \eqdef  \lambda \tw \vv \tw/\tx \ww(\tg)$.
 Since $\tw \in \NF(\tg) \cap \NB(\tg)$ it follows that $\lambda \tx \tg \cv  \tm$ from ($\alpha$).
  By using ($\sL$) and ($\sS$), 
 we obtain $\lambda \tx [\lambda \tx \ta] \cv
 \lambda \tx [ \lambda \tx \tg]$ and   $\lambda \tx[\lambda \tx \tg] \cv \lambda \tx \tm$.   From ($\sR$), we have $\td \cv \td$
 and so we get 
 $[\lambda \tx[\lambda \tx \ta]]\td \cv [\lambda \tx[\lambda \tx \tg]]\td$
 and  $[\lambda \tx[\lambda \tx \tg]]\td \cv [\lambda \tx\tm]\td$
  by ($\sA$).
  By Lemma \ref{lemmatec1},
  we have $\NB(\vv \tw/\tx \ww(\tg)) = \NB(\tg)$.
Hence, it follows that $\NB(\tm) = \NB(\vv \tw/\tx \ww(\tg)) \setminus \lb \tw \rb =
 \NB(\tg) \setminus \lb \tw \rb$.
 Since $\tw \in \NF(\td)$ and $\NB(\tg) \cup \NF(\td)
 = \Var$, we have 
 $
  \NB(\tm) \cup 
 \NF(\td)  =  \big(\NB(\tg) \setminus \lb \tw \rb \big) \cup \NF(\td)
  =  \NB(\tg) \cup \NF(\td)
  =  \Var$. 
   From this, we obtain
 $[\lambda \tx\tm]\td
 \cv \vv \td / \tx \ww(\tm)$
by ($\beta$).
As $\tx \neq \tw$, we have
 $  \vv \td / \tx \ww(\tm) =   \vv \td / \tx \ww(\lambda \tw \vv \tw/\tx \ww(\tg))
 = \lambda \tw\vv \td / \tx \ww( \vv \tw/\tx \ww(\tg))$,  and we get $\vv \td / \tx \ww( \vv \tw/\tx \ww(\tg)) =  \vv \tw/\tx \ww(\tg)$  by  Lemma \ref{lemfree}.
 Hence, we obtain
 $ \vv \td / \tx \ww(\tm) = \lambda \tw
 \vv \tw/\tx \ww(\tg) = \tm$ and 
   $[\lambda \tx\tm]\td
 \cv \tm$. Since $\tg \cv \ta$,
 we obtain $\lambda \tx \tg \cv \lambda \tx\ta$ from ($\sL$)
 and by using ($\sS$) 
  we have 
   $[\lambda \tx[\lambda \tx \ta]]\td \cv [\lambda \tx[\lambda \tx \tg]]\td  \cv [\lambda \tx\tm]\td \cv \tm \cv \lambda \tx \tg \cv \lambda \tx \ta$.
From ($\sT$), 
 we finally obtain
    $[\lambda \tx[\lambda \tx \ta]]\td \cv \lambda \tx \ta$.

($\beta_5$) Suppose that
$\tx \neq \ty$  and
$[\lambda \ty\td]\tx \cv \td$.
Let $\tc \eqdef [\lambda \ty\td]\tx$.
By Proposition \ref{renam}, there exists
a term $\tg$ such that $\tg \cv \ta$
and $\NB(\tg) \cup \NF(\tc) = \Var$.
Since $\NB(\tg) \cup \NF(\tc) = \Var$,
we have 
$ [\lambda \tx   \tg]\tc \cv\vv\tc/\tx\ww(\tg) $
by ($\beta$).
As $\tg \cv \ta$,
we have $\lambda \tx \tg \cv \lambda \tx\ta$
by ($\sL$). From this and $\tc \cv \td$
we obtain 
$[\lambda \tx \tg]\tc \cv [\lambda \tx\ta]\td$
by ($\sA$).
By using ($\sS$), 
we have $\vv\tc/\tx\ww(\tg)  \cv
[\lambda \tx   \tg]\tc \cv  [\lambda \tx\ta]\td$ and
we obtain 
$\vv\tc/\tx\ww(\tg)  \cv
 [\lambda \tx\ta]\td$ from ($\sT$).
Hence, we get 
$\lambda \ty \vv\tc/\tx\ww(\tg)  \cv
\lambda \ty [[\lambda \tx\ta]\td]$ by ($\sL$). 
Now, since $\tg \cv \ta$,
we  have $\lambda \tx [\lambda \ty \tg] \cv \lambda \tx[\lambda \ty\ta]$
by using ($\sL$). 
As $\tc \cv \td$,
we get 
$[\lambda \tx [\lambda \ty \tg]]\tc \cv [\lambda \tx[\lambda \ty\ta]]\td$
by ($\sA$).
 Since $ \tx \neq \ty $, we have $\ty \in
 \big(\NF(\td)
 \cup \lb \ty \rb\big) \cap \big(\Var \setminus \lb \tx \rb\big) = \NF(\lambda \ty\td) \cap \NF(\tx) = \NF([\lambda \ty\td]\tx) =
  \NF(\tc )$. From this and
 $\NB(\tg) \cup \NF(\tc) = \Var$ it follows that  
$
 \NB(\lambda \ty   \tg) \cup 
 \NF(\tc)   =  \big(\NB(\tg) \setminus \lb \ty \rb \big) \cup \NF(\tc) 
 =  \NB(\tg)  \cup \NF(\tc)  = \Var$. Hence,
  we obtain 
  $  [\lambda \tx   [\lambda \ty   \tg]] \tc \cv \vv\tc/\tx\ww(\lambda \ty   \tg)$ by ($\beta$). 
  Since $\tx \neq \ty$, we have 
$\vv\tc/\tx\ww(\lambda \ty   \tg) =   \lambda \ty   \vv\tc/\tx\ww(\tg)$. Thus, 
we get
 $ [\lambda \tx   [\lambda \ty   \tg]] \td \cv \lambda \ty   \vv\tc/\tx\ww(\tg)$.
By using ($\sS$) 
  we obtain 
$ [\lambda \tx[\lambda \ty\ta]]\td \cv [\lambda \tx [\lambda \ty \tg]]\tc 
 \cv \lambda \ty   \vv\tc/\tx\ww(\tg)
 \cv \lambda \ty [[\lambda \tx\ta]\td]$. Hence, we conclude that
 $ [\lambda \tx[\lambda \ty\ta]]\td \cv \lambda \ty [[\lambda \tx\ta]\td]$
  from  ($\sT$).

  ($e_a$) This condition holds by Proposition \ref{propalfa}.
\end{proof}

\begin{proposition} \label{T3}
Every extensional congruence is an extensional theory.

\end{proposition}
\begin{proof}
Let $\cv$ be an extensional congruence. In order to prove that $\cv$ is an extensional theory
we have to show that
$\cv$ satisfies conditions 
($\beta$), ($\alpha$)
and ($\eta$).
By Theorem \ref{alpha}
we know that $\cv$ is a lambda congruence and so by 
Proposition \ref{T1} it satisfies 
($\beta$) and ($\alpha$).
As for ($\eta$), 
it holds 
by  Proposition \ref{propeta}(ii).
 \end{proof}

\begin{proposition} \label{T4}
Every extensional theory is  an extensional congruence.

\end{proposition}
\begin{proof}
Let $\cv$ be an extensional theory. 
In order to prove that $\cv$ is an extensional congruence
we have to show that
$\cv$ satisfies all beta conditions 
 and ($\eta_e$).
Since, by definition,  $\cv$ is also a lambda theory, we know by Proposition \ref{T2} that
$\cv$ satisfies 
all beta conditions.
As for ($\eta_e$), suppose that
$\tx \neq  \ty$.
Then, as $\tx \in \Var \setminus \lb \ty \rb = \NF(\ty)$,
we obtain
 $ \ty \cv \lambda \tx   [\ty\tx]$
by ($\eta$).
\end{proof}

\begin{theorem} \label{TT1}
The concepts of 
lambda and extensional congruence are respectively equivalent to the notions of lambda and extensional theory.
\end{theorem}
\begin{proof}
By the previous four propositions.
\end{proof}

By using the above theorem,
it is also possible to give  simplified formalizations
of lambda conversion and extensional conversion as we now explain.

Recall that lambda  conversion is the lambda  theory inductively
defined by using  structural conditions, ($\beta$) and ($\alpha$).
In other words, lambda  conversion is the intersection
of all lambda theories.
By Theorem \ref{TT1},  the set of all lambda 
theories is equal to the set of all
lambda
congruences. Thus, lambda  conversion
is also the intersection of all lambda congruences. In particular,
lambda 
conversion can be equivalently
characterized as the relation
inductively defined by using
all structural and beta conditions
together with ($\alpha_e$).

By replacing ($\alpha_e$)
by ($\eta_e$) in the discussion above, a simplified axiomatization of extensional conversion is similarly given.

   \section{The Prelambda Congruence $\mo$}
    \label{modelsection}
   The aim of this section is to show
 that there   exists a prelambda congruence
 $\mo$
such that    for all distinct
   variables  $\tx$ and $\ty$ we have $\lambda \tx \tx \nocv\lambda \ty\ty$.
   In other words, we now prove  Theorem \ref{teononalpha} whose proof was omitted in Section \ref{PBC}.

This is not an easy task, though.
The reason is that 
the defining conditions of prelambda congruence
only deal with  ``provability'' while
our aim is to show the
 ``unprovability'' of $\lambda \tx \tx \mo \lambda \ty\ty$, that is
$\lambda \tx \tx \nocv\lambda \ty\ty$.

A   reassuring fact is that similar situations
frequently happen in logic and in that setting well--known  methods to solve these kinds of problems
are available.
% The drawback is that  
%these methods are usually quite complicated.
%Consider for instance
%the lambda calculus. In order  
%to show that two distinct
%variables are not $\beta$--convertible,
%two standard 
%--- but intrinsically complicated ---
% methods are at our disposal:
%we can either prove the Church--Rosser
%property or show the existence of a
%(non--trivial) model.
Consider for instance
 \emph{intuitionistic theories}  and the law of excluded middle --- for the purposes of this discussion an  intuitionistic theory is just a set of formulas which is closed under the rules of intuitionistic logic.
We say that an intuitionistic theory
is \emph{classical} if it contain every instance of excluded middle.
 Of course, there are several intuitionistic theories
which are classical: classical logic is the primary example.  

Suppose that we want to show that there exists
an intuitionistic theory which is \emph{not} classical. In order to do this, it suffices to
construct
a model $M$ and prove that the
set 
of formulas which are valid in $M$
forms a non--classical
 intuitionistic theory.
Note that $M$
cannot be a standard  
 model of classical logic
 ---  excluded middle never fails there --- but it has to be a model specifically designed to this aim,
such as a
 topological model or a 
 Kripke model;  see,  \eg S{\o}rensen and
               Urzyczyn \cite[Ch. 2]{CHIU}. 
               %for a brief description of these models. 
 
Now,  if we think of 
 prelambda congruences
as intuitionistic theories
and 
 ``$\lambda \tx \tx \cv \lambda \ty\ty$ for all $\tx$ and $\ty$''
as the law of excluded middle,
then  
 the situation is strikingly similar
 to the one described above.
Hence, in order to solve our problem  we now build  a kind of model with the required properties.  
Since
all models considered
 in the literature of the lambda calculus
have been conceived for the purpose of \emph{validating} (not \emph{refuting})
 ``$\lambda \tx \tx \cv \lambda \ty\ty$ for all $\tx$ and $\ty$'',  our construction,  to the best of our knowledge,
seems to be new.

%As it is discussed below, our model
%and our interpretation of terms in the model does not differ that much
%with standard models and interpretation,
%but these tiny differences  
%produce consequences.

   We now introduce some notation.
   Let $A$ and $B$ be  sets. We write $B \subseteq_f A$
  to express the fact   that $B$ is a finite subset of $A$. 
  %Sometimes
 % we write $B \subseteq_f A$
 % even if we already know that $B$ is finite, just to emphatize this.
   We also write $P(A)$ and $P_f(A)$
   for the set of all subsets and finite subsets
   of $A$, respectively. 
   Furthermore, we write $A \times B$
   to denote the Cartesian product of $A$ and $B$.
   
  In order to construct  our model
  we now introduce the notion
  of \emph{formula}.

\begin{definition}[Formula]
 Let $\tF$ be the set inductively
defined as follows:
\begin{itemize}
\item[($\tF_1$)]  $\tx \in \Var$ implies $\tx \in \tF$;
\item[($\tF_2$)]    $F \subseteq_f \tF$ and $\ell \in \tF$ imply  $( F,\ell ) \in \tF$. 
\end{itemize}

We call \textbf{formula} any element of the set $\tF$.
Henceforth, we use $F$, $G$, $H$,\ldots and $\ell$, $m$, $n$, \ldots  to denote
finite subsets and elements of $\tF$, respectively.
We also use the expression $F \vdash \ell$
as an alternative notation for the ordered pair $(F,\ell)$.
In particular, we write $G \vdash F \vdash \ell$ 
for $(G,(F,\ell))$.

Let $\ell$ be a formula.  We say that $\ell$ is an \textbf{atomic formula} if $\ell \in \Var$
and that $\ell$ is  a \textbf{compound formula} if $\ell \notin \Var$; equivalently, if it is of the form $F \vdash m$.
\hfill $\triangle$   
\end{definition}

We now observe that our set $\tF$ belongs to the class of structures 
 called \emph{graph algebras},
 in the terminology of  Engeler \cite{engeler}. 
In particular,  $\tF$
corresponds to 
the full graph algebra $G(\Var)$ built on the set of variables $\Var$.
%where the operation $\cdot$ is given
%by $X \cdot Y \eqdef \lb \ell \st $there exists $F \subseteq_f Y$ such that $F \vdash \ell \in X \rb$.
For our aims
the fact that $\tF$ is specifically constructed out of the set of variables $\Var$ turns out to be essential, as we make clear later.
Graph algebras are also known as Engeler models and PSE--algebras in the literature; see,  \eg
 Meyer \cite{MEYER}, Longo \cite{LONGO}, Barendregt \cite{BAR},  Krivine \cite{KR}, Plotkin \cite{PLOTKIN}, Berline \cite{BER,BER06} and Hindley and Seldin \cite{HS}.
Regarding our terminology,
we  refer to elements $\tF$
as formulas because a similar terminology
is employed in \cite{KR} for similar structures.

Our aim   is to interpret terms
using the structure we have just
defined. In order to do so, it is necessary
to introduce the concepts of environment
and update first.

\begin{definition}[Environment, update] We call \textbf{environment}
any function $\sigma$ from $\Var$ to $P_f(\tF)$.
In the sequel, we use $\sigma$, $\rho$, $\tau$, \ldots
to denote environments. We also denote the set of all environments by $\mathbb{E}$.

Let $\sigma$ be an environment. Let $\tx \in \Var$ and let
$F \subseteq_f \tF$. We denote by $\lb F/\tx\rb\sigma$
the environment  given by:

\smallskip

%{\centering $\lb F/\tx\rb\sigma(\ty) \, \eqdef \, $ \emph{if $\tx = \ty$ then $F$ else $\sigma(\ty)$}\ \ , \quad  for $\ty \in \Var$\enspace. \par}

\smallskip

{\centering $\lb F/\tx\rb\sigma(\ty) \, \eqdef \, $   $\left\{
  \begin{array}{cl}
F & \hbox{if $\tx = \ty$}  \\
\sigma(\ty)  & \hbox{if $\tx \neq \ty$\enspace, \quad for $ \ty \in \Var$\enspace.} \\
  \end{array}
\right.$ \par}

\smallskip

\noindent We  call $\lb F/\tx \rb\sigma$ an \textbf{update of} $\sigma$.

Let $\tx$ and $\ty$ be variables, and
let $F$ and $G$ in $P_f(\tF)$.
Let $\sigma$ be an environment.
If $\rho = \lb F/\tx \rb \sigma$
and $\tau = \lb G/\ty \rb \rho$,
then we also write  $\tau$  as 
$\lb G/\ty \rb\lb F/\tx \rb \sigma$.
\hfill $\triangle$
\end{definition}

Updates obey the  algebraic laws that we show in the next lemma.

\begin{lemma}\label{lemma ass}
Let $\sigma$ be an environment. Let 
$\tx$ and $\ty$  be variables such that $\tx \neq \ty$, and let 
$F$ and $G$ be finite subsets of $\tF$. Then, the following properties hold:
\begin{itemize}
\item[\emph{(i)}] $\lb G/\tx\rb\lb F/\tx\rb\sigma  = \lb G/\tx\rb\sigma$;
\item[\emph{(ii)}]
$\lb G/\ty\rb\lb F/\tx\rb\sigma =
\lb F/\tx\rb\lb G/\ty\rb\sigma$. 
%\item[$\emph{(iii)}$] $\sigma\{\tx \gets F\}\{\ty \gets G\} = \sigma \{\ty \gets H\}\{\tx \gets F\}\{\ty \gets G\}$;
%\item[$\emph{(iv)}$] if $\tz \neq \ty$, then $\sigma\{\ty \gets G\}\{\tx \gets F\}(\tz) = \sigma\{\tx \gets F\}(\tz) $.
\end{itemize}
\end{lemma}
\begin{proof} 
(i) 
We have 
$
\lb G/\tx\rb\lb F/\tx\rb\sigma(\tx) =G
=\lb G/\tx\rb\sigma(\tx)$.
Let  $\tz$ be a variable such that $\tx \neq \tz$. Then, we have 
$
\lb G/\tx\rb\lb F/\tx\rb\sigma(\tz) =\lb F/\tx\rb\sigma(\tz) =\sigma(\tz)
=\lb G/\tx\rb\sigma(\tz)$.

(ii)  
Since $\tx \neq \ty$, we have 
$
\lb G/\ty\rb\lb F/\tx\rb\sigma(\tx) = \lb F/\tx\rb \sigma(\tx) = $ $ F
=  \lb F/\tx\rb\lb G/\ty\rb \sigma(\tx)$
and
$\lb G/\ty\rb\lb F/\tx\rb\sigma(\ty) = G =\lb G/\ty\rb\sigma(\ty)
=\lb F/\tx\rb\lb G/\ty\rb \sigma(\ty)$.
Let  $\tz$ be a variable  such that $\tx \neq \tz$ and $\ty \neq \tz$. Then, we have 
  $
\lb G/\ty\rb\lb F/\tx\rb\sigma(\tz) = \lb F/\tx\rb\sigma(\tz)  = \sigma(\tz)  =\lb G/\ty\rb\sigma(\tz)
=\lb F/\tx\rb\lb G/\ty\rb \sigma(\tz)$.
%(iii) Let $\rho \eqdef \sigma\{\ty \gets H\}$. 
%We have to show  show that 
%$\sigma\{\tx \gets F\}\{\ty \gets G\} = \rho\{\tx \gets F\}\{\ty \gets G\} $
%Assume $\tx = \ty$. 
%By using (i) above, we have $\sigma\{\tx \gets F\}\{\ty \gets G\} = \sigma\{\ty \gets G\} $
%and  $\rho\{\tx \gets F\}\{\ty \gets G\} = \rho\{\ty \gets G\}$. By using (i) above again, we finally get
%$\rho\{\ty \gets G\} = \sigma\{\ty \gets H\}\{\ty \gets G\}
%= \sigma \{\ty \gets G\}$. 
%Assume now $\tx \neq \ty$.
%By using (i) above, we have $
% \rho\{\ty \gets G\} =  \sigma\{\ty \gets H\}\{\ty \gets G\}
% = \sigma\{\ty \gets G\}$.
%By using (ii) above, we then obtain 
%$\rho\{\tx \gets F\}\{\ty \gets G\} = \rho\{\ty \gets G\}\{\tx \gets F\} = \sigma\{\ty \gets G\}\{\tx \gets F\}$
%By using (ii) again, we finally get 
%$ \sigma\{\ty \gets G\}\{\tx \gets F\} = \sigma\{\tx \gets F\}\{\ty \gets G\}$.
%
%  
%  (iv) 
%If $\tz = \tx$, then we have 
% $\sigma\{\ty \gets G\}\{\tx \gets F\}(\tz) 
%  =  F
% = 
%\sigma\{\tx \gets F\}(\tz)$.
%If $\tz \neq \tx$, then we have 
% $\sigma\{\ty \gets G\}\{\tx \gets F\}(\tz) =
%\sigma\{\ty \gets G\}(\tz) 
% = \sigma(\tz) =
%\sigma\{\tx \gets F\}(\tz)$.
\qedhere
\end{proof}

We now define our interpretation of terms.
As is standard for graph algebras, 
each term is interpreted via an environment as  a subset of formulas.

\begin{definition}[Interpretation of terms in $P(\tF)$] \label{interm}
 We define the function  $\cI$ from $\Lam \times \mathbb{E}$ to $P(\tF)$   by induction on the structure of terms as follows:

\begin{itemize}
\item[($\cI_1$)] 
$\bE(\tx,\sigma)  \eqdef \sigma(\tx)$; 
\item[($\cI_2$)] 
$\bE(\ta\tb,\sigma)  \eqdef  \lb \ell \st $ there exists  $F \subseteq_f \bE(\tb,\sigma)$  such that \\ $\phantom{asaasasaasaasas}  F  \vdash \ell \in \bE(\ta,\sigma) \rb$;  
\item[($\cI_3$)] 
$\bE(\lambda \tx  \ta,\sigma)   \eqdef  \lb    F \vdash m  \st     m\in \bE(\ta,\lb F/\tx \rb\sigma)\rb \, \cup \, \lb \tx \rb$.

\end{itemize}

Let $\ta$ and $\tb$ be terms.
We write $\ta \mo \tb$ and say that \textbf{$\ta$ and $\tb$ have the same interpretation in $P(\tF)$} if  $\bE(\ta,\sigma) = \bE(\tb,\sigma)$
for every environment $\sigma$.
We also call the relation $\mo$
the \textbf{theory of $P(\tF)$}. 
\hfill $\triangle$
\end{definition}

While
our \emph{set} $\tF$ is just  an example of  graph algebra,
%, a
%structure which is well--known and discussed  in the literature of the lambda calculus 
our \emph{interpretation} of terms in $P(\tF)$ is definitely non--standard,  as we now explain.
%: our notion of environment and condition ($\cI_3$)
%of Definition \ref{interm}.

Firstly,   environments
are usually taken as function
from $\Var$ to $P(\tF)$ 
and not as 
functions from $\Var$ to $P_f(\tF)$ as we do in this paper.
The reason is practical: to show our results we found it is not necessary consider
environments having infinite sets of formulas in their range. 

Secondly, 
%another point of  divergence
%lies in clause $(\cI_3)$ of Definition \ref{interm}.
following  the standard interpretation of terms in graph algebras  $\cI(\lambda \tx \ta,\sigma)$  should be the set of formulas  $  \lb    F \vdash m  \st     m\in \bE(\ta,\lb F/\tx \rb\sigma)\rb$ and \emph{not} our set  $\lb    F \vdash m  \st     m\in \bE(\ta,\lb F/\tx \rb\sigma)\rb \, \cup \, \lb\tx \rb$. As a consequence of this fact, in our setting each set of the form
$\bE(\lambda \tx  \ta,\sigma)$ contains  
compound formulas of the form $ F \vdash m$ and  \emph{exactly
one} atomic formula, namely $\tx$.
This formula $\tx$  gives  us a ``tag''  for the ``$\lambda \tx$'' in the interpretation of $\lambda \tx \ta$ and the presence of this unique atomic formula turns out to be crucial for our purposes.
For this reason, the fact that  variables
are also formulas is very important in our setting.
%We also point out that
%other ---  but different from ours ---
%less standard  interpretations
%of terms of the form $\lambda \tx \ta$
%in graph algebras
%have been already considered in the literature; see,  \eg
%Longo \cite{LONGO}, Di Giannantonio and Honsell \cite{DGH} and Plotkin\cite{PLOTKIN}. 

We now show two  lemmas
that we need during the proof of Theorem \ref{impo}.
%, our version of the soundness theorem.

\begin{lemma} \label{prop up}
Let $\tx$ be a variable and  let $\ta$ be a term.
Let $F$ and $G$ finite subsets of of formulas  such that $F \subseteq G$. 
 Then, for every environment $\sigma$   we have 
 $\bE(\ta,\lb F/\tx\rb \sigma) \subseteq \bE(\ta,\lb G/\tx\rb \sigma)$. 

\end{lemma}
\begin{proof} We reason by induction on the structure of $\ta$.  Let
$\sigma$ be an environment and let $\ell$ be a formula.

Suppose that $\ta = \ty$. If $\tx = \ty$, 
then we have 
$\bE(\tx,\lb F/\tx\rb \sigma) = \lb F/\tx\rb \sigma(\tx) = F \subseteq G = \lb G/\tx\rb \sigma(\tx) = \bE(\tx,\lb G/\tx\rb \sigma)$.
If $\tx \neq \ty$,  
then we have
$\bE(\ty,\lb F/\tx\rb \sigma) = \lb F/\tx\rb \sigma(\ty) = \sigma(\ty) =  \lb G/\tx\rb \sigma(\ty) = \bE(\ty,\lb G/\tx\rb \sigma)$.

 Suppose that $\ta = \tb \tc$.
By inductive
hypothesis,  
 for every environment $\rho$ and
  every $\td \in \lb \tb, \tc \rb$     we have 
 $\bE(\td,\lb F/\tx\rb \rho) \subseteq \bE(\td,\lb G/\tx\rb \rho)$.  
Assume  
 $\ell \in \bE(\tb\tc,\lb F/\tx\rb \sigma)$. Then, there exists $H
 \subseteq_f \bE(\tc,\lb F/\tx\rb \sigma)$ such that  $ H \vdash \ell  \in \bE(\tb,\lb F/\tx\rb \sigma)$. 
 By inductive hypothesis,
 it follows that
  $H
 \subseteq_f \bE(\tc,\lb G/\tx\rb \sigma)$ and $ H \vdash \ell  \in \bE(\tb,\lb G/\tx\rb \sigma)$. From this, we conclude that
 $\ell \in \bE(\tb\tc,\lb G/\tx\rb \sigma)$.

 Suppose that $\ta = \lambda \ty\tb$.
 Assume  
 $\ell \in \bE(\lambda \ty\tb,\lb F/\tx\rb \sigma)$. Suppose that $\tx = \ty$.
 If $\ell = \tx$, then we obviously have $\tx \in 
  \bE(\lambda \tx\tb,\lb G/\tx\rb \sigma)$.
Otherwise, it follows that $\ell = H \vdash m$
  for some
 $H$  and $m$ such that $m \in \bE(\tb,\lb H/\tx \rb \lb F/\tx\rb\sigma)$. 
 By using Lemma \ref{lemma ass}(i)
  it follows that 
 $m \in \bE(\tb,\lb H/\tx \rb \lb G/\tx\rb\sigma)$ and we  obtain  
 $H \vdash m  \in \bE(\lambda \tx\tb,\lb G/\tx\rb \sigma)$.
Suppose now that $\tx \neq \ty$.
  If $\ell = \ty$, then we clearly have $\ty \in 
  \bE(\lambda \ty\tb,\lb G/\tx\rb \sigma)$.
Otherwise, it follows that $\ell = H \vdash m$
  for some
 $H$  and $m$ such that $m \in \bE(\tb,\lb H/\ty \rb \lb F/\tx\rb\sigma)$. 
 By inductive
hypothesis,  
 for every environment $\rho$  we have  
 $\bE(\tb,\lb F/\tx\rb \rho) \subseteq \bE(\tb,\lb G/\tx\rb \rho)$.
From this and  Lemma
 \ref{lemma ass}(ii)  we obtain 
 $m \in \bE(\tb,\lb H/\ty \rb \lb F/\tx\rb\sigma) = \bE(\tb, \lb F/\tx\rb\lb H/\ty \rb\sigma) \subseteq \bE(\tb, \lb G/\tx\rb\lb H/\ty \rb\sigma) =  \bE(\tb,\lb H/\ty \rb \lb G/\tx\rb\sigma)$, as  $\tx \neq \ty$.  Therefore, we  get
 $H \vdash m  \in \bE(\lambda \ty\tb,\lb G/\tx\rb \sigma)$.
  \qedhere
\end{proof}

	\begin{lemma} \label{proptecc}
	Let $\tx$ and $\tz$ be variables such
	that $\tx \neq \tz$, and let $\ta$ be a term. Let $F$ be a finite subset of formulas.
Suppose that $[\lambda \tx \ta]\tz \mo  \ta$. Then, we have $\cI(\ta,\sigma) = \cI(\ta, \lb F/\tx \rb \sigma)$ for every environment $\sigma$. 
	\end{lemma}
\begin{proof} Let $\sigma$ be an environment and let $\ell$ be a formula.
As $[\lambda \tx \ta]\tz \mo  \ta$,
we have 
$\cI([\lambda \tx \ta]\tz,\sigma) =  \cI(\ta,\sigma)$ and 
$\cI([\lambda \tx \ta]\tz,\lb F /\tx \rb\sigma) =  \cI(\ta,\lb F /\tx \rb\sigma)$.

Assume $ \ell \in \cI(\ta,\sigma) = \cI([\lambda \tx \ta]\tz, \sigma)$. Then, there exists $G \subseteq_f \cI(\tz,\sigma)$
 such that  
 $G \vdash \ell \in \cI(\lambda \tx \ta,\sigma)$ and $\ell \in \cI(\ta,\lb G/\tx\rb\sigma)$.
Since  $\tx \neq \tz$, it follows
that $G \subseteq_f \cI(\tz,\sigma) = \sigma(\tz) = \lb F /\tx \rb\sigma(\tz)
= \cI(\tz, \lb F /\tx \rb\sigma)$.
By using Lemma \ref{lemma ass}(i), it follows that  
 $\ell \in \cI(\ta,\lb G/\tx\rb\lb F/\tx \rb\sigma)$ and hence
 $G \vdash \ell \in \cI(\lambda \tx \ta, \lb F/\tx \rb\sigma)$.
 Since $G \subseteq_f 
 \cI(\tz, \lb F /\tx \rb\sigma)$,
we obtain  $\ell \in 
\cI([\lambda \tx\ta]\tz, \lb F/\tx \rb\sigma)= 
\cI(\ta, \lb F/\tx \rb\sigma)$.
This shows that  $\cI(\ta,\sigma) \subseteq \cI(\ta, \lb F/\tx \rb \sigma)$. 
To show the opposite inclusion, 
assume $ \ell \in \cI(\ta,\lb F/\tx \rb\sigma) =\cI([\lambda \tx \ta]\tz, \sigma\lb F/\tx \rb)$. Then, there exists some   $G \subseteq_f \cI(\tz,\sigma\lb F/\tx \rb)$
 such that  
 $G \vdash \ell \in \cI(\lambda \tx \ta,\lb F/\tx \rb\sigma)$ and $\ell \in \cI(\ta,\lb G/\tx\rb\lb F/\tx \rb\sigma)$.
By using Lemma \ref{lemma ass}(i), we get 
$\ell \in \cI(\ta,\lb G/\tx\rb\sigma)$.
So, we have
 $G \vdash \ell \in \cI(\lambda \tx \ta,\sigma)$.
As  $\tx \neq \tz$, we get
 $G \subseteq_f \cI(\tz, \lb F /\tx \rb\sigma) =  \lb F /\tx \rb\sigma(\tz) = \sigma(\tz) = \cI(\tz,\sigma)$ and we obtain
 $\ell \in 
\cI([\lambda \tx\ta]\tz, \sigma) =
\cI(\ta, \sigma)$.
Therefore, we have  $ \cI(\ta, \lb F/\tx \rb \sigma) \subseteq \cI(\ta,\sigma)$. 
 \end{proof}

We can now prove the theory of our model is a prelambda congruence. 

\begin{theorem} \label{impo} The relation  $ \mo$ is a prelambda congruence. 
\end{theorem}
\begin{proof}
All we have to do is to check that $\mo$ satisfies all  structural and beta conditions.
 Let $\sigma$ be an environment 
and let $\ell$ be a formula.

As for ($\sR$), ($\sS$) and ($\sT$), let $\rho$ be and environment. Since $\cI(\td,\rho)$ is a set for each term $\td$,
the following conditions hold:
  $\cI(\ta,\rho) = \cI(\ta,\rho)$,
 $\cI(\ta,\rho) = \cI(\tb,\rho)$
implies  $\cI(\tb,\rho) = \cI(\ta,\rho)$, 
 $\cI(\ta,\rho) = \cI(\tb,\rho)$
and  $\cI(\tb,\rho) = \cI(\tc,\rho)$
imply  $\cI(\ta,\rho) = \cI(\tc,\rho)$.
From this observation, the fact that $\mo$ satisfies ($\sR$), ($\sS$) and ($\sT$) follows.

($\sL$)
Suppose that
$\bE(\ta,\rho)
= \bE(\tb,\rho)$ for every environment $\rho$. We have to show  that
$\bE(\lambda \tx \ta,\sigma)
= \bE(\lambda \tx \tb,\sigma)$. For $\tm \in \lb \ta,\tb \rb$, let
% \emph{$\tn \eqdef $ if $\tm = \ta$ then $\tb$ else $\ta$}.

\smallskip

{\centering $\tn \, \eqdef \, $   $\left\{
  \begin{array}{cl}
\tb & \hbox{if $\tm = \ta$}  \\
\ta  & \hbox{if $\tm = \tb$\enspace.} \\
  \end{array}
\right.$ \par}

\smallskip

\noindent 
Suppose that $\ell \in \bE(\lambda \tx \tm,\sigma)$.
If $\ell = \tx$, then we clearly have
$\tx \in \bE(\lambda \tx \tn,\sigma)$.
Otherwise, 
it follows that $\ell = F \vdash m$
for some $F$ and $m$ such that
 and $m \in \bE(\tm, \lb F/\tx \rb\sigma)$.
 Since 
 $\bE(\ta,\lb F/\tx \rb\sigma)= \bE(\tb,\lb F/\tx \rb\sigma)$, we have
  $\bE(\tm,\lb F/\tx \rb\sigma)= \bE(\tn,\lb F/\tx \rb\sigma)$ and hence
 $m \in  \bE(\tn,\lb F/\tx \rb\sigma)$.
 From this, 
 we conclude that
$F\vdash m  \in \bE(\lambda \tx \tn,\sigma)$.
This shows that 
$\bE(\lambda \tx \ta,\sigma)
= \bE(\lambda \tx \tb,\sigma)$.

($\sA$)
 Suppose that
$\bE(\ta,\rho)
= \bE(\tb,\rho)$ 
for every environment $\rho$ 
and
$\bE(\tc,\tau)
= \bE(\td,\tau)$ for every environment $\tau$. We have to show that 
$\bE( \ta\tc,\sigma)
= \bE(\tb\td,\sigma)$. 
For $\tm\tm' \in \lb \ta\tc,\tb\td \rb$, let
% \emph{$\tn\tn' \eqdef $ if $\tm\tm' = \ta \tc$ then $\tb\td$ else $\ta\tc$}.

\smallskip

{\centering $\tn\tn' \, \eqdef \, $   $\left\{
  \begin{array}{cl}
\tb\td & \hbox{if $\tm\tm' = \ta\tc$}  \\
\ta\tc  & \hbox{if $\tm\tm' = \tb\td$\enspace.} \\
  \end{array}
\right.$ \par}

\smallskip

\noindent
% Let $\tm\tm' \in \lb \ta\tc,\tb\td \rb$.
Suppose that $\ell \in \bE(\tm\tm',\sigma)$.
Then,  there exists $F \subseteq_f \bE(\tm',\sigma)$   such that 
$ F \vdash \ell  \in \bE(\tm,\sigma)$.
 Since $\bE(\ta,\sigma)
= \bE(\tb,\sigma)$ and $\bE(\tc,\sigma)
= \bE(\td,\sigma)$, it follows that
 $\bE(\tm,\sigma)
= \bE(\tn,\sigma)$ and $\bE(\tm',\sigma)
= \bE(\tn',\sigma)$. 
Thus, we have $F \subseteq_f \bE(\tn',\sigma)$ and
$ F \vdash \ell \in \bE(\tn,\sigma)$.
So, it follows that
$\ell \in \bE(\tn\tn',\sigma)$.
This proves that
$\bE( \ta\tc,\sigma)
= \bE(\tb\td,\sigma)$.

%This shows that $\mo$ is
%a congruence.
%To show that it is also a prelambda congruence we now
% check that $\mo$  satisfies
%  all beta conditions.

($\beta_1$)  We have to show that
$\bE([\lambda \tx  \tx]\td,\sigma) = \bE
(\td,\sigma)$.  
 Suppose that $\ell \in \bE([\lambda \tx  \tx]\td,\sigma)$.
Then, there exists  $F\subseteq_f \bE(\td,\sigma)$
such that  $    F  \vdash \ell  \in \bE(\lambda \tx  \tx,\sigma)$.
Since we have $\ell \in \bE(\tx,\lb F/\tx \rb\sigma) = \lb F/\tx \rb\sigma(\tx) = F$, we conclude that $\ell \in F \subseteq_f
 \bE
(\td,\sigma)$. To show the opposite inclusion, 
 suppose now that $\ell \in \bE(\td,\sigma)$. Let $F \eqdef \lb \ell \rb$. Since
 $ \ell \in F = \lb F/\tx \rb\sigma(\tx) =  \bE(\tx,\lb F/\tx \rb\sigma)$, we have
 $F \vdash \ell 
  \in \bE(\lambda \tx  \tx,\sigma)$.
As $F = \lb \ell \rb$, it follows that
 $F  \subseteq_f \bE(\td,\sigma)$.  Thus, we obtain $ \ell \in  \bE([\lambda \tx  \tx]\td,\sigma)$.

($\beta_2$) Suppose that $\tx \neq \ty$.
 We have to show that
$\bE([\lambda \tx  \ty]\td,\sigma) = \bE
(\ty,\sigma)$. 
Assume $\ell \in \bE([\lambda \tx  \ty]\td,\sigma)$.
Then, there exists  $F\subseteq_f \bE(\td,\sigma)$
such that  $    F  \vdash \ell  \in \bE(\lambda \tx  \ty,\sigma)$.
Since $\tx \neq \ty$, 
we have  $\ell \in \bE(\ty,\lb F/\tx \rb\sigma) = \lb F/\tx \rb\sigma(\ty) = \sigma(\ty) = \cI(\ty,\sigma)$.
As for the opposite inclusion, 
 assume now  $\ell \in \bE(\ty,\sigma)$. Since $\tx \neq \ty$, we have
 $\bE(\ty,\sigma) = \sigma(\ty) = \lb \emptyset/\tx \rb\sigma(\ty) =  \bE(\ty,\lb \emptyset/\tx \rb\sigma)$. Thus, we have 
 $\ell \in \bE(\ty,\lb \emptyset/\tx \rb\sigma)$ and 
 $\emptyset \vdash \ell 
  \in \bE(\lambda \tx  \ty,\sigma)$.
As 
 $\emptyset  \subseteq_f \bE(\td,\sigma)$, we get $ \ell \in  \bE([\lambda \tx  \ty]\td,\sigma)$.

($\beta_3$)
We have to show that
$\bE([\lambda \tx  [\ta\tb]]\td,\sigma) = \bE([[\lambda \tx \ta]\td][[\lambda \tx  \tb]\td],\sigma)$. 
Suppose that $\ell \in \bE([\lambda \tx  [\ta\tb]]\td,\sigma)$. Then,
 there exists  $F \subseteq_f \bE(\td,\sigma)$  such that  $   F \vdash \ell \in \bE(\lambda \tx [\ta\tb],\sigma)$.
From this, we  obtain  $\ell \in \bE(\ta\tb, \lb F/\tx \rb\sigma)$.
Thus,
 there exists $G \subseteq_f \bE(\tb,\lb F/\tx \rb\sigma)$ such that
 $ G \vdash \ell  \in \bE(\ta, \lb F/\tx \rb\sigma)$. 
Hence, we have $  F\vdash  G \vdash \ell \in \bE(\lambda \tx \ta,\sigma)$.
Since
 $F \subseteq_f \bE(\td,\sigma)$,
we obtain $  G \vdash \ell   \in \bE([\lambda \tx  \ta]\td,\sigma)$. 
Let $n \in G$. Since   $n \in \bE(\tb,\lb F/\tx \rb\sigma)$,  it follows that
$ F \vdash n  \in \bE(\lambda \tx  \tb,\sigma)$. 
Since $F \subseteq_f \bE(\td,\sigma)$, we have  $n \in \cI(
[\lambda \tx  \tb]\td,\sigma)$. Thus, we have
 $G \subseteq_f
\bE([\lambda \tx  \tb]\td,\sigma)$.
Now, as
 $ G \vdash \ell   \in \bE([\lambda \tx  \ta]\td,\sigma)$, we obtain
$\ell \in \bE([[\lambda \tx  \ta]\td][[\lambda \tx  \tb]\td],\sigma)$.
To show the opposite inclusion, 
suppose now that $\ell \in \bE([[\lambda \tx  \ta]\td][[\lambda \tx \tb]\td],\sigma)$.
Then, there exists  $F \subseteq_f
\bE([\lambda \tx  \tb]\td,\sigma)$ such that $ F \vdash \ell  \in \bE([\lambda \tx  \ta]\td,\sigma)$. 
Let $m \in F$. Since  $m \in \bE([\lambda \tx  \tb]\td,\sigma)$, there exists $G_m \subseteq_f \bE(\td,\sigma)$ such that
$ G_m \vdash m \in \bE(\lambda \tx  \tb,\sigma)$.
From this, we obtain
 $m \in \bE(\tb,\lb G_m/\tx\rb\sigma)$.
 Now, as 
$ F \vdash \ell \in  \bE([\lambda \tx  \ta]\td,\sigma)$, there exists
$H \subseteq_f \bE(\td,\sigma)$ such that $ H \vdash  F \vdash \ell  \in \bE(\lambda \tx  \ta,\sigma)$.
From this, we obtain
$F \vdash \ell  \in \bE(  \ta,\lb H/\tx \rb\sigma)$.
Let

\smallskip

{\centering
$L \eqdef \lb n \st n \in G_m$ for some $m \in F \rb \cup H$ \par}

\smallskip

\noindent and note that it is a finite set of formulas. Let $m \in F$.
Since $G_m \subseteq L$
 we obtain
$m \in \bE(\tb,\lb G_m/\tx\rb\sigma) \subseteq \bE(\tb,\lb L/\tx\rb\sigma)$
by Lemma \ref{prop up} and
therefore
$F \subseteq_f \bE(\tb,\lb L/\tx\rb\sigma)$.
As $H \subseteq  L$,
it follows from  Lemma \ref{prop up}
that $F \vdash \ell  \in \bE(  \ta,\lb H/\tx \rb\sigma) \subseteq \bE(  \ta,\lb L/\tx \rb\sigma)$. So, we obtain
$\ell  \in \bE(  \ta\tb,\lb L/\tx \rb\sigma)$
and $L \vdash \ell \in 
\bE(  \lambda \tx[\ta\tb],\sigma)$.
Since
 $G_m \subseteq_f \bE(\td,\sigma)$ for every $m$ and $H \subseteq_f \bE(\td,\sigma)$, it follows that
 $L \subseteq_f \bE(\td,\sigma)$.
As  $L \vdash \ell \in 
\bE(  \lambda \tx[\ta\tb],\sigma)$,
we finally obtain $ \ell \in 
\bE(  [\lambda \tx[\ta\tb]]\td,\sigma)$.

($\beta_4$) We have to show that
$\bE([\lambda \tx  [\lambda \tx  \ta]]\td,\sigma) = \bE(\lambda \tx  \ta,\sigma)$. 
Suppose that $\ell \in \bE([\lambda \tx  [\lambda \tx  \ta]]\td,\sigma)$.
 Then,
 there exists   $F \subseteq_f \bE(\td,\sigma)$ such that  $ F
 \vdash \ell  \in \bE(\lambda \tx [\lambda \tx  \ta],\sigma)$. 
 We also have $\ell \in  \bE(\lambda \tx  \ta, \lb F/\tx \rb \sigma)$.
  If $\ell = \tx$, then we clearly have
 $\tx \in \bE(\lambda \tx  \ta,\sigma)$.
 Otherwise, it follows that $\ell = G \vdash m$
 for some $G$ and $m$ such that 
 $m \in  \bE(  \ta, \lb G/\tx \rb\lb F/\tx \rb \sigma)$.
 By using Lemma \ref{lemma ass}(i), 
  we obtain
 $m \in  \bE(  \ta, \lb G/\tx \rb\lb F/\tx \rb \sigma) = \bE(  \ta, \lb G/\tx \rb \sigma)$
 and hence $G \vdash m  \in \bE(\lambda \tx  \ta,\sigma)$.
As for the opposite inclusion,
 suppose now that $\ell \in \bE(\lambda \tx  \ta,\sigma)$.  
If $\ell = \tx$,
 we also have $\tx \in 
\bE(\lambda \tx \ta,\lb \emptyset /\tx \rb\sigma)$. From this, we obtain $\emptyset \vdash \tx \in 
\bE(\lambda \tx [\lambda \tx \ta],\sigma)$. Since $\emptyset \subseteq_f \cI(\td,\sigma)$ it follows that
$\tx \in \bE([\lambda \tx  [\lambda \tx  \ta]]\td,\sigma)$. 
Otherwise, we have $\ell = F \vdash  m$
for some $F$ and $m$ such that
 $m \in \bE(\ta,\lb F /\tx\rb\sigma)$.
  By using Lemma \ref{lemma ass}(i), we obtain
 $m \in \bE(  \ta, \lb F/\tx \rb \sigma) =  \bE(  \ta, \lb F/\tx \rb\lb \emptyset/\tx \rb \sigma)$ and we get $ \ell 
 \in \bE(\lambda \tx  \ta,\lb \emptyset/\tx \rb\sigma)$ and $ \emptyset \vdash\ell 
 \in \bE(\lambda \tx [\lambda \tx  \ta],\sigma)$.
 Since $\emptyset \subseteq_f \cI(\td,\sigma)$,  we then have
$\ell \in \bE([\lambda \tx  [\lambda \tx  \ta]]\td,\sigma)$.

($\beta_5$)
Suppose that $\tx \neq \ty$ and
 $[\lambda \ty\td]\tx \mo \td$.
We now show that $\bE([\lambda \tx  [\lambda \ty  \ta]]\td,\sigma)  = \bE(\lambda \ty [ [\lambda \tx  \ta]\td],\sigma)$. 
Assume $\ell \in \bE([\lambda \tx  [\lambda \ty  \ta]]\td,\sigma)$. Then,
 there exists  $F \subseteq_f \bE(\td,\sigma)$  such that $   F \vdash \ell \in \bE(\lambda \tx  [\lambda \ty  \ta],\sigma)$.
From this, it follows that $\ell \in \bE(\lambda \ty  \ta, \lb F/\tx \rb \sigma)$.
 If $\ell = \ty$, then we obviously have
$\ty \in \bE(\lambda \ty [ [\lambda \tx  \ta]\td],\sigma)$.
Otherwise, we have
 $\ell = G \vdash m$ 
 for some $G$ and $m$
 such that 
  $m \in \bE(\ta,\lb G/\ty \rb\lb F/\tx \rb \sigma)$.
Since $\tx \neq \ty$, by using Lemma \ref{lemma ass}(ii) we obtain
 $m \in  \bE(\ta,\lb F/\tx \rb \lb G/\ty \rb\sigma)$. So, we get
$F \vdash m \in \bE(\lambda \tx\ta, \lb G/\ty \rb\sigma)$.
  Since $[\lambda \ty\td]\tx \mo \td$,
 we have $\bE(\td,\sigma) = \bE(\td,\lb G/\ty \rb\sigma)$ by Lemma \ref{proptecc}.
 Since  $F \subseteq_f \bE(\td,\sigma)$,
 we have $F \subseteq_f \bE(\td,\lb G/\ty \rb\sigma)$.
 Now, as $F \vdash m \in \bE(\lambda \tx\ta, \lb G/\ty \rb\sigma)$, we obtain
 $m \in  \bE([\lambda \tx\ta]\td, \lb G/\ty \rb\sigma)$. Thus, we conclude that
 $G \vdash m \in  \bE(\lambda \ty[[\lambda \tx\ta]\td],\sigma)$. 
 To show the opposite inclusion, 
assume now
 $\ell  \in
 \bE(\lambda \ty  [ [\lambda \tx  \ta]\td],\sigma)$.
If $\ell = \ty$,
then we also have $\ty \in 
\bE(\lambda \ty \ta,\lb \emptyset /\tx \rb\sigma)$. From this, we obtain $\emptyset \vdash \ty \in 
\bE(\lambda \tx [\lambda \ty \ta],\sigma)$. Since $\emptyset \subseteq_f \cI(\td,\sigma)$, it follows that
$\ty \in \bE([\lambda \tx  [\lambda \ty  \ta]]\td,\sigma)$. 
 Otherwise,  we have  $\ell = F \vdash m$
  for some $F$ and $m$ such that
 $m \in\bE( [\lambda \tx  \ta]\td, \lb F/\ty \rb \sigma)$.
Furthermore, 
 there exists $G \subseteq_f \bE(\td, \lb F/\ty \rb \sigma)$
  such that $G \vdash m \in \bE(\lambda \tx  \ta, \lb F/\ty \rb \sigma)$.
From this, we obtain
$ m \in \bE(\ta, \lb G/\tx \rb\lb F/\ty \rb \sigma)$. Since $\tx \neq \ty$, we can apply Lemma \ref{lemma ass}(ii) and we get
 $m \in  \bE(\ta,\lb F/\ty \rb \lb G/\tx \rb\sigma)$. From this, we obtain
 $F \vdash m \in \bE(\lambda \ty\ta, \lb G/\tx \rb\sigma)$ and
 $ G \vdash F \vdash m \in \bE(\lambda \tx [\lambda \ty\ta], \sigma)$.
 Since $[\lambda \ty\td]\tx \mo \td$,
 we have $\bE(\td,\sigma) = \bE(\td,\lb F/\ty \rb\sigma)$ by Lemma \ref{proptecc}.
 Since  $G \subseteq_f \bE(\td,\lb F/\ty \rb\sigma)$,
 we have $G \subseteq_f \bE(\td,\sigma)$.
 Now, since  $ G \vdash F \vdash m \in \bE(\lambda \tx [\lambda \ty\ta], \sigma)$,
 we finally obtain 
 $F \vdash m \in \bE([\lambda \tx [\lambda \ty\ta]]\td, \sigma)$.

  The proof of the theorem is now complete.
 \qedhere
\end{proof}

We can now 
prove Theorem \ref{teononalpha} as follows.     By Theorem \ref{impo}, the relation 
$\mo$ is a prelambda congruence.
 Now, let $\tx$ and $\ty$ two variables
such that $\tx \neq \ty$.
Let $\sigma$ be an environment.
We have 
 $\tx \in  \cI(\lambda \tx \tx,\sigma)$
   and $\tx \notin \cI(\lambda \ty \ty, \sigma)$, as $\tx \neq \ty$
   and $\ty$ is the only atomic formula
   in $\cI(\lambda \ty \ty, \sigma)$.  This shows that  $\lambda \tx \tx \nocv
\lambda \ty \ty$.

%
%In order to prove
%that  pure beta conversion is consistent,
%another method is avalable.
%In fact, by Lemma \ref{lemmainc1}
%and Theorem \ref{T2} below
%--- these results do not depend on Theorem \ref{impo} ---
%we have $\cv \, \subseteq \, \ce
% \, \subseteq \, \RCv$.
%But $\RCv$ is just standard extensional
%conversion, and
%it is a  classical results
%of the lambda calculus 
%that $\RCv$ is consistent,
%see \eg \citet{BAR} and \citet{HS}.
%Using this argument, the consistency of
%$\cv$ follows.
%
%However, the aim of this section is
% to prove 
%that  $\Pi$
%is not closed under algebraic  $\alpha$--renaming,
%and for this purpose  the alternative method
%outlined above does not work.
%Therefore, we prove our result
%using Theorem \ref{impo} as follows.
%
%

\section{Conclusions and Some Directions for Future Work} \label{concl}

In this work we have obtained
alternative and simplified formulations of the concepts of lambda theory and extensional
theory without introducing the meta--theoretic  notion of substitution and the conceptually inelegant sets of all, free and
bound  variables occurring in a term. 
We have also clarified the actual
role of $\alpha$--renaming and $\eta$--extensionality in the lambda calculus: from a convenient point of view --- our prelambda congruences --- both of them can be equivalently
described as properties of
extensionality for  certain classes
of terms.  

Our proof of the elimination of the ancillary concepts is 
rather technical, but conceptually
speaking very simple. 
We also point out that in the relevant literature  discussed in Section \ref{intro} we could not find \emph{complete} and \emph{detailed}  proofs of
 results similar to those we proved in Section \ref{EQU} and Section \ref{modelsection}. For this reason --- and also because 
 we want  to make the article
readable by a wider audience ---
we decided to make our exposition self--contained as much as possible.

As for future work, we plan to
 apply the  ideas we followed
in  this paper to other contexts.
For instance, we would like to provide similar   formalizations of various notions of    \emph{derivability} (of formulas and sequents) for several first-- and higher--order logics and theories.   As already observed by \RV \cite{RV85}, some  results
 for derivability in
  classical first--order logic with equality
which are in line with our  motivations --- the elimination of some meta--theoretic notions included in  our  ancillary concepts --- have been 
already  established
by  Tarski \cite{TAR}, Kalish and Montague \cite{KM} and  Monk \cite{MONK}.
% (see also Monk's book \cite[p. 171]{MonkBook}).
Similar results for 
 second--order classic logic have been established by  Cocchiarella \cite{Cocchia}.
But in our opinion,  the most satisfactory 
 axiomatization of first--order classical logic with equality, due to N\'{e}meti, 
is reported in the book of Blok and Pigozzi
\cite[App. C]{BlokPig}.  In contrast to the aforementioned work, the formulation described  in that book has the pleasant property of being  completely free of ancillary concepts.

 We believe that there is a more  uniform way to tackle the problem of eliminating the ancillary concepts from several logics,  as we now intuitively explain.
 It is well--known that substitution
 --- which is usually introduced as  a tool  for developing the proof--theory of the quantifiers ---
 can be completely handled by  the lambda notation, as in Church's theory of simple types \cite{church1940}.  
  This theory can be seen as (an extension of) classical higher--order logic and we refer to Coquand \cite[Sec. 1]{C86} for a very elegant presentation of its purely logical part  in an intuitionistic setting.
%  There, the intuitionistic implication $\Rightarrow$ and the universal quantifier $\Pi$ are constants (in the sense of Section \ref{lambda sec})
%  and the notion of sameness adopted is the simply typed version of lambda conversion. Thus, roughly speaking,  ``formulas''
%  are just terms in our sense; for instance,   $\forall \tx (\tx \rightarrow \tx)$ is represented by the
% term $\Pi[\lambda \tx [[\Rightarrow \tx] \tx]]$.
We think that  the ideas
and the approach
employed in the present article
can be exploited  to
remove the ancillary concepts
 ---  which have no mathematical 
 and logical  substance 
 ---  from logic
and hence
 simplify the presentations of intuitionistic and classical higher--order logics.
  
  We also plan to export our ideas to give  new presentations of some frameworks based on the lambda calculus (without types)  whose aim is to formalize (considerable parts of) mathematics such as the \emph{type--free systems}
  introduced by
   Myhill and Flagg \cite{MYFL} and \emph{map theory}, a setting originally introduced by Grue  and recently simplified by Berline and Grue \cite{GRBER}.

Finally, we also believe
that   some ideas we have introduced in this paper can be useful to develop another approach to lambda calculi with \emph{explicit substitutions},  see Abadi, Cardelli, Curien and L\'{e}vi \cite{ACCL} and also Kesner \cite{KES} for a brief survey.   Our suggestion is to internalize  substitutions as  in the theory of lambda substitution algebras, the algebraic framework introduced by Diskin and Beylin \cite{DB}  discussed in Section \ref{PBC}. 
 The expected advantage should be
the following:  no  ancillary concept other than substitution --- which is  not  a meta--theoretic notion in that context --- would  appear in the formalizations of lambda and extensional theories in  systems
with explicit substitutions.
 
 More generally, we believe 
 that the \emph{algebraic} approach
 to binding operations,
 see Cardone and Hindley \cite[p. 736]{CAHI} for a brief survey, has been relatively overlooked by the computer science community --- exceptions are, of course,  lambda abstraction and lambda substitution algebras. While the primary goals of those lines of work,  namely \emph{representation theorems}, are perhaps
 more palatable to algebraists than
 computer scientists, we think that
 the study of the algebraic approach to quantification  can be very useful for the developments of better
  syntactical 
 formalizations of  theories in 
 structures with binding operations.

{\footnotesize \bibliography{basaldellabibi}}

%\printbibliography[title={Bibliography}] 

\end{document}